\documentclass[a4paper,onecolumn,11pt,accepted=2024-08-08]{quantumarticle}
\pdfoutput=1
\usepackage[utf8]{inputenc}
\usepackage[english]{babel}
\usepackage[T1]{fontenc}
\usepackage{amsmath}
\usepackage{hyperref}
\usepackage{tikz,physics,enumitem,dsfont,amsthm,mathtools,amssymb}
\usepackage{tikz}
\usepackage{lipsum}
\usetikzlibrary{automata,arrows,positioning,calc}
\newtheorem{theorem}{Theorem}
\newtheorem{proposition}{Proposition}

\newtheorem{lemma}{Lemma}
\newtheorem{corollary}{Corollary}
\newtheorem*{claim*}{Claim}
\newtheorem*{lemma*}{Lemma}
\theoremstyle{definition}
\newtheorem{definition}{Definition}[section]

\newcommand{\sw}[1]{\textcolor{orange}{}}%{SW: #1}}
\newcommand{\memento}[1]{\textcolor{red}{}}%{\sc{#1}}}
\newcommand{\reb}[1]{\textcolor{black}{#1}}

\begin{document}

\title{Entanglement buffering with two quantum memories}

\author{Bethany Davies}
\thanks{these two authors contributed equally.}
\email{b.j.davies@tudelft.nl}
\orcid{}

\author{\'{A}lvaro G. I\~{n}esta}
\thanks{these two authors contributed equally.}
\email{a.gomezinesta@tudelft.nl}
\orcid{}

\author{Stephanie Wehner}
\orcid{}

\affiliation{QuTech, Delft Univ. of Technology, Lorentzweg 1, 2628 CJ Delft, The Netherlands}
\affiliation{EEMCS, Quantum Computer Science, Delft Univ. of Technology, Mekelweg 4, 2628 CD Delft, The Netherlands}
\affiliation{Kavli Institute of Nanoscience, Delft Univ. of Technology, Lorentzweg 1, 2628 CJ Delft, The Netherlands}

\maketitle

\begin{abstract}
  Quantum networks crucially rely on the availability of high-quality entangled pairs of qubits, known as entangled links, distributed across distant nodes. Maintaining the quality of these links is a challenging task due to the presence of time-dependent noise, also known as decoherence. Entanglement purification protocols offer a solution by converting multiple low-quality entangled states into a smaller number of higher-quality ones. 
In this work, we introduce a framework to analyse the performance of entanglement buffering setups that combine entanglement consumption, decoherence, and entanglement purification.
We propose two key metrics: the availability, which is the steady-state probability that an entangled link is present, and the average consumed fidelity, which quantifies the steady-state quality of consumed links.
We then investigate a two-node system, where each node possesses two quantum memories: one for long-term entanglement storage, and another for entanglement generation. 
We model this setup as a continuous-time stochastic process and derive analytical expressions for the performance metrics.
Our findings unveil a trade-off between the availability and the average consumed fidelity.
We also bound these performance metrics for a buffering system that employs the well-known bilocal Clifford purification protocols.
Importantly, our analysis demonstrates that, in the presence of noise, consistently purifying the buffered entanglement increases the average consumed fidelity, even when some buffered entanglement is discarded due to purification failures.
\end{abstract}

\section{Introduction}
The functionality of quantum network applications typically relies on the consumption of entangled pairs of qubits, also known as \emph{entangled links}, that are shared among distant nodes \cite{Wehner2018}.
The performance of quantum network applications does not only depend on the rate of production of entangled links, but also on their quality.
In a quantum network, it is therefore a priority for high-quality entangled states to be readily available to network users.
This is a challenging task, since entangled links are typically stored in memories that are subjected to time-dependent noise, meaning that the quality of stored entangled links decreases over time. This effect is known as decoherence.

A common way of overcoming the loss in quality of entangled links is to use \emph{entanglement purification} protocols \cite{Bennett1996,Deutsch1996,Dur1999,Yan2023}. An $m$-to-$n$ entanglement purification protocol consumes $m$ entangled quantum states of low quality and outputs $n$ states with a higher quality, where typically $m > n$. The simplest form of purification schemes are $2$-to-$1$, also known as \emph{entanglement pumping} protocols. One downside of using purification is that there is typically a probability of failure, in which case the input entangled links must be discarded and nothing is produced. 

In this work, we take a crucial step towards the design of high-quality entanglement buffering systems.
The goal of the buffer is to make an entangled link available with a high quality, such that it can be consumed at any time for an application.
We develop methods to analyse the performance of an entanglement buffering setup in a system with entanglement consumption, decoherence, and entanglement pumping.
We introduce two metrics to study the performance: (\textit{i}) the \emph{availability}, which is the steady-state probability that a link is available, and (\textit{ii}) the \emph{average consumed fidelity}, which is the steady-state average quality of entangled links upon consumption. We measure the quality of quantum states with the fidelity, which is a well-known metric for this \cite{Nielsen2002}.

We use these metrics to study a two-node system where each of the nodes has two quantum memories, each of which can store a single qubit (see Figure \ref{fig.1G1Billustration}).
This system is of practical relevance since early quantum networks are expected to have a number of memories per node of this order (e.g. in \cite{Kalb2017} and \cite{Yan2022}, entanglement purification was demonstrated experimentally between two distant nodes, each with the capability of storing two qubits).
%in ref. \cite{Pompili2021} the authors realized the first three-node quantum network with NV-center technology, each node with a single storage qubit).
We study a system where each node has one good (long-term) quantum memory, G, and one bad (short-term) memory, B, per node. We therefore refer to this entanglement buffering setup as the \emph{1G1B system}. The good memories are used to store an entangled link between the nodes that can be consumed at any time. The bad memories are used to generate a new entangled link between the nodes. The new link may be used to pump the stored link with fresh entanglement.

\begin{figure}[t]
  \centering
  \includegraphics[width=0.7\linewidth]{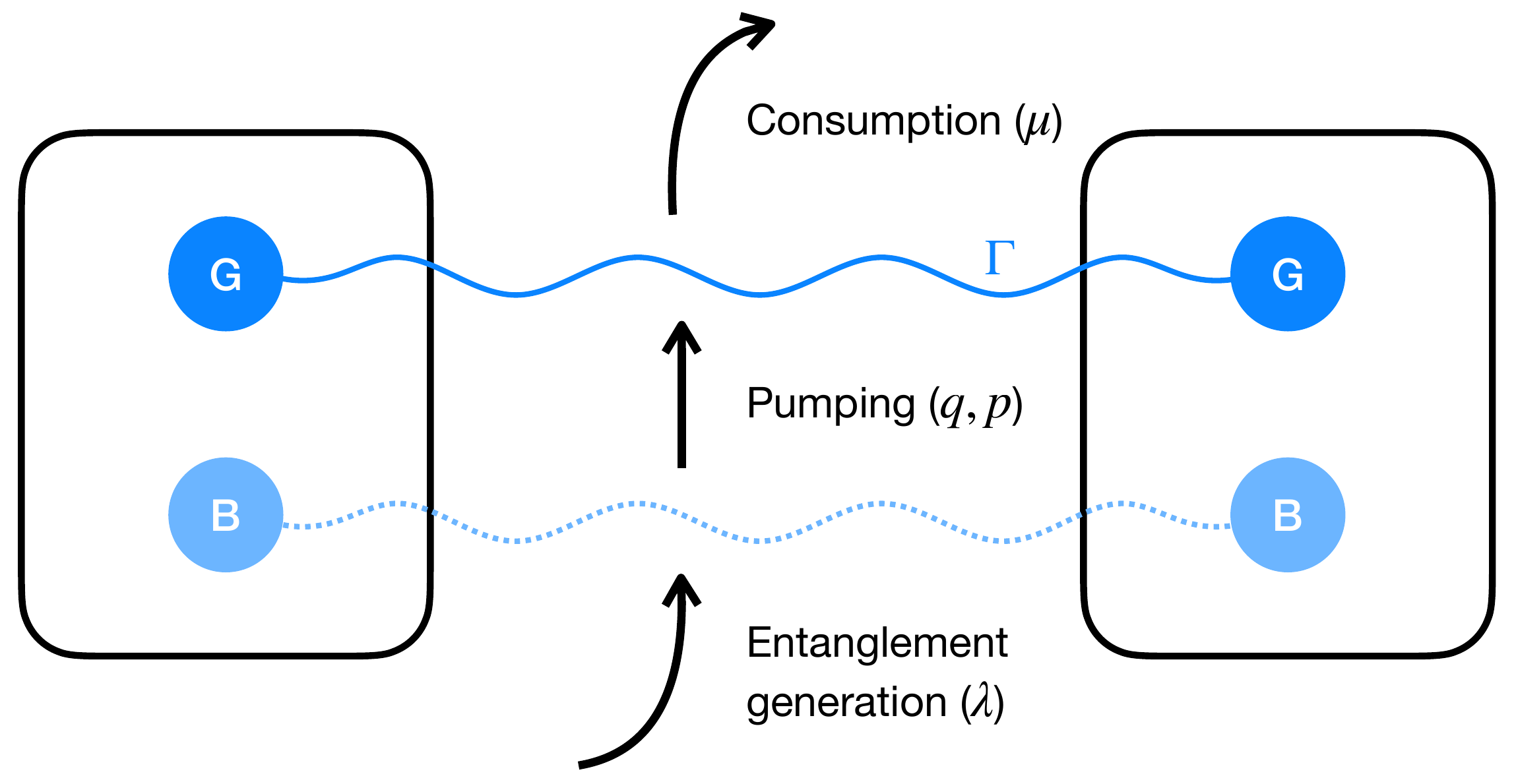}
  \caption{Illustration of the entanglement buffering system with two quantum memories (1G1B system). Each of the nodes has two memories (G and B).
  Memory G is used to store the buffered link.
  An entangled link is generated at a rate $\lambda$ in memory B.
  If memory G is empty when the new link is generated in B, the link is immediately transferred to G.
  If memory G is occupied, the new link generated in B is immediately used to purify the buffered link with probability $q$ (otherwise, the new link is discarded).
  The pumping protocol consumes the link in B to increase the quality of the buffered link in G, and it succeeds with probability $p$ (otherwise, it destroys the link in G).
  The buffered link is consumed at a fixed rate $\mu$.
  The quality of the entanglement stored in G decays exponentially with rate $\Gamma$. Formal definitions of the problem parameters can be found in Section \ref{sec:1G1B_system}.}
  \label{fig.1G1Billustration}
\end{figure}

Calculating the temporal evolution of the fidelity of an entangled link is generally a difficult task, since the fidelity depends on the history of operations that have been applied to the link in the past.
By modelling the state of the 1G1B system as a continuous-time stochastic process, we are able to find analytical solutions for the availability and the average consumed fidelity of the system.
We illustrate the application of these results in a simplified scenario where purification has a linear action on the quality of the buffered link.

Our main contributions are the following:
\begin{itemize}
    \item We propose two metrics to measure the performance of an entanglement buffering system: the availability and the average consumed fidelity.
    \item We provide a simple closed-form expression for the availability in the 1G1B system.
    \item We develop an analytical framework to calculate the average fidelity of the links consumed in a 1G1B system. We provide a closed-form expression for pumping schemes that increase the fidelity of the entangled link linearly with the initial fidelity.
\end{itemize}

Our main findings are the following:
\begin{itemize}
    %\item The average fidelity of the entanglement stored in the buffer is the same as the average fidelity upon consumption. \sw{this is not clear enough (because consumption wasnt discussed clearly enough probably)} \bd{I don't think that this is a main finding - let's discuss}
    \item We confirm the intuition that, except in some edge cases, there is a trade-off between availability and average consumed fidelity: one must either consume low-quality entanglement at a higher rate, or high-quality entanglement at a lower rate.
    \item Consider a situation where bilocal Clifford protocols are employed (this is one of the most popular and well-studied classes of purification protocols \cite{Dehaene2003}). \reb{Then, if the noise experienced by the quantum memories is above certain threshold,} pumping the stored link with fresh entanglement always increases the average consumed fidelity, even if the stored link is often discarded due to a small probability of successful pumping. \reb{We provide an explicit expression for this noise threshold, which depends on the purification protocol employed and the fidelity of newly generated links.}
\end{itemize}

The structure of the paper is the following. In Section \ref{sec.background}, we provide a short overview of related work (an additional introduction to the quantum information concepts required to understand this work is provided in Appendix \ref{app:quantum_preliminaries}).
In Section \ref{sec:1G1B_system}, we explain the physical setup and provide a formal definition of the 1G1B system as a stochastic process. In Section \ref{sec.performance}, we define the performance metrics of interest and provide analytical expressions that enable their computation. In Section \ref{sec:linear_jump}, we analyse the system in the case where the pumping protocol produces an output state whose fidelity is a linear function of the fidelity of one of the input states.
In Section \ref{sec:operating_regimes}, we use these results to bound the performance of the 1G1B system, in the case where bilocal Clifford protocols are employed for entanglement pumping.
Lastly, in Section \ref{sec.conclusions}, we discuss the implications of this work and future research directions.

%\clearpage

\section{Related work}
\label{sec.background}

The performance analysis of quantum networks is unique because of the trade-off between the rate of distribution of entangled links and the quality of distributed links, both of which are important for the functionality of networking applications. This leads to interesting stochastic problems, which are important to understand the parameter regimes of a possible architecture. For example, \cite{Shchukin2019,Li2020,Inesta2023} deal with the problem of generating an end-to-end entangled link across a chain of quantum repeaters, where both the rate of production and the quality of the end-to-end links are quantities of interest. Another example is the problem of generating multiple entangled links between two users with a high quality, which is treated in \cite{Praxmeyer2013,Davies2023}.
In these works, the time between successfully generated entangled links is modelled by a geometric distribution.
However, the time taken up by an entanglement generation attempt is generally small compared to other relevant time scales \cite{Pompili2021,Liu2023}.
Hence, a simplifying assumption that we make in this work is that the time between entanglement generation attempts is exponentially distributed.
This is a common assumption in the quantum networking literature (see e.g. \cite{Vardoyan2023,nain2020,Chandra2022}), because it can enable the finding of closed-form relations between physical variables and protocol parameters.
Here, we introduce and find expressions for the values of two key performance metrics in the steady state.

%\agi{@Bethany: What do you think about this paragraph? I'm not convinced by the last sentence.}
%The advantage of this simplification is that it is more tractable to find solutions for the relevant performance metrics in a closed form.
%Nevertheless, this is still a valid assumption to study real physical setups, where the time between attempts follows a geometric distribution but it is often small compared to other relevant time scales \cite{Pompili2021,Liu2023} as we discuss in Section \ref{sec:1G1B_system}.
Previous work that incorporates entanglement purification schemes into the analysis of quantum network architectures typically involves numerical optimisation methods (see e.g. \cite{victora2023}), or only considers specific purification protocols \cite{Bratzik2013}.
%Rozpedek2018a,Krastanov2019,victora2023
By contrast, in this work we focus on presenting the purification protocol in a general way, and finding closed-form solutions for the performance metrics of interest (albeit for a simpler architecture). This is an important step towards an in-depth understanding of how one can expect purification to impact the performance of a near-term quantum network.

Other works have introduced the concept of entanglement buffering (preparing quantum links to be consumed at a later time) over a large-scale quantum network \cite{Inesta2023a,Pouryousef2023}.
To the best of our knowledge, the only work with a similar set-up to ours is \cite{Elsayed2023}, which was developed in parallel and independently of our work. There, the authors study the steady-state fidelities of a system involving two memories used for storage (good memories), and one memory used for generation (bad memory).
This work differs from ours in multiple ways.
%Firstly, a different pumping scheme is considered: in particular, pumping takes place between the links stored in memory. In our system, pumping always takes place between the freshly generated link and the link in memory. %Additionally, our analysis applies to any pumping protocol while they assume a specific one.
%For example, they assume that the fidelity after purification cannot increase beyond the initial fidelity of freshly generated links, which is an assumption we do not make.
For example, the analysis is done in discrete time and it is assumed that the fidelity takes a discrete set of values, whereas we do not make this assumption since we work in continuous time.
Additionally, consumption of entanglement is not included in the system studied in \cite{Elsayed2023}, which may impact the steady-state behaviour.

Lastly, we note that previous work generally assumes a specific protocol for entanglement buffering between each pair of nodes, and does not address the following fundamental question: \emph{what is the best way to buffer entanglement between two users in a quantum network?} To the best of our knowledge, we address this question for the first time.

%This could be envisaged, for example, as an overlay link in a large-scale network.
%\sw{DANGER: rejection right here as its not clear why what we do is new}
%\agi{@Bethany: please have a look at this paragraph. What do you think? Also, could you please add a couple of sentences about the methods we use? (i.e., how the methods are new/interesting)}
%\bd{Done-what do you think?}

\section{The 1G1B (one good, one bad) system}
\label{sec:1G1B_system}
We now define the 1G1B system. In Section \ref{sec:system_description}, we describe and motivate the model of the system. In Section \ref{sec:system_definition}, we define the variables of interest precisely. This facilitates the definition of the performance metrics in Section \ref{sec.performance}.
In Appendix \ref{app:quantum_preliminaries}, we provide an introduction to the quantum background required.

%\sw{after reading seciton 2 entirely I think being a bit more precise and formal (and potentiallyu leaving quantum arguments for appendix and give a comment/intuitiion here)} \agi{I do not understand this comment.}
\subsection{System description}
\label{sec:system_description}
Below we provide a list of assumptions that model the 1G1B system, and provide motivation for each assumptions. An illustration of the system is given in Figure \ref{fig.1G1Billustration}.
%\agi{I think the new list (with maths details in boldface and motivation in normal font) is genius! It is super clear now!}

\begin{enumerate}
    \item 
        \textbf{Each of the nodes has two memories: one long-term memory (good, G) and one short-term memory (bad, B). The B memories are used to generate new entangled links. The G memories are used as long-term storage (entanglement buffer).}
    
    This is motivated by the fact that \emph{storage} (G) and \emph{communication} (B) qubits are often present in experimental scenarios, where the former is used to store entanglement and the latter is used to generate new links \cite{Benjamin2006,Kalb2017,Lee2022}.
    \item \textbf{New entangled links are generated in memory B according to a Poisson process with rate $\lambda$. New entangled links always have the form $\rho_\mathrm{new}$.}
    
    Physical entanglement generation attempts are typically probabilistic and heralded \cite{Bernien2013,Barrett2005}. \reb{In other words,} the attempt can fail with some probability and, when this occurs, a failure flag is raised. Therefore, the generation of a single link may take multiple attempts.
    The time taken by an attempt is typically fixed (this is both the case in present-day quantum networks \cite{Pompili2021} and an assumption that is commonly made in the theoretical analysis of quantum networks \cite{Inesta2023, Inesta2023a, Davies2023}). Then, the time between attempts follows a geometric distribution.
    Since the probability of successful generation and the length of the time step is often small compared to other relevant time scales \cite{Pompili2021,Liu2023}, \reb{we use a continuous approximation, i.e. that the time between arrivals are exponentially distributed. This is a Poisson process (see e.g. Chapter 6.8 from \cite{grimmett2020}).}
    %\sw{can you say how small? and why is this statement at all true if the nodes are futher apart?} \agi{I think that would be too much info for a SIGMETRICS reader, no @Bethany?}
    \item 
    \textbf{When a link is newly generated in memory B, if memory G is empty (no link present), the new link is immediately placed there. If memory G is not empty, the nodes immediately either (\textit{i}) attempt pumping with probability $q$, or (\textit{ii}) discard the new link from memory B (probability $1-q$).}

    This step is included because it may not always be a good idea to carry out pumping, due to there being a  possibility of this failing. 
    \item \textbf{Links stored in memory G are Werner states.}

    Werner states take the simple form 
    \begin{equation*}
\rho_\mathrm{\scriptscriptstyle W} = F \ketbra{\phi^+} + \frac{1-F}{3} \ketbra{\psi^+} + \frac{1-F}{3} \ketbra{\psi^-} + \frac{1-F}{3} \ketbra{\phi^-},
\end{equation*}
where $\left\{ \ket{\phi^+},\ket{\psi^+},\ket{\phi^-},\ket{\psi^-} \right\}$ denote the Bell basis.
A Werner state corresponds to maximally entangled state that has been subjected to isotropic noise. The state in the good memory is therefore fully described by one parameter: the fidelity $F$ to the target state $\ket{\phi^{+}}$. Any state can be transformed into a Werner state with the same fidelity by applying extra noise, a process known as \emph{twirling} \cite{Bennett1996a,Horodecki1999}. Hence, this assumption constitutes a worst-case model.
    \item \textbf{While in memory G, states are subject to depolarising noise with memory lifetime~$1/\Gamma$.}

    Depolarising noise can also be seen as a worst-case noise model \cite{Dur2005}.
    After a time $t$ in memory, this maps the state fidelity $F$ to $$F \rightarrow e^{-\Gamma t}\left(F-\frac{1}{4}\right) + \frac{1}{4}. $$
    \item \textbf{Consumption requests arrive according to a Poisson process with rate $\mu$. When a consumption request arrives, if there is a stored link in memory G, it is immediately used for an application (and therefore removed from the memory). If there is no link available, the request is ignored.}

    \reb{This means that the time until the next consumption request arrives is independent of the arrival time of previous requests, and it is exponentially distributed. This assumption is commonly made in the performance analysis of queuing systems (see e.g. Chapter 14 from \cite{VanMiegham2014}).}
    \item \textbf{Assumptions about pumping:}
    \begin{enumerate}
        \item \textbf{Pumping is carried out instantaneously.}

        This is because the execution time is generally much lower than the other timescales involved in the problem.
        For example, in state-of-the-art setups, an entangled link is generated approximately every 0.5 s \cite{Pompili2021}, while entanglement pumping may take around $0.5\cdot10^{-3}$ s \cite{Kalb2017}.
        If the nodes are far apart, classical communication between them would only add a negligible contribution to the purification protocol (e.g. classical information takes less than $10^{-4}$ s to travel over 10 km of optical fiber). 
        
        \item  \textbf{Suppose that the link in memory G has fidelity $F$ and the link in memory B is in state $\rho_{\mathrm{new}}$. If pumping succeeds, the output link has fidelity $J(F,\rho_{\mathrm{new}})$, and remains in the good memory. If pumping fails, all links are discarded from the system. } \reb{ Here, the \textit{jump function} $J(F,\rho_{\mathrm{new}})\in [0,1]$ is dependent on the choice of purification protocol. Given the assumption that one of the links is a Werner state, the form of this function is 
        \begin{equation}
            J(F,\rho_{\mathrm{new}}) = \frac{\tilde{a}(\rho_{\mathrm{new}}) F + \tilde{b}(\rho_{\mathrm{new}})}{p(F,\rho_{\mathrm{new}})},
            \label{eqn:jump_function_rational}
        \end{equation}
        with 
        \begin{equation}
            p(F,\rho_{\mathrm{new}}) = c(\rho_{\mathrm{new}}) F + d(\rho_{\mathrm{new}})
            \label{eqn:probability_linear_general}
        \end{equation}
        where $\tilde{a},\tilde{b},c,d$ are functions of $\rho_{\mathrm{new}}$. Here, $p(F,\rho_{\mathrm{new}})$ is the success probability of purification. See Appendix  \ref{app:jump_function_general_form} for an explanation of why the jump function and success probability take this form. } 
        \item  \textbf{Pumping succeeds with probability $p$, \reb{which is constant in the fidelity of memory G}.}
        \reb{We see from the above that this is a special case, and that in general the probability of purification success varies linearly with the fidelity of the good memory. However, performing the analysis with a constant probability of success does allow us to find bounds on the operating regimes of the system by considering the best-case and worst-case values of $p$ (see Section~\ref{sec:operating_regimes}). Combining this with Assumption 7b, we see that this is effectively equivalent to setting $c(\rho_{\mathrm{new}})=0$. The jump function is then linear in the fidelity of memory G, and can be written as
        \begin{equation*}
            J(F,\rho_{\mathrm{new}}) = a(\rho_{\mathrm{new}}) F + b(\rho_{\mathrm{new}}),
        \end{equation*}
        where $a\coloneqq \tilde{a}/p$ and $b\coloneqq \tilde{b}/p$.} 
    \end{enumerate}
\end{enumerate}
Implicit in the above is that the process of entanglement generation, pumping and consumption ((2),(3),(6) and (7b)) are independent. We provide a summary of the parameters involved in the 1G1B system in Table \ref{tab.variables}.

%%%%%%%%%%%%%
%%%%%%%%%%%%%
\renewcommand{\arraystretch}{1.2}
\begin{table}[t]
    \centering
    \caption{Parameters of the 1G1B system. See main text for detailed explanations.
    }\label{tab.variables}
    \vspace{-2mm} % Adjust the height of the space between caption and tabular
\begin{tabular}{p{0.1\textwidth}p{0.73\textwidth}}
\multicolumn{2}{c}{\textbf{Hardware}}\\[1pt]
\hline
	$\lambda$ & Rate of heralded entanglement generation (time between successful attempts is exponentially distributed with rate $\lambda$) \\
    $\rho_\mathrm{new}$ & Entangled state produced after a successful entanglement generation\\
	$\Gamma$ & Rate of decoherence (fidelity of the entangled link decays exponentially over time with rate $\Gamma$) \\[5pt]
\multicolumn{2}{c}{\textbf{Application}}\\[1pt]
\hline
    $\mu$ & Rate of consumption (specified by application) \\[5pt]
\multicolumn{2}{c}{\textbf{Pumping protocol}}\\[1pt]
\hline
    $q$ & Probability of attempting pumping immediately after a successful entanglement generation attempt (otherwise the new link is discarded)\\
    $p$ & Probability of successful pumping \\
    $J(F,\rho_\mathrm{new})$ & Jump function: fidelity of the output state following successful pumping ($F$ is the fidelity of the Werner state stored in the good memory)
\end{tabular}

\end{table}
%%%%%%%%%%%%%
%%%%%%%%%%%%%

\subsection{System definition}
\label{sec:system_definition}
In this subsection, we define the state of the system mathematically, which will be the main object of study in the rest of this work. We view the state of the system as the number of rounds of pumping that the link in memory has undergone. From now on, when we refer to 1G1B, we refer to the stochastic process that evolves according to the following definition.
\begin{definition}[1G1B system]
    Let $s(t)$ be the state of the 1G1B system at time $t$. This takes values 
\begin{equation}
    s(t) = \begin{cases}
        \emptyset \text{ \small if there is no link in memory,} \\
        i \geq 0 \; \text{ \small if there is a link in memory which is the result of $i$ successful pumping rounds,}
    \end{cases}
\end{equation}
where $i=0$ corresponds to a link in memory that has not undergone any pumping. Assume that the system starts with no link, i.e. $s(0)=\emptyset$. The system transitions from state $\emptyset$ to state $0$ when a new link is generated and placed in the good memory, which was previously empty. The rate of transition from $\emptyset$ to $0$ is then given by the entanglement generation rate $\lambda$. Pumping success occurs when a new link is produced (rate $\lambda$), pumping is attempted (probability $q$), and pumping succeeds (probability $p$). Therefore, the transition from state $i$ to $i+1$ occurs with rate $\lambda q p$. The final allowed transition is from $i$ to $\emptyset$ which occurs due to consumption or purification failure, which occurs with rate $\mu+\lambda q (1-p)$.
\label{def:1G1B_system}
\end{definition}

We also refer to the state $i\geq 0$ as the $i$th \textit{purification level}. Since the transitions between each state in 1G1B occur according to an exponential distribution with rate that is only dependent on the current state of the system, this is a continuous-time Markov chain (CTMC) on the state space $\{\emptyset, 0,1,... \}$. The resulting CTMC and the rate of transitions is depicted in Figure \ref{fig.ctmc_1g1b}. This is the main object of study in our work. 
\begin{figure}[h]
\centering
	\begin{tikzpicture}[->, >=stealth', auto, semithick, node distance=1.8cm]
	\tikzstyle{every state}=[fill=white,draw=black,thick,text=black,scale=1]
	\node[state]    (empty)               {$\emptyset$};
	\node[state]    (zero)[right of=empty]   {$0$};
	\node[state]    (one)[right of=zero]   {$1$};
	\node[state]    (two)[right of=one]   {$2$};
        \node          (dots)[right of=two]   {...};
	\path (empty) edge [bend left] node [above] {$\lambda$} (zero); 
        \path (zero) edge [bend left] node [above] {$\lambda qp$} (one);
        \path (one) edge [bend left] node [above] {$\lambda qp$} (two);
        \path (two) edge [bend left] node [above] {$\lambda qp$} (dots);
        \path (zero) edge [bend left]  (empty);
        \path (one) edge [bend left]  (empty);
        \path (two) edge [bend left] (empty);
        \path (dots) edge [bend left] node [below] {$\mu+q(1-p)\lambda$} (empty);
	\end{tikzpicture}

\caption{The transitions of the 1G1B system.}
\label{fig.ctmc_1g1b}
\end{figure}
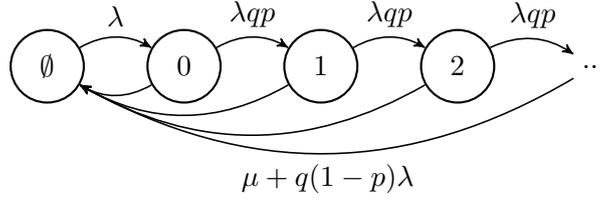

Recall that we are also interested in the fidelity of the link in memory. This is dependent not only on the state $s(t)\in \{\emptyset,0,1,... \}$, but also on the time spent in the states leading up to the current purification level. This motivates the following definition.
\begin{definition} Suppose that $s(t)=i$. Then, we define random variable $\vec{T}(t)$ to be the length-$(i+1)$ vector storing the times spent in the recent purification levels $0,1,\dots,i$ leading up to the current one, where time $T_j(t)$ was spent in the most recent visit to state $j$ ($j\leq i-1$), and time $T_i(t)$ is the time spent so far in state $i$. See Figure \ref{fig:fidelity_vs_time} for a depiction of this.
\label{def:T(t)}
\end{definition} 
We also need a framework with which to compute the fidelity at time $t$. Recalling assumption (5) of Section \ref{sec:system_description}, we denote decoherence by the following.
\begin{definition}
    Let $D_t : [0,1]\rightarrow [0,1]$ denote the action of depolarising noise on the state fidelity $F$. This has action $$ D_t[F]= e^{-\Gamma t}\left(F-\frac{1}{4}\right) +\frac{1}{4}. $$
    \label{def:depol_noise}
\end{definition}
We now formally define the jump function. 
 \begin{definition}
     After successfully applying purification to a Werner state with fidelity $F$ and a general two-qubit state $\rho_{\mathrm{new}}$, the output state has fidelity $J(F,\rho_{\mathrm{new}})$. We refer to $J$ as the \textit{jump function} of the protocol. \reb{The general form of this is given in (\ref{eqn:jump_function_rational}).}
 \end{definition}
We note that every purification protocol has a corresponding jump function. The exact form of $J$ is dependent on the choice of pumping protocol, but in general is a continuous rational function of $F$, taking values in $[0,1]$.

We also need to compute the fidelity after many rounds of decoherence and pumping. This essentially means composing $D_t$ and $J$.
\begin{definition}
Let $F^{(i)}(t_0,...,t_i)$ denote the fidelity after spending time $t_0,...,t_i$ in each purification level $0,1,...,i$. This may be defined recursively as
\begin{equation}
    F^{(i)}(t_0,...,t_i) = D_{t_i}\left[ J(F^{(i-1)}(t_0,\dots,t_{i-1}), \rho_\mathrm{new}) \right],
\label{eqn:composite_fidelity_fn}
\end{equation}
with $F^{(0)}(t_0) = D_{t_0}[F_{\mathrm{new}}]$, where $F_{\mathrm{new}}$ is the fidelity of $\rho_{\mathrm{new}}$. 
\label{def:composite_fidelity_fn}
\end{definition}
Note that $F^{(i)}$ is a continuous and bounded function of its inputs, since the same is true for $D_t$ and $J$. We are now equipped to define the fidelity of the system.
\begin{definition}
    The fidelity of the 1G1B system at $t$ is given by 
    \begin{equation}
        F(t) = \begin{cases}
            F^{(i)}(\vec{T}(t)) \text{    if } s(t) = i \geq 0, \\
            0, \text{   if } s(t) = \emptyset.
        \end{cases}
    \end{equation}
\label{def:F(t)}
\end{definition}

\begin{figure}
    \centering
    \includegraphics[width=0.7\textwidth]{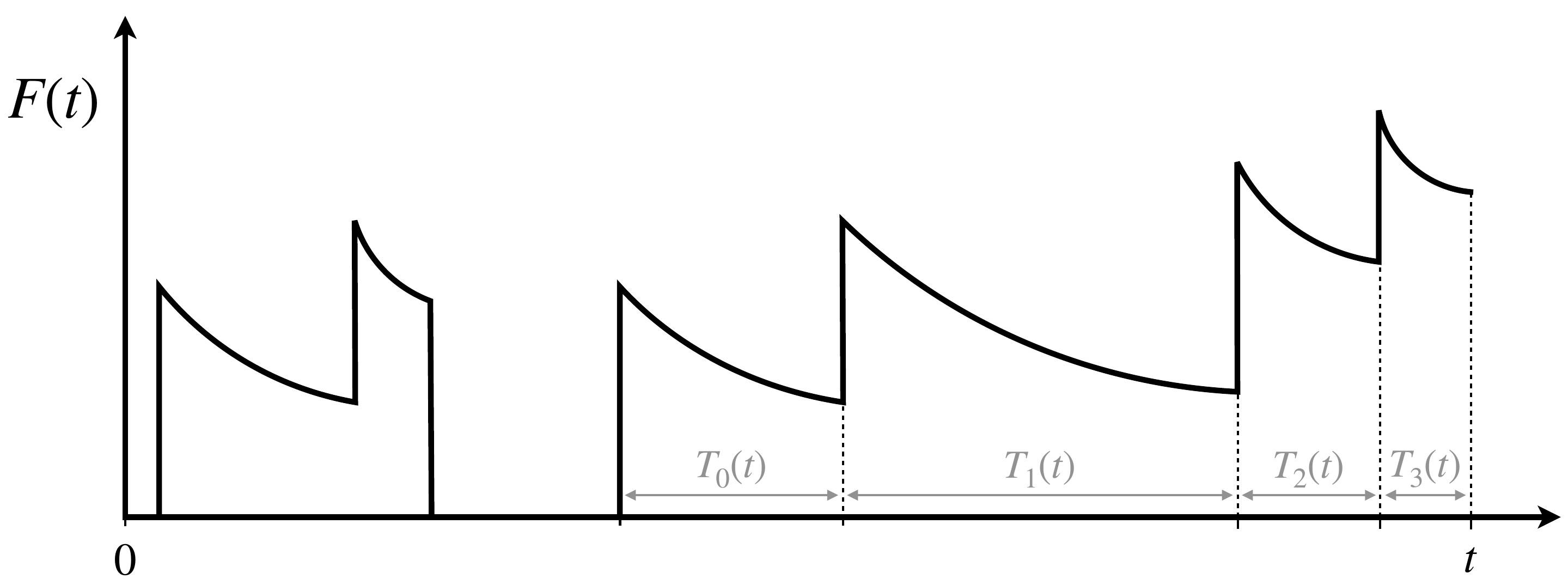}
    \caption{Example of the evolution of the fidelity of the buffered entanglement over time. The fidelity experiences a sudden boost every time a pumping protocol is successful. Then, it decays exponentially due to decoherence.
    Each state in the CTMC is identified by the number of times the current buffered link has been purified.
    If $s(t)=i$, the random variables $\{ T_j(t): j=0,1,...,i-1 \}$ are the times spent in each state of the CTMC immediately leading up to the current state, and $X_i(t)$ is the time so far spent in state $i$.}
    \label{fig:fidelity_vs_time}
\end{figure}
Note that this formulation can also be adapted to incorporate a system where we apply a different pumping protocol in each state of the CTMC. In that case, we would employ a more general recurrence relation:
\begin{equation}
    F^{(i)}(t_0,\dots,t_{i}) = D_{t_i}\left[ J^{(i)}(F^{(i-1)}(t_0,\dots,t_{i-1}), \rho_\mathrm{new}) \right],
\end{equation}
where the $J^{(i)}$ is the jump function corresponding to the pumping protocol applied in state $i$ of the CTMC. For simplicity, however, we study recurrence relations of the form (\ref{eqn:composite_fidelity_fn}). This may be used to model the situation where the same pumping protocol is applied every time, or provide bounds for using multiple protocols, as we do in Section \ref{sec:operating_regimes}.

\section{Performance metrics}\label{sec.performance}
In this Section, we define two metrics to evaluate the performance of an entanglement buffering system: the \emph{availability} and the \emph{average consumed fidelity}. We also provide analytical expressions for both metrics in the 1G1B system.

\subsection{Availability}
A natural measure for the quality of service provided to users is the probability that a consumption request may be served at any given time.
If there is a link stored in the good memories, the consumption request is immediately served. However, if there is no entanglement available, the request is ignored. Letting $P(s(t)=i)$ be the probability that the system is in state $i$ at time $t$, we define the steady-state distribution as
\begin{equation}\label{eq.steadystate_general}
    \pi_i := \lim_{t\rightarrow \infty} P(s(t)=i).
\end{equation}
Then, we define our first performance metric as follows.
\begin{definition}[Availability] The availability $A$ is defined as 
\begin{equation}\label{eq.availability_general}
    A \coloneqq 1-\pi_{\emptyset},
\end{equation}
which is the probability that there is a link in memory in the limit $t\rightarrow \infty$.
\end{definition}

This definition can be applied to any entanglement buffering setup. In the 1G1B system, the availability is well-defined, as shown in Appendix \ref{app.performance_metrics}.
Moreover, it is possible to derive a closed-form expression for the availability, as stated in the proposition below.
\begin{proposition} Consider the 1G1B system (Definition \ref{def:1G1B_system}). The \textit{availability} is given by
\begin{equation}
    A = 1-\pi_{\emptyset} = \frac{\lambda}{\lambda + \mu + \lambda q(1-p)},
    \label{eqn:availability_1g1b}
\end{equation}
and the rest of the steady-state distribution is given by 
\begin{equation}
    \pi_i = \frac{\lambda^{i+1} q^i p^i}{(\mu + \lambda q)^{i+1}} \pi_{\emptyset}.
\end{equation}
\label{prop:1G1B_stationary_distribution}
\end{proposition}
See Appendix \ref{app.performance_metrics} for a proof of this proposition. We note that this can be derived in a straightforward manner using the balance equations for a CTMC. Instead, we use renewal theory, for two reasons. Firstly, this approach ties in neatly with the proof of the formula for the average fidelity (see the next subsection). Secondly, this approach provides a formula for the availability that is more general, as it also applies to the case where entanglement generation is described by a general random variable instead of being exponentially distributed. See Appendix \ref{app.performance_metrics} for the general formula for the availability.
\subsection{Average consumed fidelity}
The quality of service of an entanglement buffering system can also be measured in terms of the quality of the entanglement provided to the users.
Therefore, the average fidelity of the entangled links upon consumption can be used as an additional metric to assess the performance of the system.

\begin{definition}[Average consumed fidelity]\label{def.avg_cons_fid}
The \textit{average consumed fidelity} is the average fidelity of the entangled link upon consumption, in the steady state. More specifically,
\begin{equation}\label{def:average_cons_fidelity}
    \overline F \coloneqq \lim_{t\rightarrow \infty}\mathbb{E}\big[\, F(t) \,|\, s(t) \neq \emptyset \, \big].
\end{equation}
%(this would result in a transformation from $\emptyset$ to itself and does not affect the dynamics of the system).
\end{definition}
In the definition of $\overline{F}$, we condition on not being in $\emptyset$ since consumption events do not happen when there is no link present.
As before, this performance metric can be applied to any entanglement buffering setup. In the case of the 1G1B system, it is possible to derive an analytical expression for $\overline F$ which explicitly depends on the steady-state distribution. The formula is given in the following theorem.

\begin{theorem} In the 1G1B system, the average consumed fidelity can be written as  
\begin{gather}
    \overline F = \frac{1}{A}\sum_{i=0}^\infty c_i \pi_i,
    \label{eqn:acf_formula}
\end{gather}
where $\pi_i = \lim_{t\rightarrow \infty} P(s(t)=i)$, and
\begin{equation}
    c_i = \mathbb{E} \left[F^{(i)}\!\left(Q_0,Q_1,...,Q_i \right) \right]
%    = \lim_{t\rightarrow \infty}\int_{0}^{t} \dd t_i f_{\alpha}(t_i) \: ... \int_0^{t} \dd t_0 f_{\alpha}(t_0) F^{(i)}(t_0,...,t_{i-1},t_i),
\label{eqn:formula_conditional_average_fidelity}
\end{equation}
where $A$ is the availability, $Q_0,Q_1,\dots,Q_i$ are i.i.d. random variables with $Q_0\sim \text{Exp}(\mu+\lambda q)$, and $F^{(i)}$ is given in Definition \ref{def:composite_fidelity_fn}.
%\agi{What about rephrasing to "where $Q_0,Q_1,\dots,Q_i$ are i.i.d. random variables that are exponentially distributed with rate $\mu+\lambda q$?"}
%the  $\alpha \coloneqq \mu + \lambda q$; and $f_{\alpha}(t)\coloneqq \alpha e^{- \alpha t}$ is the density function of an exponential random variable with mean $1/\alpha$.
\label{thm:average_fidelity_formula}
\end{theorem}
\begin{proof}[Sketch proof of Theorem \ref{thm:average_fidelity_formula}] 
\reb{A first step is to expand by conditioning on the value of $s(t)$,
\begin{align*}
    \mathbb{E}[F(t)|s(t)\neq \emptyset] &= \sum_{i=0}^{\infty} \mathbb{E}[F(t)|s(t)=i ] P\!\left(s(t)=i|s(t)\neq \emptyset \right) \\ &= \frac{1}{P\!\left(s(t)\neq \emptyset \right)}\sum_{i=0}^{\infty} \mathbb{E}[F(t)|s(t)=i ] P\!\left(s(t)=i \right).
\end{align*}
In Proposition \ref{prop:convergence_sum} (Appendix \ref{app.num_pur_subsec}), we show that, when $t\rightarrow \infty$, the limit can be brought inside of the sum, and so
\begin{align*}
    \overline{F} &= \lim_{t\rightarrow \infty} \mathbb{E}[F(t)|s(t)\neq \emptyset] \\ &= \frac{1}{A}\sum_{i=0}^{\infty} \pi_i\cdot \lim_{t\rightarrow \infty} \mathbb{E}\big[ F(t)|s(t)=i \big],
\end{align*}
where we have used the definition of the steady-state distribution and the availability (see (\ref{eq.steadystate_general})  and (\ref{eq.availability_general})).
The values $\pi_i$ may be computed using Proposition \ref{prop:1G1B_stationary_distribution}. The remaining work is then to show that 
\begin{equation}
    \lim_{t\rightarrow \infty}\mathbb{E}\big[ F(t)|s(t)=i \big] = \mathbb{E}\left[F^{(i)}\left(Q_0,\dots,Q_i\right)\right], \label{eqn:limit_of_F(i)_sketch}
\end{equation}
which essentially requires the characterisation of the limiting distribution of $\vec{T}(t)$, since from Definition \ref{def:F(t)} we recall that $\mathbb{E}\big[ F(t)|s(t)=i \big] = \mathbb{E}\left[F^{(i)}\left(\vec{T}(t)\right)\big|s(t)=i\right].$ This is achieved with the following result: conditional on $s(t)=i$, $\vec{T}(t)\rightarrow (Q_0,\dots, Q_i)$ in distribution as $t\rightarrow \infty$, where the $Q_j$ are i.i.d. random variables with $Q_0\sim \mathrm{Exp}(\mu+\lambda q)$. There are two main steps to show this (see Figure \ref{fig.sketch_proof_timeline} for graphical intuition):}

\begin{figure}[t]
  \centering
  \includegraphics[width=0.9\linewidth]{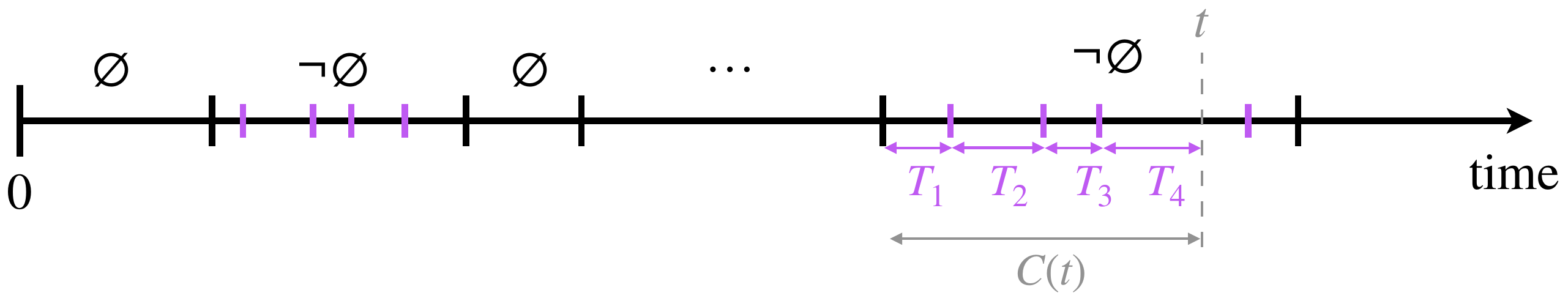}
  \caption{\reb{An example timeline of the 1G1B process. Black dashes are link generation and removal. Shorter purple dashes are pumping rounds. If there is a link present at time $t$, the random variable $C(t)$ is the total time spent so far in $\neg \emptyset$ (link present). Pumping rounds occur within the time $C(t)$ as a Poisson process with rate $\lambda q p$. This may be used to characterise the distribution of $\vec{T}(t)$ in the limit $t\rightarrow \infty$, which is needed to prove Theorem \ref{thm:average_fidelity_formula}.}}
  \label{fig.sketch_proof_timeline}
\end{figure}
\reb{
\begin{enumerate}
    \item Let $C(t)$ be the total time spent so far in $\neg \emptyset$ (link in memory G) at the time $t$. The first step is to show that $C(t)\rightarrow C$ in distribution as $t\rightarrow \infty$, where $C\sim \text{Exp}(\mu+\lambda q(1-p))$. This is shown with renewal theory. For more details, see the results of Appendix \ref{app:subsec_simplified_1G1B}.
    \item Characterise the limiting distribution of the time spent in each purification level \textit{within} the time $C(t)$. These are the $T_j(t)$. We use the fact that pumping rounds occur as a Poisson process within the time $C(t)$. For more details, see the results of Appendix \ref{app.num_pur_subsec}.
\end{enumerate}
Finally, since $F^{(i)}$ is a continuous function of its inputs, (\ref{eqn:limit_of_F(i)_sketch}) follows.}
\end{proof}

For the full proof, see Appendix \ref{app.performance_metrics}. The particularly simple form of (\ref{eqn:formula_conditional_average_fidelity}) can be attributed to the fact that in a CTMC, %the amount of time spent in a particular state is exponentially distributed with rate equal to the sum of the rates of all transitions that leave the state.
the time spent in a state is not influenced by the state to which the system transitions. As an example, in the CTMC from Figure \ref{fig.ctmc_transition_example}, the time spent in state B before a transition does not depend on the transition itself, and this time is exponentially distributed with rate $r_{\scriptscriptstyle BA}+r_{\scriptscriptstyle BC}$.
In the 1G1B system, the times spent in the states $j =0,1,\dots,i-1$ leading up to state $i$ are all exponentially distributed with rate $\lambda q p + \mu + \lambda q (1-p) = \mu + \lambda q$.
As a consequence, the average fidelity after $i$ successful purifications, $c_i$, does not depend on the probability of successful purification $p$.
%More specifically, even though we condition on $s(t)=i$, this does not alter the distribution of times spent in the states leading up to it.

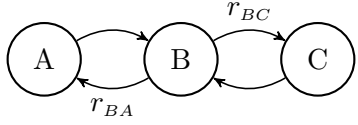
\begin{figure}[h]
\centering
	\begin{tikzpicture}[->, >=stealth', auto, semithick, node distance=1.8cm]
	\tikzstyle{every state}=[fill=white,draw=black,thick,text=black,scale=1]
	\node[state]    (A)               {A};
	\node[state]    (B)[right of=A]   {B};
	\node[state]    (C)[right of=B]   {C};
	\path (A) edge [bend left] node [above] {} (B); 
    \path (B) edge [bend left] node [above] {$r_{\scriptscriptstyle BC}$} (C);
    \path (B) edge [bend left] node [below] {$r_{\scriptscriptstyle BA}$} (A); 
    \path (C) edge [bend left] node [below] {} (B);
    \end{tikzpicture}
\caption{In a CTMC, the time spent in a state is independent of the transition that happens next. In this example, the time spent in state B before leaving is exponentially distributed with rate $r_{\scriptscriptstyle BA}+r_{\scriptscriptstyle BC}$.}
\label{fig.ctmc_transition_example}
\end{figure}

%Equation (\ref{eqn:acf_formula}) is essentially an average of the fidelities at each state of the CTMC, weighted by the probability of being in that state.
%A proof of Theorem \ref{thm:average_fidelity_formula} is given in Appendix \ref{app.performance_metrics}.

%Note that the coefficients $c_i$ correspond to the average fidelity of the buffered link, conditional on having gone through $i$ successful purifications.
%This can be seen by rewriting $c_i$ as an expectation: letting $X_0,X_1,...,X_i$ be i.i.d. exponentially distributed random variables with mean $\alpha$, one can show that
%\begin{equation}
 %   c_i = \mathbb{E} \left[F^{(i)}\left(X_0,X_1,...,X_i \right) \right].
%\label{eqn:conditional_average_fidelity_as_expectation}
%\end{equation}

%The fidelity of the buffered link in state $i$ from the CTMC depends on the amount of time spent in the states $j=0,1,\dots,i-1$ leading up to it.
%Therefore, the average fidelity after $i$ successful purifications, $c_i$, depends on the distributions of the times spent in the states leading up to state $i$.
%In a CTMC, the amount of time spent in a particular state is exponentially distributed with rate equal to the sum of the rates of all transitions that leave the state.

Having systematic closed-form expressions for the functions $F^{(i)}$ enables the efficient computation of $c_i$ and, therefore, $\overline F$.
The calculation of $F^{(i)}$ in closed-form for a general $J$ is quite involved, since \reb{the recurrence relation} (\ref{eqn:composite_fidelity_fn}) becomes a rational difference equation with arbitrary coefficients.
%We have been unable to obtain such an expression for a jump function containing a fraction, such as that given in (\ref{eqn:jump_dejmps}).
However, in the following sections we consider a jump function which is linear, for which it is possible to find a closed-form solution for $\overline F$. 

%\clearpage

\section{Entanglement buffering with a linear jump function}
\label{sec:linear_jump}
In a purification protocol with a linear jump function, the output fidelity is a linear function of the fidelity of one of the input entangled links. \reb{When the probability of successful purification is constant with the fidelity of the good memory, as we assume in 1G1B, this implies that the jump function is linear. This is shown in Appendix \ref{app:jump_function_general_form}.}
In this Section, we compute a closed-form solution for the average consumed fidelity in a 1G1B system assuming a linear jump function. Then, we analyse the performance of the system using the performance metrics defined in Section \ref{sec.performance} (availability and average consumed fidelity). In Section \ref{sec:operating_regimes}, we focus on bilocal Clifford protocols, an important type of purification scheme. 
For a given value of target availability, we provide upper and lower bounds on the average consumed fidelity that can be achieved by any bilocal Clifford protocol in the 1G1B system.

Purification protocols with linear jump functions are relevant for two main reasons:
\begin{enumerate}[label=(\roman*)]
    \item Purification protocols are generally more effective within some range of input fidelities (the increase in fidelity is larger when the input fidelities are within some interval). If the system operates within a small range of fidelities, one may approximate the true jump function with a linear jump function.% and obtain a closed-form solution for the average consumed fidelity.
    \item One can find linear jump functions that upper and lower bound a set of jump functions of interest. These may then be used to upper and lower bound a fidelity-based performance metric (such as the average consumed fidelity) of a system that has the freedom to employ any of these jump functions. 
\end{enumerate}
In Appendix \ref{app.linear_bounds}, we demonstrate (ii) in the case where bilocal Clifford protocols are employed in the 1G1B system. The output fidelity of a bilocal Clifford protocol can be upper and lower bounded by nontrivial linear functions when one of the input states is a Werner state (using some additional minor assumptions).
%As we explain in Section \ref{sec:operating_regimes}, bilocal Clifford protocols are the most relevant type of purification protocol.

Consider a pumping scheme that takes as input a Werner state with fidelity $F$ and an arbitrary state $\rho_\mathrm{new}$. In the 1G1B system, these are the states in the good and the bad memories, respectively.
A linear jump function can be written as
\begin{equation}\label{eq.linear_jump_function}
    J(F,\rho_\mathrm{new}) = a(\rho_\mathrm{new}) F + b(\rho_\mathrm{new}),
\end{equation}
with $0\leq a(\rho_\mathrm{new}) \leq 1$ and $(1-a(\rho_\mathrm{new}))/4\leq b(\rho_\mathrm{new}) \leq 1-a(\rho_\mathrm{new})$, as shown in Proposition~\ref{prop:range_coeffs_linear_jump}.
In what follows, we implicitly assume that $a$ and $b$ depend on $\rho_\mathrm{new}$.

We now derive a closed-form solution for the average consumed fidelity of 1G1B when the jump function is linear, using Theorem \ref{thm:average_fidelity_formula}. The formula requires knowledge of the steady state distribution $\{ \pi_i: i = \emptyset,0,1,... \}$, and the expected fidelities $c_i$, as defined in (\ref{eqn:formula_conditional_average_fidelity}).
Recall that we assume a constant $p$, and therefore the steady-state distribution is independent of the jump function. Hence, we can use the formula for $\pi_i$ from Proposition \ref{prop:1G1B_stationary_distribution}. The work then lies in computing the $c_i$, which are dependent on the choice of jump function, \reb{recalling their definition in (\ref{eqn:formula_conditional_average_fidelity}). From the same equation, we see that} the first step to compute $c_i$ is to find an explicit solution for the function $F^{(i)}$.
The linear jump function (\ref{eq.linear_jump_function}) allows us to do this by solving the recurrence relation (\ref{eqn:composite_fidelity_fn}). The explicit form of $F^{(i)}$ is provided in the following proposition (see Appendix \ref{app:average_fidelity_derivation_linear_jump} for a proof).
\begin{proposition}
    Consider a 1G1B system with $J(F,\rho_\mathrm{new}) = aF+b$ and $F^{(0)}(t_0)=D_{t_0}(F_\mathrm{new})$, where $F_\mathrm{new}$ is the fidelity of the state $\rho_\mathrm{new}$. Then,
    \begin{equation}
        F^{(i)}(t_0,...,t_{i-1},t_i) = \frac{1}{4} + \sum_{j=0}^i m_j^{(i)} e^{-\Gamma(t_j+t_{j+1}...+t_i)}
        \label{eqn:fidelity_fn_linear}
    \end{equation}
    where the constants $m^{(i)}_j$ are given by $m_0^{(0)}=F_\mathrm{new}-\frac{1}{4}$, and 
    \begin{equation}
    m_j^{(i)} = 
    \begin{cases}
    a^{i-j}
    \left( \frac{a}{4} + b - \frac{1}{4} \right), \text{ if } j>0, \\ 
    a^i \left(F_\mathrm{new} - \frac{1}{4}\right)\text{ if } j=0.
    \end{cases} 
    \end{equation}
    for $i>0$.
\label{prop:formula_fidelity_functions_linear_sym_jump}
\end{proposition}

In the following Lemma, we use the formula for $F^{(i)}$ \reb{(found in Proposition 2) and combine this with Theorem \ref{thm:average_fidelity_formula}} to derive a closed-form expression for $c_i$, \reb{and therefore} for the average consumed fidelity.

\begin{lemma}
Consider a 1G1B system with $J(F,\rho_\mathrm{new})=aF+b$ and $F^{(0)}(t_0)=D_{t_0}(F_\mathrm{new})$, where $F_\mathrm{new}$ is the fidelity of the state $\rho_\mathrm{new}$. Then, the average fidelity after $i\geq 0$ purification rounds is given by 
\begin{equation}
c_i = \frac{1}{4} +  \left(F_\mathrm{new}-\frac{1}{4}\right)\cdot a^i  \gamma^{i+1} + \left( \frac{a}{4} + b - \frac{1}{4} \right) \gamma  \frac{1-a^i \gamma^i}{1-a\gamma},
\label{eqn:conditional_average_fidelity_linear_jump}
\end{equation}
where $\alpha = \mu + \lambda q$ and $\gamma = \alpha/(\alpha+ \Gamma)$. Moreover, the average consumed fidelity is given by
\begin{equation}
    \overline F_{\scriptscriptstyle \mathrm{linear}}  = \frac{\frac{1}{4}\Gamma + b \lambda q p + F_\mathrm{new} \Big(\mu + \lambda q (1 - p)\Big)}{\Gamma + \mu + \lambda q (1-pa)}.
    \label{eqn:formula_average_fidelity_linear_jump}
\end{equation}
\label{lem:formula_average_fidelity_linear_jump}
\end{lemma}
The closed-form solution (\ref{eqn:formula_average_fidelity_linear_jump}) is obtainable since $\overline{F}=\frac{1}{A}\sum_{i=0}^{\infty} \pi_i c_i$ is a geometric series with the linear jump function, \reb{as can be seen from the form of $\pi_i$ and $c_i$ as found in Proposition \ref{prop:1G1B_stationary_distribution} and Equation \ref{eqn:conditional_average_fidelity_linear_jump}}. In the following proposition, we see how $\overline{F}$ varies with $p$ and $q$.
\begin{proposition}
    The quantity $\overline F_{\scriptscriptstyle \mathrm{linear}}$ has the following properties:
    \begin{enumerate}[label=(\alph*)]
        \item $\overline F_{\scriptscriptstyle \mathrm{linear}}$ is a monotonic function of $q$;
        \item $\overline F_{\scriptscriptstyle \mathrm{linear}}$ is a monotonic function of $p$;
    \end{enumerate}
\label{lem:properties_average_fidelity_linear_jump}
\end{proposition}

We provide a proof of Lemma \ref{lem:formula_average_fidelity_linear_jump} and Proposition \ref{lem:properties_average_fidelity_linear_jump} in Appendix \ref{app:average_fidelity_derivation_linear_jump}. We now have closed-form expression for $A$ and $\overline F_{\scriptscriptstyle \mathrm{linear}}$, which allows for a thorough analysis of the performance of the 1G1B system with the linear jump function. In particular, the following conclusions may already be drawn.
\begin{itemize}
    \item Result ($a$) from Proposition \ref{lem:properties_average_fidelity_linear_jump} implies that the average consumed fidelity is maximized for $q=0$ or $q=1$. Consider a 1G1B system with a fixed set of parameters and a pumping scheme with a linear jump function.
If the pumping protocol is good enough (e.g. when $b\geq F_\mathrm{new}(1-a)$, as explained in Appendix \ref{app:average_fidelity_derivation_linear_jump}), then pumping every time a link is generated ($q=1$) maximises the average consumed fidelity. Sometimes, the pumping protocol chosen may impact the average consumed fidelity negatively and in that case one should never pump entanglement ($q=0$) to increase the average consumed fidelity.
\item Result ($b$) from Proposition \ref{lem:properties_average_fidelity_linear_jump} provides similar insights: a pumping protocol with a good jump function always benefits from a larger probability of success, i.e. $\overline F_{\scriptscriptstyle \mathrm{linear}}$ is maximized for $p=1$. When the protocol is detrimental, failure ($p=0$) benefits the overall procedure, since it frees the good memory and allows for a fresh entangled link to be allocated there. 
\end{itemize}

%\agi{I think we cannot show that $\overline F_{\scriptscriptstyle \mathrm{linear}}$ is always increasing or decreasing (neither in $q$ or $p$), as I found that the sign of the partial derivative depends on the specific combination of parameters.} \bd{Can we find regimes of $a$, $b$ for which it is increasing or decreasing?} \agi{We can, but those regimes also depend on the other parameters, so it is not very illustrative and I think it is not worth it.}

When the jump function is good (i.e. when $\overline F_{\scriptscriptstyle \mathrm{linear}}$ is monotonically increasing in $q$), we observe a trade-off between $\overline F_{\scriptscriptstyle \mathrm{linear}}$ and the availability $A$, which is a decreasing function of $q$, as can be seen from (\ref{eqn:availability_1g1b}).
This behaviour is shown in Figure \ref{fig.linear_jump}.
If we rarely purify (small $q$), a low-quality entangled state (small $\overline F_{\scriptscriptstyle \mathrm{linear}}$) will be available most of the time (large $A$).
In that case, the average consumed fidelity can be lower than the fidelity of newly generated links, since the entanglement is not being purified often enough to compensate the noise introduced by the memory over time (in Figure \ref{fig.linear_jump}, the average consumed fidelity is below the dashed line for small $q$).
When purification is performed more often (larger $q$), the quality of the stored entanglement will be higher (larger $\overline F_{\scriptscriptstyle \mathrm{linear}}$), at the expense of a more limited availability (smaller $A$), since purification can fail and destroy the entanglement stored in the long-term memory.
This trade-off disappears when the pumping scheme is deterministic ($p=1$): the availability remains constant when varying $q$ since purification will always succeed and the stored entanglement will not be destroyed.
Note that, if the system is dominated by decoherence ($\Gamma \gg \lambda, \mu$), the average consumed fidelity will always be smaller than~$F_0$.

As a validation check, we also implemented a Monte Carlo simulation of the 1G1B system, which provided the same availability and average consumed fidelity that we obtained analytically (our code is available at \href{https://github.com/AlvaroGI/buffering-1G1B}{https://github.com/AlvaroGI/buffering-1G1B}).

\begin{figure}[t]
  \centering
  \includegraphics{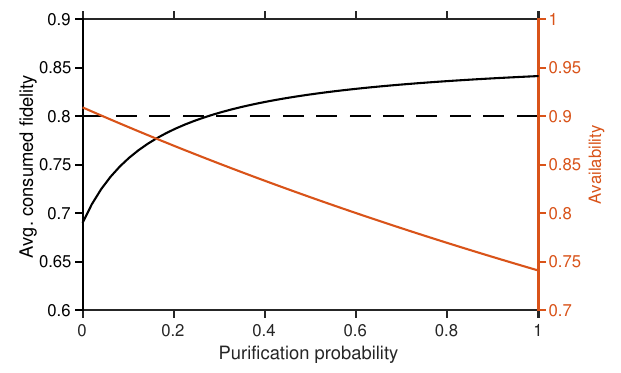}
  \caption{Trade-off between average consumed fidelity and availability. When the pumping is good enough (see discussion in main text), the average consumed fidelity $\overline F$ (black line) increases with increasing purification probability $q$, while the availability $A$ (orange line) decreases. The dashed line corresponds to the fidelity of newly generated links ($F_\mathrm{new} = 0.8$). Other parameters used in this example (times and rates in the same arbitrary units): $\lambda = 1$, $\mu = 0.1$, $p=0.75$, $\Gamma = 1/40$, $J(F,\rho_\mathrm{new}) = (1/3)F + (1+F_\mathrm{new})/3$,
  $\rho_\mathrm{new} = F_\mathrm{new} \ketbra{\phi^+} + (1-F_\mathrm{new}) \left(\ketbra{\psi^+} +  \ketbra{\psi^-} + \ketbra{\phi^-}\right)/3$.
  This jump function corresponds to a linear approximation of a specific bilocal Clifford protocol (the DEJMPS protocol) in the high-fidelity regime \cite{Deutsch1996}.}
  \label{fig.linear_jump}
  %\Description{Figure description.}
\end{figure}

\subsection{Operating regimes of bilocal Clifford protocols}
\label{sec:operating_regimes}
In this Subsection, we study the operating regimes of the 1G1B system, under the assumption that the pumping protocol employed is a bilocal Clifford protocol \cite{Dehaene2003, Jansen2022}.
Firstly, we find upper and lower bounds for the availability.
Then, for a desired value of the availability within these bounds, we find lower and upper bounds for the average consumed fidelity that can be provided by bilocal Clifford protocols. This analysis finds limits to the performance of the 1G1B buffering system.

    Bilocal Clifford protocols are one of the most well-studied types of protocol \cite{Dehaene2003, Jansen2022,Goodenough2023}. One of their main advantages is that they are relatively simple to execute, since they involve a basic set of gates. 
    To the best of our knowledge, bilocal Clifford circuits have been the only purification protocols implemented experimentally so far (see, e.g. \cite{Kalb2017,Yan2022}).
    In Appendix \ref{app.linear_bounds} we provide further details on bilocal Clifford protocols.

Let us start our performance analysis by discussing the availability.
The maximum value that can be achieved by any protocol (bilocal Clifford or not) is $\lambda / (\lambda + \mu)$, as can be seen from (\ref{eqn:availability_1g1b}).
This maximum value is obtained when there is no pumping or the pumping protocol succeeds deterministically, i.e. when $q=0$ or $p=1$.
The availability is lower bounded by $\lambda/(2\lambda + \mu )$, and the minimum value is attained when a pumping protocol is always applied and it never succeeds, i.e., when $q=1$ and $p=0$.

To find bounds for the average consumed fidelity, we first need to bound the jump functions of all bilocal Clifford protocols, which we do in the following Lemma.
We only consider nontrivial protocols, i.e. we do not consider bilocal Clifford protocols with $J(F, \rho_\mathrm{new}) = F$ or $J(F, \rho_\mathrm{new}) = F_\mathrm{new}$, where $F_\mathrm{new}$ is the fidelity of $\rho_\mathrm{new}$.
The former trivial jump function corresponds to a protocol that leaves the buffered link untouched, while the second trivial jump function corresponds to a protocol that replaces the buffered link by the newly generated link.

\begin{lemma}
    Let $J(F, \rho_\mathrm{new})$ be the jump function of a nontrivial bilocal Clifford protocol ($J(F, \rho_\mathrm{new})\neq F$ and $J(F, \rho_\mathrm{new})\neq F_\mathrm{new}$, where $F_\mathrm{new}$ is the fidelity of $\rho_\mathrm{new}$). Assume $\rho_\mathrm{new}$ is a Bell-diagonal state:
    \begin{align}
        \rho_\mathrm{new} &= F_\mathrm{new} \ketbra{\Phi^+} + \lambda_1 \ketbra{\Psi^+} +\lambda_2 \ketbra{\Psi^-} +\lambda_3 \ketbra{\Phi^-},
    \end{align}
    with $F_\mathrm{new}+\lambda_1+\lambda_2+\lambda_3 = 1$. Let us define $F^*$ as
    \begin{equation}
        F^* = \frac{2F_{\mathrm{new}}-1+\sqrt{(2F_{\mathrm{new}}-1)^2 - 2 \lambda_\mathrm{min}(1-2F_{\mathrm{new}} - 2 \lambda_\mathrm{min})}}{2(2F_{\mathrm{new}}-1 + 2 \lambda_\mathrm{min})},
        \label{eqn:def_F_star_maintext}
    \end{equation}
    where $\lambda_\mathrm{min} = \min(\lambda_1,\lambda_2,\lambda_3)$.
    Then, for all $F\in [\frac{1}{4},F^*]$, the jump function is lower bounded as follows:
\begin{equation}\label{eq.newlinearboundsF}
    a_\mathrm{\scriptscriptstyle l} F + b_\mathrm{\scriptscriptstyle l} \leq J(F, \rho_\mathrm{new}) 
\end{equation}
where 
\begin{align}\label{eq.lower_bounds_ab}
    a_\mathrm{\scriptscriptstyle l} &= \frac{2(4F^*-1)\left[2F_\mathrm{new} - (F_\mathrm{new}+\lambda_{\min})(F_\mathrm{new}+\lambda_{\max})\right]+4(\lambda_{\max} - \lambda_{\min})(1-F^*)}{(4F^*-1)\left[(4F_\mathrm{new}+4\lambda_{\max}-2)F^*+2-F_{\mathrm{new}}-\lambda_{\max} \right]}, \text{ and }\nonumber \\ 
    \;\;\;  b_\mathrm{\scriptscriptstyle l} &= \frac{F_{\mathrm{new}}+\lambda_{\max}}{2}-\frac{a_\mathrm{\scriptscriptstyle l}}{4},
\end{align}
with $\lambda_{\max} = \max\{\lambda_1,\lambda_2,\lambda_3 \} $. For $F\in [1/4,1]$, the jump function is upper bounded as
\begin{equation}\label{eq.newlinearboundsF_upp}
   J(F, \rho_\mathrm{new}) \leq a_\mathrm{\scriptscriptstyle u} F + b_\mathrm{\scriptscriptstyle u},
\end{equation}
with
\begin{equation}\label{eq.upper_bounds_ab}
    a_\mathrm{\scriptscriptstyle u} = \frac{4(1-F_{\mathrm{new}})}{3}, \text{ and } \;\; b_\mathrm{\scriptscriptstyle u} = \frac{4F_{\mathrm{new}}-1}{3}.
\end{equation}
Moreover, the success probability of the protocol is bounded by $p_\mathrm{\scriptscriptstyle l} \leq p \leq p_\mathrm{\scriptscriptstyle u}$, where
\begin{equation}\label{eq.linearboundsp}
    p_\mathrm{\scriptscriptstyle l} = \frac{1}{2}, \text{ and } \;\; p_\mathrm{\scriptscriptstyle u} = F_\mathrm{new} + \max(\lambda_1,\lambda_2,\lambda_3).
\end{equation}
\label{lem:bilocal_clifford_bounds}
\end{lemma}
A proof of Lemma \ref{lem:bilocal_clifford_bounds} can be found in Appendix \ref{app.linear_bounds}. \reb{We show this by considering propoerties of the jump functions of bilocal Clifford protocols, which may be found explicitly.} Note that, despite the fact that we assume that newly generated entangled links are Bell-diagonal, other forms of the density matrix are also valid in practice, since they can be brought to Bell-diagonal form by adding extra noise \cite{Bennett1996a,Horodecki1999}. 
Note also that the lower bound for the jump function (\ref{eq.newlinearboundsF}) only applies when the fidelity of the buffered link is below $F^*$, but this is always the case in the 1G1B system, as shown in Appendix~\ref{app.linear_bounds}.

If we regard $F_\mathrm{new}$ as a fixed parameter, the upper and lower bounds to the jump function (\ref{eq.newlinearboundsF}) and (\ref{eq.newlinearboundsF_upp}) are linear in $F$, and the bounds to the success probability (\ref{eq.linearboundsp}) are constant.
It is now possible to find an upper and lower bound for the average consumed fidelity by combining Lemmas \ref{lem:formula_average_fidelity_linear_jump} and \ref{lem:bilocal_clifford_bounds}, as we do in the following corollary.

\begin{corollary}\label{corollary.bounds_avgconsfid}
    The average consumed fidelity of the 1G1B system when using any (nontrivial) bilocal Clifford protocol is lower bounded by
    \begin{equation}\label{eq.lowerbound_avgconsfid}
        \overline F_{\scriptscriptstyle \mathrm{l}} = \frac{\frac{1}{4}\Gamma + b_{\scriptscriptstyle \mathrm{l}} \lambda q p + F_\mathrm{new} \Big(\mu + \lambda q (1 - p_{\scriptscriptstyle \mathrm{l}})\Big)}{\Gamma + \mu + \lambda q (1-p_{\scriptscriptstyle \mathrm{l}}a_{\scriptscriptstyle \mathrm{l}})},
    \end{equation}
    and upper bounded by
    \begin{equation}\label{eq.upperbound_avgconsfid}
        \overline F_{\scriptscriptstyle \mathrm{u}} = \frac{\frac{1}{4}\Gamma + b_{\scriptscriptstyle \mathrm{u}} \lambda q p + F_\mathrm{new} \Big(\mu + \lambda q (1 - p_{\scriptscriptstyle \mathrm{u}})\Big)}{\Gamma + \mu + \lambda q (1-p_{\scriptscriptstyle \mathrm{u}}a_{\scriptscriptstyle \mathrm{u}})}
    \end{equation}
    where $a_{\scriptscriptstyle \mathrm{l}}$, $b_{\scriptscriptstyle \mathrm{l}}$, $p_{\scriptscriptstyle \mathrm{l}}$, $a_{\scriptscriptstyle \mathrm{u}}$, $b_{\scriptscriptstyle \mathrm{u}}$, and $p_{\scriptscriptstyle \mathrm{u}}$, are given by (\ref{eq.lower_bounds_ab}), (\ref{eq.upper_bounds_ab}), and (\ref{eq.linearboundsp}).
\end{corollary}

Now, we analyse the limits of the performance of the 1G1B system using the bounds on $\overline F$ from Corollary \ref{corollary.bounds_avgconsfid}.
%\begin{itemize}[label=-]
 %   \item \textbf{Replacement} --- Every time a new link is generated in the bad memory, the link in the good memory is replaced by the new one, without any form of purification. We use this protocol as a baseline, although it is not bilocal Clifford because success is always declared (in bilocal Clifford circuits, success depends on some measurement outcomes \cite{Dehaene2003}).
    %\item \textbf{DEJMPS} --- Common protocol for entanglement pumping \cite{Deutsch1996}. It is optimal in terms of output fidelity when the inputs are two identical Bell-diagonal states \cite{Jansen2022}. In a high-fidelity regime, as we show in Appendix \ref{app:linear_approximation_dejmps} the jump function of this protocol can be approximated as
    %$J_{\scriptscriptstyle \mathrm{DEJMPS}}^\mathrm{linear}(F_1, F_2) = (1+F_1+F_2)/3$. We refer to a protocol with this jump function as \emph{linearized DEJMPS}. \bd{Does this still work with non-Werner states?}
  %  \item \textbf{Upper/lower bound} --- Hypothetical protocol with a jump function saturating the upper/lower bound (\ref{eq.newlinearboundsF}) and with success probability saturating the upper/lower bound (\ref{eq.linearboundsp}). Note that these protocols may not exist in practice, but they bound the performance of every bilocal Clifford protocol.
%\end{itemize}
Let us start with a 1G1B system with perfect memories, i.e. with $\Gamma = 0$.
This corresponds to an ideal situation that we can use as a benchmark: once we introduce noise, the average consumed fidelity will be lower than in this ideal case.
Figure \ref{fig.operating_regimes}(a) shows the achievable combinations of average consumed fidelity and availability for $F_\mathrm{new}=0.8$, generation rate $\lambda = 1$, and consumption rate $\mu = 0.1$. Below, we list some important observations that may be drawn from this Figure:
\begin{itemize}
    \item The regions shaded in grey correspond to \textbf{unattainable values} of average fidelity and availability, and they apply to any pumping scheme (bilocal Clifford or not). The average consumed fidelity cannot be larger than the one provided by a hypothetical protocol with jump function $J(F, \rho_\mathrm{new})=1$ and probability of success $p=1$, which is applied with probability $q=1$ (however, such a protocol does not exist).
    \item The performance of a 1G1B system that uses any \textbf{bilocal Clifford protocol} is contained within the \textbf{region shaded in blue and yellow}.
    The yellow/blue line corresponds to a hypothetical protocol with jump function and success probability saturating the lower/upper bounds from (\ref{eq.newlinearboundsF}) and (\ref{eq.linearboundsp}).
    For a fixed target availability, the blue line provides an upper bound on the maximum average consumed fidelity that can be achieved by using bilocal Clifford protocols. Here, we observe again the tradeoff between both performance metrics: if our target availability is very close to the maximum value, we cannot increase the average consumed fidelity beyond $F_\mathrm{new}$ (dotted line); but as we decrease the desired availability, we can achieve a higher consumed fidelity until we reach a maximum value.
    \item As a reference, we show the performance of the \textbf{\emph{replacement protocol}} (red star): in such a protocol, every time a new link is generated in the bad memory, the link in the good memory is replaced by the new one, without any form of purification. The replacement protocol is not bilocal Clifford because success is always declared (in bilocal Clifford circuits, success depends on some measurement outcomes \cite{Dehaene2003}). This simple protocol achieves maximum availability, given by $A = \lambda / (\lambda + \mu)$. However, since no purification is performed, this protocol cannot increase the fidelity above the initial value $F_\mathrm{new}$. In the absence of decoherence, the replacement protocol is equivalent to applying no purification at all ($q=0$).
\end{itemize}

%The availability is minimized for $q=1$, and is lower bounded by $A \geq \lambda / (2\lambda - p\lambda + \mu )$ (in the example from Fig. \ref{fig.operating_regimes_G0}, $A \geq 0.645$ for any jump function).
In Figure \ref{fig.operating_regimes}(b), we perform a similar analysis for a 1G1B system in which the good memory has a finite lifetime, i.e. $\Gamma > 0$. This is a more realistic scenario. The following observations may be drawn from this Figure:
\begin{itemize}
    \item \textbf{Imperfect memories decrease the average consumed fidelity} but do not affect the availability. The availability is unaffected by the decoherence experienced by the entangled links, and therefore can take the same range of values as in Figure \ref{fig.operating_regimes}(a).
    \item The replacement protocol no longer provides an average fidelity $F_\mathrm{new}$. Instead, the average fidelity is lower than $F_\mathrm{new}$ since the quality of the state stored in the good memory decreases over time and is never increased beyond $F_\mathrm{new}$ due to the absence of purification. However, the \textbf{replacement protocol performs better than no pumping} at all ($q=0$). This is because the system can improve its fidelity every time a new link is produced, instead of waiting for a consumption event.
    \item In the presence of noise, the lower and upper bounds for bilocal Clifford protocols also shift towards lower values of average fidelity. Both the upper and lower bounds take their minimum value at $q=0$.
    This means that, in the presence of noise, any pumping protocol will increase the average consumed fidelity, i.e. \textbf{any pumping \reb{($q>0$)} is better than no pumping \reb{($q=0$)}}, even if it succeeds with the lowest-possible probability. 
    This is in contrast to when there is no noise (Figure \ref{fig.operating_regimes}(a)), where the lower bound takes its minimum at $q=1$ and no such conclusion can be drawn.
    In fact, this conclusion (any pumping is better than no pumping) always applies when the amount of noise, $\Gamma$, \reb{is above the following threshold:
    \begin{equation}\label{eq.noise_threshold}
        \Gamma > 4 \mu p \frac{F_\mathrm{new}(1-a)-b}{4F_\mathrm{new}(1-p) + (4b+a)p - 1},
    \end{equation}
    where $a$, $b$, and $p$ are given by the choice of purification protocol (see (\ref{eq.linear_jump_function})).
    In Appendix \ref{app.noise_threshold}} we compute this threshold analytically.
\end{itemize}

\begin{figure}[h]
  \centering
  \includegraphics[width=0.9\textwidth]{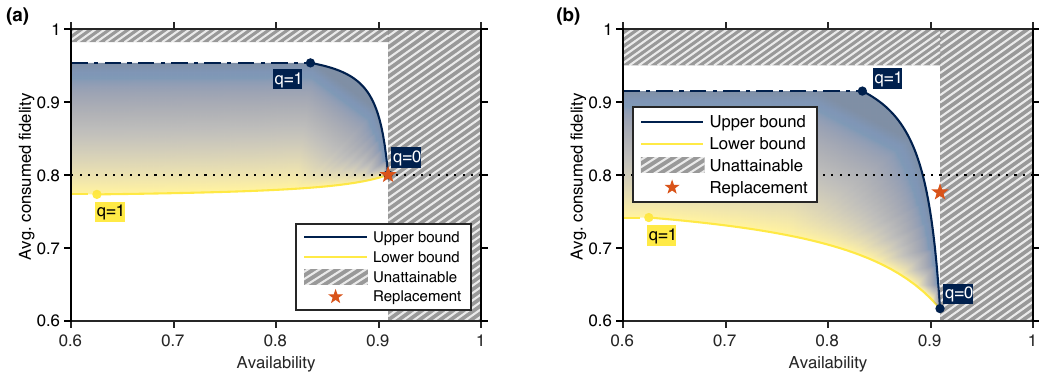}
  \caption{Noise in the memories decreases the average consumed fidelity but does not affect the availability.
  Bounds on the performance of a 1G1B system with a bilocal Clifford protocol,
  and with \textbf{(a)} noiseless memories ($\Gamma = 0$) or \textbf{(b)} noisy memories ($\Gamma = 5\cdot10^{-2}$ a.u.).
    For a given target availability, the average consumed fidelity is within the blue/yellow region (see Corollary \ref{corollary.bounds_avgconsfid}).
Availability is maximized for $q=0$ ($q$ is the probability of purification after successful entanglement generation), and it decreases for increasing $q$.
  White regions cannot be achieved by bilocal Clifford protocols.
  Striped regions cannot be achieved by any pumping protocol.
  Red star: performance of the replacement protocol (buffered link is replaced by new links).
  Dotted line: fidelity of newly generated entangled links.
    Parameters used in this example (times and rates in the same arbitrary units): $\lambda = 1$, $\mu = 0.1$, $F_\mathrm{new}=0.8$, $\rho_\mathrm{new} = F_\mathrm{new} \ketbra{\phi^+} + (1-F_\mathrm{new}) \left(\ketbra{\psi^+} +  \ketbra{\psi^-}\right)/2$.
  %\memento{Plot from RESULTS\_FvsA\_linearbounds.m.}
  }
  \label{fig.operating_regimes}
\end{figure}

%\clearpage

\section{Conclusions and outlook}\label{sec.conclusions}
Our work sheds light on how to buffer high-quality entanglement shared among remote nodes in a quantum network.
We have proposed two metrics to measure the performance of an entanglement buffering system: the availability and the average consumed fidelity.
The availability corresponds to the fraction of time in which entanglement is available for consumption. The average consumed fidelity measures the quality of the entanglement upon consumption. We have used these performance metrics to analyse the 1G1B system, an entanglement buffering setup that uses two quantum memories per node. One of these memories has a finite lifetime and is used to buffer the entanglement, while the other memory is only used for entanglement generation. Entanglement generated in the bad memory can be used to pump the entanglement stored in the good memory.
We have modelled the system as a continuous-time stochastic process and derived analytical expressions for both performance metrics.
Our results confirm the intuition that, except in some edge cases, there is a trade-off between consuming entanglement at a high rate (high availability) and consuming high-quality entanglement (high average consumed fidelity).
Remarkably, we found that, in a practical scenario (i.e. when the pumping protocol is bilocal Clifford and there is noise in the good memory), pumping the buffered entanglement is better than no pumping in terms of average consumed fidelity, even if the pumping has some probability of failure.

An assumption that allows us to find analytical solutions for our performance metrics is that the success probability of purification is constant over time. The model would be more realistic if the probability of successful purification was dependent on the state fidelity at that time, since this is the case for most protocols (in particular, the probability of successful purification is typically lower for input states with lower fidelity). This may mean that, realistically, the computation of the average fidelity when \textit{conditioning} on successful purification may bias the system towards higher fidelity. 
However, we believe the comparison of our model (constant success probability) with a more realistic one incorporating this effect (success probability dependent on $F(t)$) to be beyond the scope of this work, since we expect this to greatly complicate the analysis of the problem. 

Our proposed metrics can be used to evaluate the performance of other entanglement buffering systems.
An interesting extension of this work would be to compare the performance of the 1G1B system to a bipartite entanglement buffering setup with $n$ quantum memories per node. In such a system, one could employ more advanced purification protocols that consume more than two entangled states.
We also expect that the mathematical framework developed in this work can be used to initiate the performance analysis of more complex systems.
We leave this as future work.

%Lastly, in our analysis of the 1G1B system with a bilocal Clifford pumping protocol, we have considered depolarizing noise, although similar methods can be used for other types of noise, such as dephasing noise.

\reb{\section*{Acknowledgements}
We thank K. Goodenough, S. de Bone, and P. Kaku for discussions and feedback, and T. Beauchamp, F. Ferreira da Silva, and S. Kar for proofreading.
BD aknowledges financial support from a KNAW Ammodo Award (SW).
\'{A}GI acknowledges support from the Netherlands Organisation for Scientific Research (NWO/OCW), as part of the Frontiers of Nanoscience program.
SW acknowledges support from an NWO VICI grant.}

\bibliographystyle{plain}
%\begin{thebibliography}{9}
%\input{refs.bib}
%\end{thebibliography}

%\bibliographystyle{plain}
%\bibliography{refs}

\onecolumn\newpage
\appendix

\clearpage
\section{Quantum preliminaries}
\label{app:quantum_preliminaries}
Here, we introduce the quantum information concepts needed in this work. For a general reference on quantum information theory, see, e.g., ref. \cite{Nielsen2002}.

In classical computer science, information is generally stored in the form of \emph{bits}, discrete binary variables that can take value zero or one. In quantum information, bits are generalized to \emph{qubits}.
Qubits describe systems that can be in a linear combination of two different states (say, state zero and state one).
As a consequence, a qubit must not be described as a binary variable but as a vector $\ket{\psi}\in\mathbb{C}^2$ with unit norm.
We represent these vectors using the Dirac notation, and we refer to $\ket{\cdot}$ as a ket and to its conjugate transpose $\bra{\cdot} \equiv \ket{\cdot}^\dagger$ as a bra.
The state of a \emph{pure} $n$-qubit system is then described by a $d$-dimensional vector $\ket{\psi}\in\mathbb{C}^d$, with $d=2^n$.
The elements of the basis of $\mathbb{C}^d$ are usually labeled $\ket{\mathbf{x}}$, with $\mathbf{x}\in\{0,1\}^n$. For example, the basis of the two-qubit space $\mathbb{C}^4$ can be written as $\{ \ket{00}, \ket{01}, \ket{10}, \ket{11} \}$.

Let us now consider a two-qubit system $\ket{\psi}\in\mathbb{C}^4$.
When one of the qubits is in state $\ket{\psi}_1\in\mathbb{C}^2$ and the other qubit is in state $\ket{\psi}_2\in\mathbb{C}^2$, we can write the joint state of the system as the tensor product of the individual qubits: $\ket{\psi} = \ket{\psi}_1 \otimes \ket{\psi}_2$. This is called a \emph{product state}.
When the two qubits are \emph{entangled}, it is not possible to describe their joint state as a tensor product.
One of the intuitive effects of the entanglement is that, if both qubits are measured, the measurement outcomes will be correlated.
This is the basis for many quantum networking applications, such as quantum key distribution \cite{Ekert1991} and quantum secret sharing \cite{Hillery1999,Cleve1999}.

Noise and quantum operations, such as \emph{gates} and \emph{measurements}, modify the state of a quantum system.
When dealing with noisy systems, it is generally more convenient to represent quantum states using the \emph{density matrix} formalism instead of kets and bras.
The density matrix of a \emph{pure} state $\ket{\psi}$ can be written as the outer product $\rho = \ketbra{\psi}$.
When the system is not pure, it is called \emph{mixed}, and it is written as a mixture of pure states: $\rho = \sum_i \alpha_i \ketbra{\psi_i}$, with $\alpha_i\in [0,1]$ and $\sum_i \alpha_i = 1$.
Intuitively, this corresponds to a system with some degree of uncertainty from a classical point of view. For example, consider a device that prepares the state $\ket{\psi_1}$ with probability $\alpha$ and the state $\ket{\psi_2}$ with probability $1-\alpha$. We can write the output state of the device as a mixed state $\alpha \ketbra{\psi_1} + (1-\alpha) \ketbra{\psi_2}$.

We refer to a two-qubit state $\rho$ as an \emph{entangled link} (in our setup, each qubit is located at a distant node, and the entanglement can be regarded as "link" connecting them).
Entangled links can be entangled to different degrees.
The \emph{Bell states} are examples of two-qubit maximally entangled (pure) states:
    \begin{equation}
        \ket{\phi^+} = \frac{\ket{00}+\ket{11}}{\sqrt{2}}, \;
        \ket{\psi^+} = \frac{\ket{01}+\ket{10}}{\sqrt{2}}, \;
        \ket{\psi^-} = \frac{\ket{01}-\ket{10}}{\sqrt{2}}, \;
        \ket{\phi^-} = \frac{\ket{00}-\ket{11}}{\sqrt{2}}.
    \end{equation}
Measurements on the individual qubits of any of these states yield maximally correlated outcomes.
All two-qubit maximally entangled states are equivalent, since they can be mapped via single-qubit operations to each other. Therefore, we can measure the \emph{quality of the entanglement} of a two-qubit state by measuring how close the state is to one of the Bell states, say $\ket{\phi^+}$.
Formally, we do this using the \emph{fidelity} of the state:
\begin{equation}
    F(\rho) = \bra{\Phi^+}\rho \ket{\Phi^+},
\end{equation}
where $\rho$ is the density matrix of an arbitrary two-qubit state.

Lastly, an important type of mixed entangled state is the \emph{Bell-diagonal state}:
\begin{equation*}
        \rho_\mathrm{\scriptscriptstyle BD} = F_\mathrm{\scriptscriptstyle BD} \ketbra{\phi^+} + \lambda_1 \ketbra{\psi^+} + \lambda_2 \ketbra{\psi^-} + \lambda_3 \ketbra{\phi^-},
    \end{equation*}
with $F_\mathrm{\scriptscriptstyle BD}, \lambda_1, \lambda_2, \lambda_3 \in [0,1]$ subjected to the normalization constraint $F_\mathrm{\scriptscriptstyle BD} + \lambda_1 + \lambda_2 + \lambda_3 = 1$. The fidelity of this state is $F_\mathrm{\scriptscriptstyle BD}$.
Bell-diagonal states are relevant because any two-qubit state can be transformed into Bell-diagonal form while preserving the fidelity by applying extra noise, a process known as \emph{twirling} \cite{Bennett1996a,Horodecki1999}. % This is also used as motivation to use Bell-diagonal states in Rozpedek2018a
A specific instance of Bell-diagonal state is the Werner state \cite{Werner1989}:
\begin{equation}
    \rho_\mathrm{\scriptscriptstyle W} = F \ketbra{\phi^+} + \frac{1-F}{3} \ketbra{\psi^+} + \frac{1-F}{3} \ketbra{\psi^-} + \frac{1-F}{3} \ketbra{\phi^-},
    \label{eqn:werner_state}
\end{equation}
with fidelity $F\in[0,1]$. The Werner state corresponds to a maximally entangled state that has been subjected to isotropic noise.

\clearpage
\reb{ \section{General form of jump function}
\label{app:jump_function_general_form}
In this Appendix, we explain the form (\ref{eqn:jump_function_rational}) and (\ref{eqn:probability_linear_general}) of the jump function and success probability for a general purification protocol, for two input states $\rho_\mathrm{\scriptscriptstyle W}$ and $\rho_{\mathrm{new}}$, where
    \begin{equation*}
\rho_\mathrm{\scriptscriptstyle W} = F \ketbra{\phi^+} + \frac{1-F}{3} \ketbra{\psi^+} + \frac{1-F}{3} \ketbra{\psi^-} + \frac{1-F}{3} \ketbra{\phi^-}
\end{equation*}
is a Werner state and $\rho_{\mathrm{new}}$ is a general two-qubit state. Suppose that the purification protocol is described by a sequence of (possibly noisy) quantum operations that are described by a CPTP map $\Lambda$, and the final measurement outcome that signals success has measurement operator $M_{\mathrm{succ}}$. From e.g. Chapter 2.4 of \cite{Nielsen2002}, the output state is then given~by 
\begin{equation}
    \rho' = \frac{M_{\mathrm{succ}} \Lambda\left( \rho_\mathrm{\scriptscriptstyle W} \otimes \rho_{\mathrm{new}} \right)M_{\mathrm{succ}}^{\dag}}{p(F,\rho_{\mathrm{new}})},
\end{equation}
where 
\begin{equation}
    p(F,\rho_{\mathrm{new}}) = \mathrm{Tr}\left[M_{\mathrm{succ}} \Lambda\left( \rho_\mathrm{\scriptscriptstyle W} \otimes \rho_{\mathrm{new}} \right)M_{\mathrm{succ}}^{\dag} \right].
\end{equation}
We next rewrite the Werner state as
\begin{align*}
    \rho_\mathrm{\scriptscriptstyle W} &= F \ketbra{\phi^+} + (1-F) \rho^{\perp} \\ &= \rho^{\perp} + F \left( \ketbra{\phi^+} - \rho^{\perp}\right)
\end{align*}
where 
\begin{equation*}
    \rho^{\perp} = \frac{1}{3}\left( \ketbra{\psi^+} + \ketbra{\psi^-} + \ketbra{\phi^-}\right),
\end{equation*}
and $p$ is the probability of success. We therefore have 
\begin{align*}
M_{\mathrm{succ}} \Lambda\left( \rho_\mathrm{\scriptscriptstyle W} \otimes \rho_{\mathrm{new}} \right)M_{\mathrm{succ}}^{\dag} =\; &M_{\mathrm{succ}} \Lambda\left( \rho^{\perp} \otimes \rho_{\mathrm{new}} \right)M_{\mathrm{succ}}^{\dag}\\
&+ F\cdot M_{\mathrm{succ}} \Lambda\left( \left(\ketbra{\phi^+} - \rho^{\perp}\right) \otimes \rho_{\mathrm{new}} \right)M_{\mathrm{succ}}^{\dag},
\end{align*}
and taking the trace of the above yields
\begin{equation*}
    p(F,\rho_{\mathrm{new}}) = d(\rho_{\mathrm{new}} ) + F\cdot c(\rho_{\mathrm{new}} ),
\end{equation*}
where $c$ and $d$ are obtained from the choice of purification protocol, i.e. from $\Lambda$ and $M_{\mathrm{succ}}$. Similarly, the output fidelity of upon success is given by 
\begin{align*}
    \bra{\phi^+} \rho' \ket{\phi^+} = \frac{F\cdot \tilde{a}(\rho_{\mathrm{new}}) + \tilde{b}(\rho_{\mathrm{new}})}{p(F,\rho_{\mathrm{new}})},
\end{align*}
where 
\begin{align*}
    \tilde{a}(\rho_{\mathrm{new}}) &= \bra{\phi^+} M_{\mathrm{succ}} \Lambda\left( \left( \ketbra{\phi^+} - \rho^{\perp} \right) \otimes \rho_{\mathrm{new}} \right)M_{\mathrm{succ}}^{\dag} \ket{\phi^+}, \\  \tilde{b}(\rho_{\mathrm{new}}) &= \bra{\phi^+} M_{\mathrm{succ}} \Lambda\left(\rho^{\perp}  \otimes \rho_{\mathrm{new}} \right)M_{\mathrm{succ}}^{\dag} \ket{\phi^+}.
\end{align*}
This confirms the form (\ref{eqn:jump_function_rational}) and (\ref{eqn:probability_linear_general}) for the jump function and success probability.}

\clearpage

\section{Formulae for performance metrics}\label{app.performance_metrics}
In this Appendix, we prove Proposition \ref{prop:1G1B_stationary_distribution} and Theorem \ref{thm:average_fidelity_formula}, which provide the formulae for our two performance metrics (availability and average consumed fidelity).
First, in \ref{app:subsec_simplified_1G1B}, we describe the stochastic process in the 1G1B setup in a simplified form and we provide some intermediate results that are necessary for the main proofs.
Then, in \ref{app.num_pur_subsec}, we employ the results from \ref{app:subsec_simplified_1G1B} to prove Proposition \ref{prop:1G1B_stationary_distribution} and Theorem \ref{thm:average_fidelity_formula}.

\subsection{Simplified 1G1B}
\label{app:subsec_simplified_1G1B}
We now only view the 1G1B system as taking one of two states: $\emptyset$ (no entangled link in memory G), or $\neg \emptyset$ (link in memory G). The system then alternates between these two states. For an illustration, see Figure \ref{fig:1g1b_simplified}. 

\begin{figure}[h]
    \centering
    \includegraphics[width=0.9\linewidth]{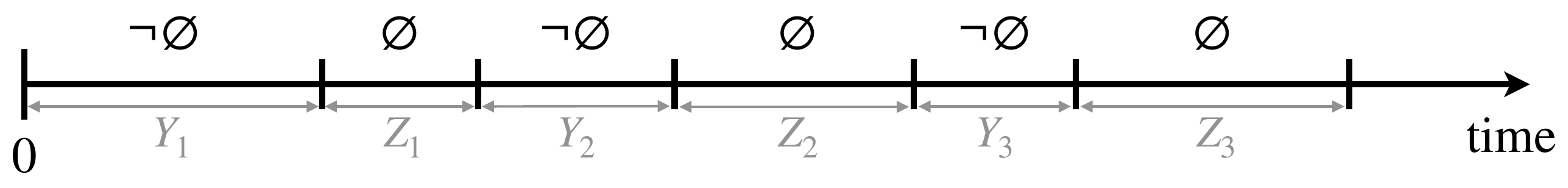}
    \caption{The simplified 1G1B process. The system alternates between the states $\neg \emptyset$ (link in memory G) and $\emptyset$ (no link in memory G). The system starts in $\neg \emptyset$.
    The times spent in $\neg \emptyset$ and $\emptyset$ are denoted by $Y_i$ and $Z_i$, respectively.}
    \label{fig:1g1b_simplified}
\end{figure}

More formally, the simplified 1G1B process is the following.

\begin{definition}[Simplified 1G1B]\label{def.simplified1G1B}
    Let $r(t)\in \{ \emptyset, \neg \emptyset\}$ denote the state of simplified 1G1B at time $t$. Suppose that $r(0)= \neg \emptyset$, i.e. the system starts when there is a link in memory F. Let $Y_1$ be the time until this first link is removed, and let $Z_1$ be the time for which the system is empty until a fresh link is produced again. Let $\{Y_i\}_{i\geq 1}$ be the times spent in $\neg \emptyset$ until the link was removed from memory G (due to consumption or failed purification), and $\{Z_i\}_{i\geq 1}$ be the times which the system spent in $\emptyset$ until a link was produced. Then, according to our model of 1G1B, the $Y_i$ are i.i.d. and exponentially distributed with rate $\beta \coloneqq \mu + \lambda q (1-p)$, and the $Z_i$ are i.i.d and exponentially distributed with rate $\lambda$.
\end{definition}

Recall that $\lambda$ is the rate of generation of new entangled links, $\mu$ is the rate of consumption of links in memory G, $q$ is the probability of immediately using new links for pumping, and $p$ is the probability of successful pumping.

We will write the distribution functions as $F_Y(t) = P(Y_1\leq t) = 1-e^{-\beta t}$, and $F_Z(t) = P(Z_1\leq t) = 1-e^{-\lambda t}$. The process $X_i \coloneqq Y_{i} + Z_i$ defines a \textit{renewal process}, which we introduce with the following definition. 
\begin{definition}
    A \textit{renewal process} $ \{ N = N(t): t\geq 0 \}$ is a process such that 
    \begin{equation}
        N(t) = \max \{n: A_n \leq t \}
    \end{equation}
    where $A_0 = 0$, $A_n = X_1+...+X_n$ for $n\geq 1$, and $X_i$ is a sequence of i.i.d. and strictly positive random variables. 
    \label{def:renewal_process}
\end{definition}
The value $A_n$ is referred to as the $n$th \textit{arrival time} of the process, and the values $X_i$ are known as the \textit{interarrival times}. From now on, we also use $A_0 = 0$, $A_n = X_1+....+X_n$ to denote the $n$th time at which a fresh link is produced, causing the system to move from $\emptyset$ into $\neg \emptyset$.

The \textit{renewal function} is central to renewal theory, which we define below. Throughout, we use the convention $\dd g(x) \equiv g'(x) \dd x$ for differentiable functions $g$.
\begin{definition}
    Let $N(t)$ be a renewal process. Then, the \textit{renewal function} is $m(t)\coloneqq \mathbb{E}[N(t)]$.
\end{definition}
We will derive formulae for the availability and average consumed fidelity using this mathematical framework. An important result that we will use in order to do this is the renewal theorem, which we state below. This result assumes that the $X_i$ are not \textit{arithmetic}. If $X_1$ is arithmetic, this essentially means that $X_1$ only takes values in a set $\{m k: m=0,\pm 1,... \}$, with $k>0$. For more details of arithmetic random variables, see Chapter 10 of \cite{grimmett2020}.
\begin{theorem}[Renewal Theorem/ Theorem 10.1.11 from \cite{grimmett2020})] Consider a renewal process as given in Definition \ref{def:renewal_process}. Let $F_X$ be the distribution function of the random variable $X_1$, where $X_1$ is not arithmetic. Let $H(t)$ be a bounded function. Consider solutions $f$ to the renewal-type equation
\begin{equation}
    f(t) = H(t) + \int_{0}^t f(t-x) \dd F_X(x). 
    \label{eqn:renewal_type_equation}
\end{equation}
Then, a solution is 
\begin{equation}\label{eq.renewal_type_solution}
    f(t) = H(t) + \int_{0}^t H(t-x) \dd m(x).
\end{equation}
 If $H$ is bounded on finite intervals
 %\agi{Isn't this always true by definition?}
 then $f$ is bounded on finite intervals, and (\ref{eq.renewal_type_solution}) is the unique solution of (\ref{eqn:renewal_type_equation}) with this property.
\label{thm:renewal}
\end{theorem}
The renewal-type equation often arises when studying renewal processes, as we will see further on. The following result may be derived using Theorem \ref{thm:renewal}, and is useful when taking the infinite limit.
\begin{theorem}[Key renewal theorem/Theorem 11.2.7 from \cite{grimmett2020}]
    If $g:[0,\infty)\rightarrow [0,\infty)$ is such that 
 \begin{enumerate}[label=(\alph*)]
 \item $g(t) \geq 0 $ for all $t$,
 \item $\int_{0}^{\infty} g(t) \dd t <\infty$,
 \item $g$ is a non-increasing function,
 \end{enumerate}
 then $$\lim_{t\rightarrow\infty} \int_{0}^t g(t-x)\dd m(x) = \frac{1}{\mathbb{E}[X_1]} \int_{0}^{\infty} g(x) \dd x,$$
 whenever $X_1$ is not arithmetic.
 \label{thm:key_renewal}
\end{theorem}

We are now partially equipped to show the formulae for the availability and average fidelity. Next, we show a set of intermediate results that we will need for the main proofs.

\begin{proposition}
    Let $p(t) = P(r(t)=\neg \emptyset)$ be the probability that a link is available at time $t$ in the simplified 1G1B process. Then,
    \begin{equation}
        \lim_{t\rightarrow \infty} p(t) = \frac{\mathbb{E}(Y_1)}{\mathbb{E}(Y_1)+\mathbb{E}(Z_1)}.
        \label{eqn:lim_p(t)}
    \end{equation}
    \label{prop:limit_p(t)}
\end{proposition}
\begin{proof}
    We proceed by conditioning on the value of $X_1$. Now,
    \begin{align}
        p(t) = P(r(t)=\neg \emptyset \cap X_1>t)+P(r(t)=\neg \emptyset \cap X_1<t).
    \label{eqn:prob_expansion}
    \end{align}
    Notice that the event $\{r(t)=\neg \emptyset \cap X_1>t \} $ occurs if and only if $Y_1 >t$. Further, if $x<t$, then 
    \begin{align}
        P(r(t)=\neg \emptyset |X_1=x) = p(t-x),
    \end{align}
    since the process starts afresh at time $x$. Then, (\ref{eqn:prob_expansion}) becomes
    \begin{equation}
        p(t) = 1-F_Y(t) +\int_{0}^t p(t-x) \dd F_X(x),
    \end{equation}
    where $\dd F_X(x) \equiv F'_X(x) \dd x$. We now see that this is of the form (\ref{eqn:renewal_type_equation}) with $H(t) = 1-F_Y(t)$, and so by Theorem \ref{thm:renewal}, 
    \begin{equation}
        p(t) = 1-F_Y(t)+\int_0^{t}\left(1-F_Y(t-x) \right)\dd m(x).
    \end{equation}
    Taking the infinite limit,
    \begin{equation}
        \lim_{t\rightarrow \infty } p(t) = 1 - 1 + \lim_{t\rightarrow \infty } \int_0^{t}(1-F_Y(t-x))\dd m(x).
    \end{equation}
    It can be seen that $H(t) = 1-F_Y(t)$ satisfies the conditions (a)-(c) required by Theorem \ref{thm:key_renewal}, so we may apply this Theorem to take the limit:
    \begin{align}
        \lim_{t\rightarrow \infty } p(t) &= \frac{1}{\mathbb{E}[X_1]}\int_0^{\infty} (1-F_Y(x))\dd x \\ &= \frac{1}{\mathbb{E}[X_1]}\int_0^{\infty} P(Y_1>x) \dd x = \frac{\mathbb{E}[Y_1]}{\mathbb{E}[X_1]}.
    \end{align}
    Finally, using $\mathbb{E}[X_1] = \mathbb{E}[Y_1+Z_1] = \mathbb{E}[Y_1]+ \mathbb{E}[Z_1]$ suffices to show (\ref{eqn:lim_p(t)}).
\end{proof}
Recall that the average fidelity of the system at a given time $t$ is dependent on the time spent in each purification level leading up to this point. Therefore, in order to understand the average fidelity we first of all look at the \textit{current lifetime} in this simplified setting. 

\begin{definition}[Current lifetime] Consider the simplified $1G1B$ system. Let $C(t)$ be the time spent so far in a state at time $t$. More formally,
\begin{equation}
    C(t) = 
    \begin{cases}
       t - A_{N(t)}, \text{   if } r(t) = \neg \emptyset, 
        \\ t - A_{N(t)} - Y_{N(t)+1}, \text{  
 if } r(t) = \emptyset.
    \end{cases}
\end{equation}
\label{def:current_lifetime}
\end{definition}

\begin{figure}[h]
    \centering
    \includegraphics[width=0.9\linewidth]{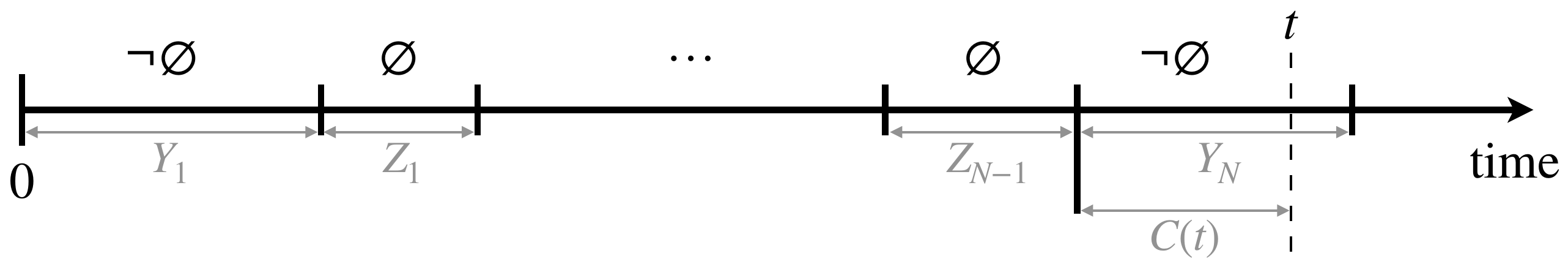}
    \caption{Current lifetime of the simplified 1G1B process. The random variable $C(t)$ denotes the time spent so far in the current state at time $t$. This is most interesting when $r(t)=\neg \emptyset$, because it tells us the age of a link in memory.}
    \label{fig:current_lifetime}
\end{figure}

The first case ($r(t) = \neg \emptyset$) is of most interest here, because it corresponds to when a link is in memory and is subject to decoherence. See Figure \ref{fig:current_lifetime} for an illustration of this concept. In the following Lemma, we characterise the distribution of $C(t)$, conditional on being in the state $\neg \emptyset$.
\begin{lemma}
Consider the simplified 1G1B system. The limiting distribution of $C(t)$ conditional on there being a link is given by 
\begin{equation}
    \lim_{t\rightarrow \infty} P\left(C(t)>x|r(t)=\neg \emptyset \right) = \frac{1}{\mathbb{E}[Y_1]} \int_{x}^{\infty} (1-F_Y(s)) \dd s,
\end{equation}
which is an exponential distribution with parameter $\beta$ when $Y_1\sim \text{Exp}(\beta)$.
\label{lem:current_lifetime_limit}
\end{lemma}
\begin{proof}
 Writing
 \begin{equation}
     P(C(t)>x|r(t)=\neg \emptyset) = \frac{P(C(t)>x \cap r(t)=\neg \emptyset)}{P(r(t)=\neg \emptyset)},
 \end{equation}
 we see that the bottom of the fraction has already been dealt with in Proposition \ref{prop:limit_p(t)}. We therefore focus on 
 \begin{equation}
     G(t,x) \coloneqq P(C(t)>x \cap r(t)=\neg \emptyset).
 \end{equation}
 Conditioning on $X_1$, we see that 
 \begin{equation}
     G(t,x) = P(C(t)>x \cap r(t)=\neg \emptyset \cap X_1>t) + P(C(t)>x \cap r(t)=\neg \emptyset \cap X_1\leq t).
     \label{eqn:G(t)_expansion}
 \end{equation}
 Now, the event $\{C(t)>x \cap r(t)=\neg \emptyset \cap X_1>t \}$ occurs if and only if $Y_1>t>x$. Moreover, if $y<t$ then the process starts afresh from time $y$, and
 \begin{equation}
     P(C(t)>x \cap r(t)=\neg \emptyset | X_1=y) = G(t-y).
 \end{equation}
 Then, noting that $G(t,x)=0$ for $t<x$,  (\ref{eqn:G(t)_expansion}) becomes
 \begin{equation}
     G(t,x) = \mathds{1}_{\{t\geq x\}} (1-F_Y(t)) + \int_{0}^t G(t-y)\dd F_X(y),
 \end{equation}
 which is in the form of (\ref{eqn:renewal_type_equation}) with $H(t)=\mathds{1}_{\{t\geq x\}} (1-F_Y(t))$. Then, by Theorem \ref{thm:renewal}, $G(t,x)$ is given by
 \begin{equation}
     G(t,x) = \mathds{1}_{\{t\geq x\}} (1-F_Y(t)) + \int_{0}^t \mathds{1}_{\{t-y\geq x\}} (1-F_Y(t-y)) \dd m(y)
 \end{equation}
which has limit
\begin{align}
    \lim_{t\rightarrow \infty} G(t,x) &= 0 + \lim_{t\rightarrow \infty} \int_{0}^{t-x}(1-F_Y(t-y)) \dd m(y) \\ &= \lim_{s\rightarrow \infty} \int_{0}^s (1-F_Y(s+x-y)) \dd m(y),
\end{align}
letting $s=t-x$. Then, noting that $g(s) = 1-F_Y(s+x)$ satisfies conditions (a)-(c) of Theorem \ref{thm:key_renewal}, we may apply this to find
\begin{align}
    \lim_{t\rightarrow \infty} G(t,x) = \frac{1}{\mathbb{E}[X_1]}\int_0^{\infty}g(s) \dd s &= \frac{1}{\mathbb{E}[X_1]} \int_{0}^{\infty} (1-F_Y(s+x)) \dd s \\ &= \frac{1}{\mathbb{E}[X_1]} \int_{x}^{\infty} (1-F_Y(s)) \dd s.\label{eq.G_intermediate_step}
\end{align}
From Proposition \ref{prop:limit_p(t)}, we observe that $$\mathbb{E}[X_1] = \frac{\mathbb{E}[Y_1]}{\lim_{t\rightarrow\infty}P\left(r(t)=\neg \emptyset\right)}.$$
We can use this to rewrite (\ref{eq.G_intermediate_step}) as follows:
\begin{equation}
    \lim_{t\rightarrow \infty} P\left(C(t)>x|r(t)=\neg \emptyset \right) = \frac{1}{\mathbb{E}[Y_1]} \int_{x}^{\infty} (1-F_Y(s)) \dd s,
\end{equation}
which we notice is only dependent on the distribution of $Y_1$. In the case $Y_1\sim \text{Exp}(\beta)$, as considered in the 1G1B system, 
\begin{equation}
    \lim_{t\rightarrow \infty} P\left(C(t)>x|r(t)=\neg \emptyset \right) = \beta \int_{x}^{\infty} e^{-\beta s} \dd s = e^{-\beta x},
\end{equation}
and so conditional on there being a link, the current lifetime approaches an exponential distribution.
\end{proof}

We have now characterised the availability (Proposition \ref{prop:limit_p(t)}) and current lifetime (Lemma \ref{lem:current_lifetime_limit}) for the simplified 1G1B system. However, note that both Proposition \ref{prop:limit_p(t)} and Lemma \ref{lem:current_lifetime_limit} assumed that the system starts in the state $r(0)=\neg \emptyset$. This was necessary in order to satisfy all of the conditions (a)-(c) of Theorem \ref{thm:key_renewal}. The result below states that Theorem \ref{thm:key_renewal} still holds, even if the renewal process is \textit{delayed}, which means that the first arrival has a different distribution to the others. For more details of delayed renewal processes, see \cite{grimmett2020} or \cite{Mitov2014}.
\begin{definition}
    Let $\{ X_i\}_{i\geq 1}$ be independent positive random variables such that $\{X_i \}_{i\geq 2}$ have the same distribution. Let $A_0=0$, $A_n = \sum_{i=1}^n X_i$, and $N^\mathrm{d} = \max \{n:A_n\leq t \}$. Then, $N^\mathrm{d}(t)$ is a delayed renewal process.
\end{definition}
\begin{definition}
    Let $N^\mathrm{d}$ be a delayed renewal process. Then, $m^\mathrm{d}(t)\coloneqq \mathbb{E}[N^\mathrm{d}(t)] $ is the delayed renewal function.
\end{definition}
\begin{theorem}[Key renewal theorem for delayed renewal processes/Theorem 1.20 of \cite{Mitov2014}]
    Consider a delayed renewal process $N^\mathrm{d}(t)$. If $g : [0,\infty)\rightarrow [0,\infty)$ satisfies the same conditions (a)-(c) of Theorem \ref{thm:key_renewal}, then
    \begin{equation}
         \lim_{t\rightarrow \infty} \int_{0}^t g(t-x)\dd m^\mathrm{d}(x) = \frac{1}{\mathbb{E}[X_2]} \int_{0}^{\infty} g(x) \dd x.
    \end{equation}
\label{thm:key_renewal_delayed}
\end{theorem}
A consequence of Theorem \ref{thm:key_renewal_delayed} is that even for delayed renewal processes, the limiting distribution is the same as for the non-delayed case. Therefore, the results of Proposition \ref{prop:limit_p(t)} and Lemma \ref{lem:current_lifetime_limit} hold even when the distribution of $X_1$ is not the same as $\{X_i\}_{i\geq 2}$. In particular, they still hold when the process starts in $\emptyset$. This is summarised with the following corollary.

\begin{corollary}
    Consider the simplified 1G1B process, now altered to start in $r(0)=\emptyset$. Let $Z_0$ be the time for which the system is empty until the first fresh link is produced. Let $Y_1$ be the time in which this link is present in memory until it is removed again, and so on. Let the probability of finding a link at time $t$ be $p(t)=P(r(t)=\neg \emptyset)$. Then,
    \begin{equation}
        \lim_{t\rightarrow \infty} p(t) = \frac{\mathbb{E}[Y_1]}{\mathbb{E}[Y_1]+\mathbb{E}[Z_1]} = \frac{\lambda}{\lambda+\beta},
        \label{eqn:limit_p(t)_delayed}
    \end{equation}
    and the distribution of the current lifetime of a link satisfies
    \begin{equation}
        \lim_{t\rightarrow \infty} P\left(C(t)>x|r(t)=\neg \emptyset \right) = \frac{1}{\mathbb{E}[Y_1]} \int_{x}^{\infty} (1-F_Y(s)) \dd s
         = e^{-\beta x}.
    \end{equation}
\label{cor:limit_delayed_1G1B_simplified}
\end{corollary} 
Recalling that $\beta = \mu + \lambda q (1-p)$, we see that the formula for the availability in Proposition \ref{prop:1G1B_stationary_distribution} is already shown by (\ref{eqn:limit_p(t)_delayed}).

%\subsection{Time spent in pumping levels}
\subsection{Availability and average consumed fidelity in 1G1B}
\label{app.num_pur_subsec}
Here, we compute the availability and the rest of the steady-state distribution of the 1G1B system (Proposition \ref{prop:1G1B_stationary_distribution}), as well as the average consumed fidelity (Theorem \ref{thm:average_fidelity_formula}).

In order to calculate the average fidelity, we not only need the time spent in $\neg \emptyset$, but also the times spent in each pumping level leading up to the current one. 

From 1G1B (Definition \ref{def:1G1B_system}), one may define a simplified 1G1B system as 
\begin{equation*}
    r(t) = \begin{cases}
        \neg \emptyset \text{  if  } s(t) \geq 0 \\ 
        \emptyset \text{  if  } s(t) = \emptyset.
    \end{cases}
\end{equation*}

For the characterisation of the fidelity of the link in memory at time $t$, $F(t)$, we are interested in the successful pumping attempts that occur in the the time interval $[A_{N(t)},A_{N(t)}+C(t))$, where $C(t)$ is the current lifetime (Definition \ref{def:current_lifetime}).
In 1G1B, the successful pumping attempts are a Poisson process with rate $\delta \coloneqq \lambda p q$. Since the rate is constant for all $t$, the number of successful pumping attempts within the interval $[A_{N(t)},A_{N(t)}+C(t))$ has the identical distribution as the number of successful pumping attempts in the time interval $[0,C(t))$.
From Corollary \ref{cor:limit_delayed_1G1B_simplified}, we see that $C(t)$ converges in distribution to $C\sim \text{Exp}(\beta)$. In the following Lemma, we characterise the number of successful pumping attempts that occur within the time $C$, and the time spent between each pair of consecutive pumping rounds. See Figure \ref{fig:num_pumping_rounds} for an illustration. An observation that we use below is that within the time interval $[0,C)$, the times at which pumping occurs form a separate renewal process, which is convenient for notation.

\begin{figure}
    \centering
    \includegraphics[width=0.9\linewidth]{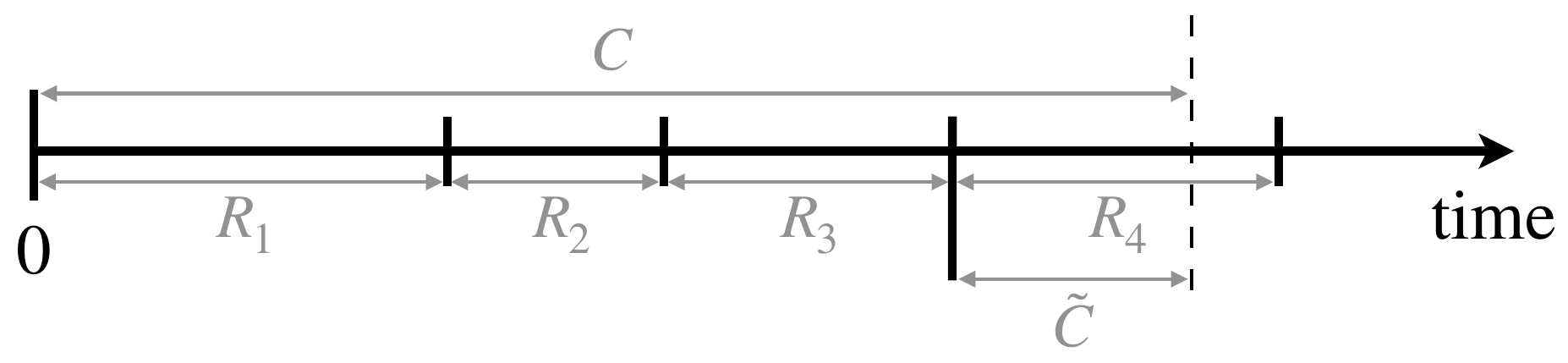}
    \caption{Number of pumping rounds. We are interested in the number of pumping rounds that have been carried out while a link is in memory. Here, $C$ is the (limiting) distribution of the current lifetime in memory (see Figure \ref{fig:current_lifetime}), and $R_i$ is the time between the $(i-1)$th and $i$th pumping round.}
    \label{fig:num_pumping_rounds}
\end{figure}

\begin{lemma}
Consider a renewal process $\tilde{N}(t)$ with arrival times $S_0=0$, $S_n = \sum_{i=1}^n R_i$, with $R_1\sim \text{Exp}(\delta)$. Let $C\sim \text{Exp}(\beta)$ be independent of the $R_i$. Let $M=N(C)$ be the number of arrivals that have occurred by time $C$. Let $\tilde{C}\coloneqq C - S_M$ be the current lifetime at time $C$. Then, 
    \begin{enumerate}
        \item The distribution of $M$ is given by 
    \begin{equation}
    P(M\geq m)=\left(\frac{\delta}{\beta + \delta}\right)^{m}, 
     \label{eqn:cumu_dist_M}
     \end{equation} 
     or equivalently 
     \begin{equation}
            P(M = m) = \left(\frac{\delta}{\beta + \delta}\right)^{m}\left(\frac{\beta}{\beta + \delta}\right).      \label{eqn:dist_M}
        \end{equation}
        \item Conditional on $M= m$, the random variables  $(R_1,...,R_m,\tilde{C})$ are mutually independent and identically distributed as $\text{Exp}(\beta+\delta)$.
     \end{enumerate}
\label{lem:num_purifications_in_limit}
\end{lemma}
\begin{proof}
\begin{enumerate}
    \item We proceed by induction. Letting $F_R\coloneqq P(R\leq x)$ We have
    \begin{align*} 
    P(M \geq 1) = P(C>R_1) &= \int_{0}^{\infty} P(C>R_1|R_1=x)\dd F_R(x)\\ &= \int_0^{\infty} e^{-\beta x}\cdot \delta e^{-\delta x}\dd x = \frac{\delta}{\delta + \beta},
\end{align*} 
where we have used $P(C>R_1|R_1=x) = P(C>x) = e^{-\beta x}$ and $R_1\sim \text{Exp}(\delta)$. Then, assuming (\ref{eqn:cumu_dist_M}), 
\begin{align*}
    P(M\geq m+1) &= P(C>S_{m+1}) \\ &= P(C>R_{m+1}+S_m) \\ &\stackrel{a}{=} P(C>R_{m+1})P(C>S_m), \\
    &= P(C>R_{1})P(M\geq m)
    \\ &\stackrel{b}{=} \left(\frac{\delta}{\beta + \delta}\right) \left(\frac{\delta}{\beta + \delta}\right)^{m} = \left(\frac{\delta}{\beta + \delta}\right)^{m+1}.
\end{align*}
In  step ($b$), we have used the inductive assumption. In step ($a$) we have made use of the memoryless property of the exponential distribution: Since $R_{m+1}$ and $S_{m}$ are positive and independent random variables, this has as a consequence
\begin{align}
    P(C>R_{m+1}+S_m) &=\int_{0}^{\infty} \int_{0}^{\infty} \dd F_R(r) \dd F_{S_m}(s) P(C>r+s) \nonumber \\ &=  \int_{0}^{\infty} \int_{0}^{\infty} \dd F_R(r) \dd F_{S_m}(s) P(C>r)P(C>s) \nonumber \\ &= P(C>R_{m+1})P(C>S_m).
    \label{eqn:memoryless_property}
\end{align}
Finally, (\ref{eqn:dist_M}) follows from 
\begin{align*}
    P(M = m) &= P(M\geq m) - P(M\geq m+1) \\ &= \left(\frac{\delta}{\beta + \delta}\right)^{m} -\left(\frac{\delta}{\beta + \delta}\right)^{m+1} \\ &= \left(\frac{\delta}{\beta + \delta}\right)^{m}\left(\frac{\beta}{\beta + \delta}\right).
\end{align*}
\item  We firstly note that for any events $E_1$, $E_2$, $E_3$, it holds that 
        \begin{equation}
            P(E_1\cap E_2 \cap E_3) = P(E_1 \cap E_2) - P(E_1 \cap E_2 \cap \neg E_3),
            \label{eqn:intersection_events}
        \end{equation}
        where $\neg E$ denotes the complement of the event $E$.
        Now, consider the events 
        \begin{equation*}
            E_1 = \Big\{ R_i > x_i \; \forall i = 1,\dots,m+1\Big\}, \;\; E_2 = \Big\{C \geq x_{m+1}+\sum_{i=1}^m R_i \Big\}, \;\;E_3= \Big\{ C  < \sum_{i=1}^{m+1} R_i \Big\}.
        \end{equation*}
        Now, 
        \begin{align*}
            E_1 \cap E_2 \cap E_3 &= \Big\{ R_1 > x_1,\dots,\: R_m>x_m, R_{m+1}>x_{m+1} \cap \sum_{i=1}^{m+1}R_i> C \geq x_{m+1}+\sum_{i=1}^m R_i \Big\} \\ &\stackrel{a}{=} \Big\{ R_1 > x_1,\dots,\: R_m>x_m, \tilde{C}>x_{m+1} \cap \sum_{i=1}^{m+1} R_i> C \geq \sum_{i=1}^m R_i \Big\} \\ &\stackrel{b}{=} \Big\{ R_1 > x_1,\dots,\: R_m>x_m, \tilde{C}>x_{m+1} \cap M=m \Big\},
        \end{align*}
        where in ($a$) we have used the definition of $\tilde{C}$, and in ($b$) we have used the definition of $M$. Then, by (\ref{eqn:intersection_events}), we see that 
        \begin{multline}
            P\left( R_1 > x_1,\dots,\: R_m>x_m, \tilde{C}>x_{m+1} \cap M=m \right)  \\ = P\left( R_1 > x_1,\dots,\: R_{m+1}>x_{m+1} \cap C \geq x_{m+1}+\sum_{i=1}^m R_i \right) \\ - P\left( R_1 > x_1,\dots,\: R_{m+1}>x_{m+1} \cap C \geq \sum_{i=1}^{m+1} R_i\right).
            \label{eqn:cor_step_2}
        \end{multline}
By the independence of the $R_i$, this is equivalent to
    \begin{multline}
        (\ref{eqn:cor_step_2}) 
     = \Bigg[ P\left(C\geq x_{m+1}+\sum_{i=1}^m R_i \Big|  R_1 > x_1 , ... , R_{m+1} > x_{m+1} \right) \\  -P\left(C\geq\sum_{i=1}^{m+1} R_i \Big| R_1>x_1,..., R_{m+1}>x_{m+1}\right) \Bigg]  \prod_{i=1}^{m+1}P\left(R_i > x_i\right),  
\label{eqn:current_lifetime_in_limit_step_2}
    \end{multline}
and we now use the memoryless property of the exponential distribution (see the argument leading to (\ref{eqn:memoryless_property})) to rewrite as
\begin{multline}
    (\ref{eqn:cor_step_2})  = \Bigg[ P\left(C\geq x_{m+1}\big| R_{m+1}>x_{m+1}\right) \prod_{i=1}^m P\left(C\geq R_i\big|R_i>x_i\right)
    \\  - \prod_{i=1}^{m+1} P\left(C\geq R_i\big|R_i>x_i\right) \Bigg]  \prod_{i=1}^{m+1}P\left(R_i > x_i\right), 
\end{multline}
which, using that $P(C_i\geq R_i\big|R_i>x_i)P(R_i>x_i) = P(C_i\geq R_i >x_i)$, becomes
\begin{align*}
    (\ref{eqn:cor_step_2}) &= \Bigg[ P\left(C\geq x_{m+1}\cap R_{m+1}>x_{m+1}\right)-P\left(C\geq R_{m+1}>x_{m+1})\right) \Bigg]  \prod_{i=1}^{m}P\left(C\geq R_i > x_i\right), \\ &=   P\left(R_{m+1}>C\geq x_{m+1}\right)   \prod_{i=1}^{m}P\left(C \geq  R_i > x_i\right),
\end{align*}
where we have again made use of (\ref{eqn:intersection_events}) to rewrite the factor on the left. Now, 
\begin{align*}
    P(C \geq R_1>x_1)
    &= \int_{x_1}^{\infty}  P(C\geq y) \dd F_R(y)
    \\ &= \int_{x_1}^{\infty}  e^{-\beta y}\cdot \delta e^{-\delta y} = \frac{\delta}{ \beta+\delta} e^{-(\beta+\delta)x_1},
\end{align*}
and by symmetry
\begin{align*}
    P(R_1\geq C>x_{m+1}) = \frac{\beta}{ \beta+\delta} e^{-(\beta+\delta)x_{m+1}}.
\end{align*}
We therefore see that 
\begin{align*}
    (\ref{eqn:cor_step_2}) &= \frac{\beta}{ \beta+\delta} e^{-(\beta+\delta)x_{m+1}} \cdot \prod_{i=1}^m \left[ \frac{\delta}{ \beta+\delta} e^{-(\beta+\delta)x_i}\right] \\ &= \frac{\beta}{\beta+\delta} \cdot \left( \frac{\delta}{ \beta+\delta}\right)^m \prod_{i=1}^{m+1} e^{-(\beta+\delta)x_i} = P(M=m) \prod_{i=1}^{m+1}  e^{-(\beta+\delta)x_i}.
\end{align*}
It therefore follows that 
\begin{align*}
    P\left( R_1 > x_1,\dots,\: R_m>x_m, \tilde{C}>x_{m+1} \big| M=m \right) = \prod_{i=1}^{m+1}  e^{-(\beta+\delta)x_i},
\end{align*}
which suffices to show the second result.
\end{enumerate}
\end{proof}
Recalling that $C(t)$ converges in distribution to $C$, we now adapt Lemma \ref{lem:num_purifications_in_limit} to apply to $C(t)$. In order to do this, we use the following result (for a proof, see Chapter 7 of  \cite{grimmett2020}).
\begin{theorem}[Continuous mapping theorem]
    Let $\{X_n\}$ be a sequence of random variables taking values in $\mathbb{R}^k$. If $X_n \rightarrow X$ in distribution as $n\rightarrow \infty$ and $g:\mathbb{R}^k \rightarrow \mathbb{R}^l$ is continuous, then $g(X_n) \rightarrow g(X)$ in distribution as $n\rightarrow \infty$.
    \label{thm:continuous_mapping}
\end{theorem}
\begin{corollary}
    Suppose that $C(t)$ and $X$ are independent random variables, and $C(t)$ converges in distribution to $C$ as $t\rightarrow \infty$. Then, 
    \begin{equation}
        \lim_{t\rightarrow \infty} P\left( C(t)>X \right) = P(C>X).
        \label{eqn:prob_C(t)>X_converges}
    \end{equation}
    \label{cor:limit_sum_random_variables}
\end{corollary}
\begin{proof}
    Consider a sequence of times $\{t_n \}_{n\geq 1}$ such that $0<t_1<t_2<...$ and $\lim_{n\rightarrow \infty}t_n = \infty $. Let $C_n\coloneqq C(t_n)$. Then, $ C_n \rightarrow C$ in distribution. Moreover, since $C_n$ and $X$ are independent for all $n$, the pair $(C_n,-X)\rightarrow (C,-X)$ in distribution. Now, the function $g:\mathbb{R}^2\rightarrow \mathbb{R}$, with $g(x,y)=x+y$ is continuous. Then, by Theorem \ref{thm:continuous_mapping}, $C_n-X \rightarrow C-X$ in distribution, and so 
    \begin{equation*}
        \lim_{n\rightarrow \infty} P(C_n-X>0) = \lim_{n\rightarrow \infty} P(C(t_n)-X>0) = P(C-X>0).
    \end{equation*}
    Since this is true for all such sequences $\{t_n \}$, the result follows.
\end{proof}
In the following corollary, we let the current lifetime be dependent on the parameter $u$ to avoid confusion with the time of the renewal process (which is denoted as $t$). 
\begin{corollary}
Consider a renewal process $N(t)$ with arrival times $S_0=0$, $S_n = \sum_{i=1}^n R_i$, with $R_1\sim \text{Exp}(\delta)$. Suppose that $C(u)$ converges in distribution to $C\sim \text{Exp}(\beta)$ as $u\rightarrow \infty$. Let $M(u)=N(C(u))$ be the number of arrivals that have occurred by time $C(u)$. Let $\tilde{C}(u) \coloneqq C(u) - S_{M(u)}$ be the current lifetime at time $C(u)$. Then, the results of Lemma \ref{lem:num_purifications_in_limit} still hold in the limit $u\rightarrow \infty$. In particular,
    \begin{enumerate}
        \item The limiting distribution of $M(u)$ is that of $M$,
        \begin{equation}
            \lim_{u\rightarrow \infty } P(M(u)\geq m) = P(M\geq m)=\left(\frac{\delta}{\beta + \delta}\right)^{m}.      
        \end{equation}
        \item Conditional on $M(u) = m$, the random variables  $(R_1,...,R_m,\tilde{C})$ converge in distribution to mutually independent and identically distributed  $\text{Exp}(\beta+\delta)$ as $u\rightarrow \infty$, i.e. 
        \begin{equation}
            \lim_{u \rightarrow \infty } P(X_1>x_1,...,X_m>x_m,\tilde{C}(u)>x_{m+1}|M(u)= m) = \prod_{i=1}^{m+1}  e^{-(\beta+\delta)x_i},
        \end{equation}
     \end{enumerate}
\label{cor:num_purifications_limit}
\end{corollary}
\begin{proof}
    \begin{enumerate}
        \item Making use of Corollary \ref{cor:limit_sum_random_variables}, we have 
        \begin{align*}
            \lim_{u\rightarrow \infty } P(M(u) \geq m) =\lim_{u\rightarrow \infty } P(C(u)>S_m) = P(C>S_m) = \left(\frac{\delta}{\beta + \delta}\right)^{m}.
        \end{align*}
        \item 
        One may use exactly the same arguments as were used to obtain (\ref{eqn:current_lifetime_in_limit_step_2}), only replacing $C$ with $C(u)$ and $M$ with $M(u)$, to show that
        \begin{multline*}
    P\left( R_1 > x_1,\dots,\: R_m>x_m, \tilde{C}(u)> x_{m+1} \cap M(u)=m \right) \\ = \Bigg[ P\left(C(u)\geq x_{m+1}+\sum_{i=1}^m R_i \Big|  R_1 > x_1 , ... , R_{m+1} > x_{m+1} \right) \\  -P\left(C(u)\geq \sum_{i=1}^{m+1} R_i \Big| R_1>x_1,..., R_{m+1}>x_{m+1}\right) \Bigg]  \prod_{i=1}^{m+1}P\left(R_1 > x_i\right).  
\end{multline*}
By Corollary \ref{cor:limit_sum_random_variables}, in the limit $u\rightarrow \infty$ this satisfies
\begin{multline*}
    \lim_{u\rightarrow \infty} P(R_1>x_1,....,R_m>x_m, \tilde{C}(u)>x_{m+1} 
    \; \cap \; M(u)= m) \\ = \Bigg[ P\left(C\geq x_{m+1}+\sum_{i=1}^m R_i \Big|  R_1 > x_1 , ... , R_{m+1} > x_{m+1} \right) \\  -P\left(C\geq \sum_{i=1}^{m+1} R_i \Big| R_1>x_1,..., R_{m+1}>x_{m+1}\right)\Bigg]\prod_{i=1}^{m+1}P\left(R_1 > x_i\right) = (\ref{eqn:current_lifetime_in_limit_step_2}).
\end{multline*}
%Then, 
%\begin{align*}
%    \lim_{u\rightarrow \infty} P(R_1>x_1,....,R_m>x_m,\: &\tilde{C}(u)>x_{m+1} 
%    \; \cap \; M(u)= m)  \\ &= P(R_1>x_1,....,R_m>x_m, \tilde{C}>x_{m+1} 
%    \; \cap \; M= m) \\ &= e^{-(\beta+\delta)x_1} ... e^{-(\beta+\delta)x_m} e^{-(\beta+\delta)x_{m+1}} \cdot P\left(M = m \right),
%\end{align*}
%by Lemma \ref{lem:num_purifications_in_limit}. 
It then follows that
\begin{align*}
    \lim_{u\rightarrow \infty} P(R_1>x_1,....,&R_m>x_m, \tilde{C}(u)>x_{m+1} 
    \Big| M(u)= m) \\ &= \lim_{u\rightarrow \infty} \frac{P(R_1>x_1,....,R_m>x_m, \tilde{C}(u)>x_{m+1} 
    \cap M(u)= m)}{P(M(u)=m)} \\ &= \frac{P(R_1>x_1,....,R_m>x_m, \tilde{C}>x_{m+1} 
    \cap M= m)}{P(M=m)} = \prod_{i=1}^{m+1}  e^{-(\beta+\delta)x_i},
\end{align*}
by Lemma \ref{lem:num_purifications_in_limit}. 
\end{enumerate}
\end{proof}
For the case when $C(u)$ is the current lifetime of simplified 1G1B, the random variable $(R_1,\dots,R_m,\tilde{C}(u))$ by definition has the same distribution as $\vec{T}(u)$.
Recall that $\vec{T}(u)$ contains the times spent in each purification level leading up to the current one at time $u$ in 1G1B (Definition \ref{def:T(t)}).
This leads to the following results.
\begin{corollary}
    Conditional on $s(t)=i$, $\vec{T}(t)$ converges in distribution to $(Q_0,\dots,Q_i)$ as $t\rightarrow \infty$, where the $Q_j$ are i.i.d. random variables with $Q_0\sim \mathrm{Exp}(\beta + \delta)$.
    \label{cor:limiting_dist_T(t)}
\end{corollary}

We now continue with the formulae for the performance metrics. The availability in the 1G1B system was given in Proposition \ref{prop:1G1B_stationary_distribution} and the average consumed fidelity was given in Theorem \ref{thm:average_fidelity_formula}. Next, we prove both of them.
\begin{proof}[Proof of Proposition \ref{prop:1G1B_stationary_distribution}]
From Corollary \ref{cor:limit_delayed_1G1B_simplified}, we see that 
\begin{align*}
    A = \lim_{t\rightarrow \infty} P(s(t) = \neg \emptyset) = \frac{\lambda}{\lambda+\mu+\lambda q (1-p)}.
\end{align*}
Further, for $i\geq 0$
\begin{align*}
    P(s(t)=i) &= P(s(t)=i|s(t)\neq\emptyset)\cdot P(s(t)\neq \emptyset).  
\end{align*}
Letting $C(t)$ denote the current lifetime of simplified 1G1B at time $t$, and $M(t)$ denote the number of purifications that have occurred within this time, by Corollary \ref{cor:num_purifications_limit} it follows that
\begin{align*}
    P(s(t)=i) = P(M(t)=i)\cdot P(s(t)\neq \emptyset) \rightarrow P(M=i)\cdot A 
\end{align*}
as $t\rightarrow \infty$. Recalling the distribution of $M$ as found in Lemma \ref{lem:num_purifications_in_limit}, we obtain 
\begin{align*}
\lim_{t\rightarrow \infty} P(s(t)=i) &= \left(\frac{\lambda q p}{\mu + \lambda q }\right)^i \cdot \frac{\mu + \lambda q (1-p)}{\mu + \lambda q } \cdot A \\ &= \frac{\lambda^{i+1}q^ip^i}{(\mu + \lambda q)^{i+1}}\cdot \frac{\mu + \lambda q (1-p)}{\lambda+\mu+\lambda q(1-p)}.
\end{align*}
We note that this result can also be derived with the global balance equations of a CTMC. Here, we chose to use the derivation with renewal theory since it offers a more general formula for the availability (see (\ref{eqn:limit_p(t)_delayed})) and ties in more neatly with the derivation of the formula for the average consumed fidelity, as we will see below.
\end{proof}
The following proposition will be helpful in the proof of Theorem \ref{thm:average_fidelity_formula} (formula for average consumed fidelity).
\begin{proposition} 
    Let $\{p_i(t)\}_{i\geq 0}$ and $\{ e_i(t) \}_{i\geq 0}$ be such that for all $i$,  $\lim_{t\rightarrow \infty} p_i(t)=\pi_i$ and $\lim_{t\rightarrow \infty } e_i(t) = c_i$. Suppose also that for all $t$, $0 \leq e_i(t) \leq 1$, $0 \leq p_i(t) \leq 1$ and $\sum_{i=0} p_i(t) = 1$. Then
    \begin{equation}
        \lim_{t\rightarrow \infty } \sum_{i=0}^{\infty} e_i(t)p_i(t)  = \sum_{i=0}^{\infty} c_i \pi_i.
        \label{eqn:sum_convergence}
    \end{equation}
    \label{prop:convergence_sum}
    \end{proposition}
        \begin{proof}[Proof of proposition \ref{prop:convergence_sum}] To show (\ref{eqn:sum_convergence}), it suffices to show that for any $\epsilon >0$, there exists a $T$ such that for all $t>T$,
    \begin{equation}
        \Bigg| \sum_{i=0}^{\infty} e_i(t)p_i(t) - \sum_{i=0}^{\infty} c_i \pi_i \Bigg| < \epsilon .
    \end{equation}

We firstly bound the sum using the triangle inequality,
\begin{align}
    \Bigg| \sum_{i=0}^{\infty} e_i(t) p_i(t) - \sum_{i=0}^{\infty} c_i \pi_i \Bigg| &= \Bigg| \sum_{i=0}^{\infty} e_i(t) (p_i(t) - \pi_i) + (e_i(t)-c_i)\pi_i \Bigg| \nonumber \\ &\leq \underbrace{\sum_{i=0}^{\infty} e_i(t) |p_i(t)-\pi_i|}_{(A)} + \underbrace{\sum_{i=0}^{\infty} |e_i(t)-c_i|\pi_i}_{(B)} \label{eqn:split_sum_triangle_ineq} .
\end{align}

We then show that $(A)\rightarrow 0$ and $(B)\rightarrow 0$ as $t\rightarrow\infty$. We firstly show that
\begin{equation}
    \lim_{t\rightarrow \infty} \sum_{i=0}^{\infty} \big| p_i(t) - \pi_i\big| = 0.
    \label{eqn:convergence_1}
\end{equation}
Note that, since $\sum_{i=0}^{\infty}p_i(t) = 1$, it follows that $\sum_{i=0}^{\infty} \pi_i = 1$. Then, choose $N$ such that 
    \begin{equation*}
        \sum_{i=0}^N \pi_i > 1-\frac{\epsilon}{2}
    \end{equation*}
and choose $T_1$ such that
\begin{equation*}
    \sum_{i=0}^{N} \big|p_i(t) - \pi_i \big| < \frac{\epsilon}{2} \text{, } \forall t>T_1
\end{equation*}
which is possible since the sum is finite. Then $\forall t>T_1$,
\begin{align}
    \Bigg| 1-\sum_{i=0}^{N}p_i(t) \Bigg| &= \Bigg| 1-\sum_{i=0}^{N}\left( \pi_i - (\pi_i-p_i(t))\right)\Bigg| \nonumber \\ &< 
 \Bigg| 1 - \sum_{i=0}^{N}\pi_i \Bigg| +  \sum_{i=0}^N |\pi_i - p_i(t) | < \frac{\epsilon}{2}+\frac{\epsilon}{2} = \epsilon. \label{eqn:sum_p_i(t)_from_N}
\end{align}
Now, choose $T_2$ such that $\forall t>T_2$,
\begin{equation*}
    |p_i(t) - \pi_i|<\frac{\epsilon}{N} \text{, }\forall i = 0,...,N
\end{equation*}
and let $T = \max \{T_1,T_2 \}$. Then, $\forall t>T,$ 
\begin{align*}
    \sum_{i=0}^{\infty}|p_i(t)-\pi_i| &= \sum_{i=0}^N |p_i(t) - \pi_i| + \sum_{i>N} |p_i(t) - \pi_i| \\ &< N\cdot \frac{\epsilon}{N} + \sum_{i>N}p_i(t) + \sum_{i>N}\pi_i \\ &< \epsilon + \epsilon + \frac{\epsilon}{2},
\end{align*}
from (\ref{eqn:sum_p_i(t)_from_N}). This suffices to show (\ref{eqn:convergence_1}). Combined with the fact that the $e_i$ are bounded, it follows that $(A)\rightarrow 0$. We now show that $(B)\rightarrow 0$, i.e. 
\begin{equation}
   \lim_{t\rightarrow \infty} \sum_{i=0}^{\infty} |e_i(t)-c_i|\pi_i = 0.
   \label{eqn:convergence_2}
\end{equation}
To show this, let $\epsilon>0$. Choose $N$ such that $\sum_{i=0}^{N}\pi_i>1-\epsilon $. Choose $T$ such that 
\begin{equation*}
    \sum_{i=0}^N|e_i(t)-c_i|< \epsilon, \; \forall \; t>T.
\end{equation*}
This is possible since the LHS is a finite sum. Then,  
\begin{align*}
    \sum_{i=0}^{\infty} |e_i(t)-c_i|\pi_i &= \sum_{i=0}^N |e_i(t)-c_i|\pi_i + \sum_{i>N} |e_i(t)-c_i|\pi_i 
    \\ &< \left(\sum_{i=0}^N |e_i(t)-c_i| \right) \cdot\left( \sum_{i=0}^N \pi_i \right) +  \sum_{i>N} \pi_i \\ &< \epsilon (1-\epsilon) + \epsilon \; \forall \; t>T,
\end{align*}
which shows (\ref{eqn:convergence_2}).
\end{proof}  
Combining these results ($(A)\rightarrow 0$ and $(B)\rightarrow 0)$ in (\ref{eqn:split_sum_triangle_ineq})) suffices to show Proposition \ref{prop:convergence_sum}.
We are now ready to prove Theorem \ref{thm:average_fidelity_formula} (formula for average consumed fidelity).
\begin{proof}[Proof of Theorem \ref{thm:average_fidelity_formula}]
We firstly expand out the average consumed fidelity (Definition \ref{def.avg_cons_fid}) as a sum by conditioning on the value of $s(t)$,
\begin{align}
    \overline F \coloneqq \mathbb{E}[F(t)|s(t)\neq \emptyset] &= \sum_{i=0}^{\infty} \mathbb{E}[F(t)|s(t)=i ] P\!\left(s(t)=i|s(t)\neq \emptyset \right) \nonumber \\ &= \frac{1}{P\left(s(t)\neq \emptyset \right)}\sum_{i=0}^{\infty} \mathbb{E}[F(t)|s(t)=i ] P\! \left(s(t)=i \right).
    \label{eqn:average_fidelity_expand_as_sum}
\end{align}
Recall that we are interested in the limit $t\rightarrow \infty$ of the above. Note that from Proposition \ref{prop:1G1B_stationary_distribution}, we know the limiting values of $P\left(s(t)\neq \emptyset \right)$ and $ P\left(s(t)=i \right)$. We now claim that
\begin{equation}
     \lim_{t\rightarrow \infty } \mathbb{E}[F(t)|s(t)=i ] = \mathbb{E} \left[F^{(i)}\left(Q_0,Q_1,...,Q_i \right) \right]
     \label{eqn:limit_conditional_fidelity_functions}
\end{equation}
where $Q_0,Q_1,\dots,Q_i$ are i.i.d. random variables with $Q_0\sim \text{Exp}(\mu+\lambda q)$, and $F^{(i)}$ is given in Definition \ref{def:composite_fidelity_fn}. We use the following result:
\begin{theorem}[Theorem 7.2.19 of \cite{grimmett2020}]
    Let $X_n$ be a sequence of random variables. Then, $X_n\rightarrow X$ in distribution if and only if  $\mathbb{E}[g(X_n)]\rightarrow \mathbb{E}[g(X)]$ for all bounded continuous functions $g$.
    \label{thm:convergence_expectation}
\end{theorem}
Recall that conditional on $s(t)=i$, we have $ F(t) = F^{(i)}\left(\vec{T}(t)\right)$ (from Definition \ref{def:F(t)}). As mentioned in Section \ref{sec:1G1B_system}, $F^{(i)}$ is a continuous and bounded function. Therefore,  (\ref{eqn:limit_conditional_fidelity_functions}) follows by combining Theorem \ref{thm:convergence_expectation} and Corollary \ref{cor:limiting_dist_T(t)}.

From Proposition \ref{prop:convergence_sum}, we therefore see that
\begin{align*}
    \lim_{t\rightarrow \infty }\mathbb{E}[F(t)|s(t)\neq \emptyset] &= \lim_{t\rightarrow \infty} \frac{1}{P\left(s(t)\neq \emptyset \right)}\sum_{i=0}^{\infty} \mathbb{E}[F(t)|s(t)=i ] P\left(s(t)=i \right) \\ &=\frac{1}{A}  \sum_{i=0}^{\infty} \lim_{t\rightarrow \infty}\mathbb{E}[F(t)|s(t)=i ] \lim_{t\rightarrow \infty}P\left(s(t)=i \right) \\ &= \frac{1}{A}\sum_{i=0}^{\infty} c_i \pi_i,
\end{align*}
where $c_i = \mathbb{E} \left[F^{(i)}\left(Q_0,Q_1,...,Q_i \right) \right]$ and $\pi_i = \lim_{t\rightarrow \infty} P(s(t)=i)$.
\end{proof}

\clearpage 

\section{Average consumed fidelity with a linear jump function}
In this Appendix we focus on linear jump functions. In \ref{app.range_a_b_linear_jump}, we provide bounds for the coefficients of a linear jump function.
In \ref{app:average_fidelity_derivation_linear_jump}, we first prove Proposition \ref{prop:formula_fidelity_functions_linear_sym_jump}.
Then, we use that Proposition to derive the average consumed fidelity in a 1G1B system that uses a pumping protocol with a linear jump function, which we denote by $\overline F_{\scriptscriptstyle \mathrm{linear}}$ (i.e., we show Lemma \ref{lem:formula_average_fidelity_linear_jump}).
We also show that $\overline F_{\scriptscriptstyle \mathrm{linear}}$ is monotonic in the probability of pumping $q$ and the probability of successful pumping $p$ (Proposition \ref{lem:properties_average_fidelity_linear_jump}).
\reb{Lastly, in \ref{app.noise_threshold}, we discuss in which situations $\overline F_{\scriptscriptstyle \mathrm{linear}}$ is monotonically increasing in $q$, and we compute the noise threshold (\ref{eq.noise_threshold}) discussed in Section \ref{sec:operating_regimes} (above this threshold, any purification is better than no purification).}

\subsection{Bounds on the parameters of a linear jump function} \label{app.range_a_b_linear_jump}
\begin{proposition}
    Consider a jump function that is linear with the fidelity of one of the input states, i.e.,
\begin{equation}
    J(F,\rho) = a(\rho)F + b(\rho),
\end{equation}
where $F$ is the fidelity of one of the input states and $\rho$ is the second input state.
Then, the coefficients $a(\rho)$ and $b(\rho)$
must satisfy $$0\leq a(\rho) \leq 1 \;\;\;\;\mathrm{and}\;\;\;\; \frac{1}{4}\big(1-a(\rho) \big) \leq b(\rho) \leq 1-a(\rho).$$
\label{prop:range_coeffs_linear_jump}
\end{proposition}

\begin{proof}
    First, we require $J(F,\rho)\leq 1$, which is equivalent to $b \leq 1-a$.
    We also require $J(F,\rho)\geq 1/4$, which leads to $b\geq (1-a)/4$.
    By imposing that the upper bound on $b$ has to be larger than the lower bound, we find that $a\leq 1$.
Finally, since we want jump functions that increase with increasing $F$, we want $a\geq0$.
\end{proof}

\subsection{Derivation of $\overline F_{\scriptscriptstyle \mathrm{linear}}$}
\label{app:average_fidelity_derivation_linear_jump}
\begin{proof}[Proof of Proposition \ref{prop:formula_fidelity_functions_linear_sym_jump}]

Here, we consider a 1G1B system with $J(F,\rho_\mathrm{new}) = aF+b$ and $F^{(0)}(t_0)=D_{t_0}(F_\mathrm{new})$, where $F_\mathrm{new}$ is the fidelity of the state $\rho_\mathrm{new}$.
Our goal is to find an analytical solution for the fidelity of the entangled link after $i$ consecutive successful purifications, $F^{(i)}(t_0,...,t_{i-1},t_i)$. The time passed between purification $j$ and $j+1$ is given by $t_j$. After the $i$-th purification the system spent time $t_i$ without any transitions (i.e., no purification or consumption events).
We show in this proof that $F^{(i)}$ is given by
    \begin{equation}\label{eq.Fi_appendix}
        F^{(i)}(t_0,...,t_{i-1},t_i) = \frac{1}{4} + \sum_{j=0}^i m_j^{(i)} e^{-\Gamma(t_j+t_{j+1}...+t_i)}
    \end{equation}
    where the constants $m^{(i)}_j$ are given by $m_0^{(0)}=F_\mathrm{new}-\frac{1}{4}$, and 
    \begin{equation*}
    m_j^{(i)} = 
    \begin{cases}
    a^{i-j}
    \left( \frac{a}{4} + b - \frac{1}{4} \right), \text{ if } j>0, \\ 
    a^i \left(F_\mathrm{new} - \frac{1}{4}\right)\text{ if } j=0.
    \end{cases} 
    \end{equation*}
    for $i>0$.

We proceed by induction. For $i=0$, we have 
\begin{gather}
    F^{(0)}(t_0) = D_{t_0}(F_{\text{new}}) = e^{-\Gamma t_0}\left(F_{\text{new}} - \frac{1}{4}\right) + \frac{1}{4},
\end{gather}
from which we see that $m_0^{(0)}=F_{\text{new}}-\frac{1}{4}$.
If we assume that (\ref{eq.Fi_appendix}) is true for some $i$, using the recursive relation from (\ref{eqn:composite_fidelity_fn}) we can show that (\ref{eq.Fi_appendix}) is also true for $i+1$:
\begin{align*}
    F^{(i+1)}(t_0,...,t_i) &= D_{t_{i+1}}\left( J(F^{(i)},\rho)\right) \\ &= D_{t_{i+1}}\left(aF^{(i)}+b \right) \\  &= e^{-\Gamma t_{i+1}}\left( aF^{(i)} +b - \frac{1}{4}\right) +\frac{1}{4} \\ &= \frac{1}{4} + \left(\frac{a}{4}+b-\frac{1}{4}\right)e^{-\Gamma t_{i+1}} + \sum_{j=0}^i am_j^{(i)}e^{-\Gamma (t_j+...+t_i + t_{i+1})},
\end{align*}
from which it follows that
\begin{align*}
    m_{j}^{(i+1)} &= am_j^{(i)} \;\; (0 \leq j \leq i)\\
    m_{i+1}^{(i+1)} &= \frac{a}{4}+b-\frac{1}{4}
\end{align*}
Then, by the inductive assumption, $m_{0}^{(i+1)} = a^{i+1}\left(F_{\text{new}} - \frac{1}{4}\right)$, and $m_{j}^{(i+1)} = a^{i+1 - j}\left(\frac{a}{4}+b-\frac{1}{4}\right)$ for $j>0$. 
\end{proof}

\begin{proof}[Proof of Lemma \ref{lem:formula_average_fidelity_linear_jump}]
Here we consider a 1G1B system with $J(F,\rho_\mathrm{new})=aF+b$ and $F^{(0)}(t_0)=D_{t_0}(F_\mathrm{new})$, where $F_\mathrm{new}$ is the fidelity of the state $\rho_\mathrm{new}$. Our goal is to find a closed-form solution for the average fidelity after $i\geq 0$ purification rounds, $c_i$, and for the average consumed fidelity, $\overline F_{\scriptscriptstyle \mathrm{linear}}$.

We defined $c_i$ as the average value of $F^{(i)}$ (see (\ref{eqn:formula_conditional_average_fidelity})).
Using the expression for $F^{(i)}$ from Proposition \ref{prop:formula_fidelity_functions_linear_sym_jump} (also given in (\ref{eq.Fi_appendix})), we can evaluate $c_i$ as follows
\begin{equation*}
\begin{split}
    c_i :=&
    \int_{0}^{\infty} \dd t_i f_{\alpha}(t_i) \: ... \int_0^{\infty} \dd t_0 f_{\alpha}(t_0)
    F^{(i)}(t_0,...,t_{i-1},t_i)\\
    =& \int_{0}^{\infty} \dd t_i f_{\alpha}(t_i) \: ... \int_0^{\infty} \dd t_0 f_{\alpha}(t_0) \left[\frac{1}{4} + \sum_{j=0}^i m_j^{(i)} e^{-\Gamma(t_j+...+t_{i-1}+t_i)} \right] \\
    =& \frac{1}{4} + \sum_{j=0}^i m_j^{(i)} \left( \frac{\alpha}{\alpha+\Gamma}\right)^{i-j+1} \\
    =& \frac{1}{4} +  \left(F_\mathrm{new}-\frac{1}{4}\right)\cdot a^i  \gamma^{i+1} + \gamma \left( \frac{a}{4} + b - \frac{1}{4} \right)\sum_{j=1}^{i} a^{i-j} \gamma^{i-j},
\end{split}
\end{equation*}
where $\alpha = \mu+\lambda q$, $f_\alpha(t_i) = \alpha e^{-\alpha t_i}$ (since the times $t_i$ are exponentially distributed with rate $\alpha$), $\gamma = \alpha/(\alpha+\Gamma)$.
Using the fact that this is a geometric series, we may now obtain a closed-form solution for $c_i$:
\begin{equation}
c_i = \frac{1}{4} +  \left(F_\mathrm{new}-\frac{1}{4}\right)\cdot a^i  \gamma^{i+1} + \gamma \left( \frac{a}{4} + b - \frac{1}{4} \right)  \frac{1-a^i \gamma^i}{1-a\gamma}.
\label{eqn:conditional_average_fidelity_linear_jump_appendix}
\end{equation}

The final formula for the average fidelity may then be computed with the results of Proposition \ref{prop:1G1B_stationary_distribution} and Theorem \ref{thm:average_fidelity_formula} as 
\begin{equation}
\begin{split}
    \overline F_{\scriptscriptstyle \mathrm{linear}} &= \lim_{t\rightarrow \infty}\mathbb{E}\left(F(t)| s(t) \neq \emptyset\right) = \frac{1}{1-\pi_{\emptyset}}\sum_{i=0}^\infty c_i \pi_i\\
    &= \frac{1}{4} + \frac{\gamma}{1-a\gamma} \cdot \left( \frac{a}{4} + b - \frac{1}{4} \right) +  \frac{\gamma}{(1-\pi_{\emptyset})} \cdot\left( F_\mathrm{new}-\frac{1}{4} - \frac{\frac{a}{4} + b - \frac{1}{4}}{1-a\gamma}\right) \sum_{i=0}^{\infty} \pi_i (a\gamma)^i,
\label{eqn:average_fidelity_evaluation_linear_jump}
\end{split}
\end{equation}
where the constant terms are no longer in the sum since 
\begin{equation*}
   \frac{1}{1-\pi_{\emptyset}} \sum_{i=0}^{\infty} \pi_i = 1,
\end{equation*}
by the normalisation of the steady state distribution. Recalling the distribution of $\pi$ from Proposition \ref{prop:1G1B_stationary_distribution}, we may evaluate the sum as a geometric series,
\begin{gather*}
    \sum_{i=0}^{\infty} \pi_i(a\gamma)^i = \frac{\lambda}{\mu + \lambda q}\sum_{i=0}^{\infty} \left( \frac{\lambda q p a \gamma}{\mu+\lambda q}\right)^i\pi_{\emptyset} \\ = \frac{\lambda}{\mu + \lambda q} \cdot \frac{1}{1-\frac{\lambda q p a \gamma}{\mu+ \lambda q}}\pi_{\emptyset} \\
    = \frac{\lambda}{\mu + \lambda q - \lambda q p a \gamma} \pi_{\emptyset}.
\end{gather*}
We may now substitute this into (\ref{eqn:average_fidelity_evaluation_linear_jump}) to obtain a closed-form solution for the average fidelity,
\begin{equation}%\label{eq.avgfid_linear_jump}
\begin{split}
\overline F_{\scriptscriptstyle \mathrm{linear}} &= \frac{1}{4} + \frac{\gamma}{1-a\gamma} \cdot \left( \frac{a}{4} + b - \frac{1}{4} \right) +  \gamma \cdot\left( F_\mathrm{new}-\frac{1}{4} - \frac{\frac{a}{4} + b - \frac{1}{4}}{1-a\gamma}\right) \frac{\lambda}{\mu + \lambda q - \lambda q p a\gamma} \frac{\pi_{\emptyset}}{1-\pi_{\emptyset}} \\ 
&= \frac{1}{4} + \frac{\gamma}{1-a\gamma} \cdot \left( \frac{a}{4} + b - \frac{1}{4} \right) +  \gamma \cdot\left( F_\mathrm{new}-\frac{1}{4} - \frac{\frac{a}{4} + b - \frac{1}{4}}{1-a \gamma}\right) \frac{\mu + \lambda q(1-p)}{\mu + \lambda q - \lambda q p a \gamma}\\
&= \frac{\frac{1}{4}\Gamma + b \lambda q p + F_\mathrm{new} \Big(\mu + \lambda q (1 - p)\Big)}{\Gamma + \mu + \lambda q (1-pa)},
\end{split}
\end{equation}
which completes the closed-form solutions for our two performance metrics in this set-up (in the last step we used Mathematica to simplify the expression).
\end{proof}
\begin{proof}[Proof of Proposition \ref{lem:properties_average_fidelity_linear_jump}]
To show $(a)$, we compute the partial derivative of the average consumed fidelity with respect to $q$:
\begin{equation}\label{eq.partial_q}
        \frac{\partial \overline F_{\scriptscriptstyle \mathrm{linear}}}{\partial q}
        = \lambda \frac{\Gamma \big( 4F_\mathrm{new}(1-p) + (4b+a)p -1 \big) 
        + 4\mu p \big( b-F_\mathrm{new}(1-a) \big)}
        {4\big(\Gamma + \mu + \lambda q (1-ap)\big)^2}.
\end{equation}
Since the sign of the derivative does not depend on $q$, we conclude that $\overline F_{\scriptscriptstyle \mathrm{linear}}$ is monotonic in $q$.

To show $(b)$, we proceed similarly:
\begin{equation}
        \frac{\partial \overline F_{\scriptscriptstyle \mathrm{linear}}}{\partial p}
        = \lambda q \frac{4(b-F_\mathrm{new})(\Gamma+\mu+\lambda q) + a \big( \Gamma+4F_\mathrm{new}(\mu+\lambda q) \big)}
        {4\big(\Gamma + \mu + \lambda q (1-ap)\big)^2}.
\end{equation}
Since the sign of this derivative does not depend on $p$, we conclude that $\overline F_{\scriptscriptstyle \mathrm{linear}}$ is monotonic in $p$.

%\agi{I have commented the proof for $J=(1-2a)+aF_1+aF_2$.}
% To show $(a)$, we compute the partial derivative
% \begin{equation}
%     \begin{split}
%         \frac{\partial \overline F}{\partial q}
%         &= \lambda \frac{\Gamma (u-(1-pa)) + \mu (u-4F_\mathrm{new}\mu(1-pa))}{4\big(\Gamma + \mu + \lambda q (1-pa)\big)^2}\\
% &\stackrel{a}{\geq} \lambda\mu \frac{u-4F_\mathrm{new}\mu(1-pa)}{4\big(\Gamma + \mu + \lambda q (1-pa)\big)^2}\\
% &= \lambda\mu p \frac{(1-F_\mathrm{new})(1-2a)}{\big(\Gamma + \mu + \lambda q (1-pa)\big)^2}\\
% &\stackrel{b}{\geq}0,
%     \end{split}
% \end{equation}
% where $u \equiv 4(1-2a)p + 4F_\mathrm{new}(1-p(1-a))$, and with the following steps: $(i)$ we use $u\geq 1/2$ (it can be shown using $F_\mathrm{new}\geq 1/4$, $a<1/2$, and $p\leq 1$) and $1-pa<1/2$ (it can be shown using $a<1/2$, and $p\leq 1$); $(ii)$ we use $1-2a\geq 0$ and $F_\mathrm{new}\leq 1$.
\end{proof}

\subsection{Noise threshold}\label{app.noise_threshold}
\reb{In the previous Section, we showed that $\overline F_{\scriptscriptstyle \mathrm{linear}}$ is monotonic in $q$ and $p$ (Proposition \ref{lem:properties_average_fidelity_linear_jump}).
Nevertheless, note} that $\overline F_{\scriptscriptstyle \mathrm{linear}}$ can be monotonically increasing or decreasing in $q$ and in $p$ depending on the values of the other parameters.
For a pumping protocol with a good enough jump function, $\overline F_{\scriptscriptstyle \mathrm{linear}}$ becomes increasing in $q$. A sufficient condition is for the jump function to satisfy $b\geq F_\mathrm{new}(1-a)$, as we show next. The partial derivative with respect to $q$ \reb{from (\ref{eq.partial_q})}  can be written as follows:
\begin{equation}
        \frac{\partial \overline F_{\scriptscriptstyle \mathrm{linear}}}{\partial q}
        = \frac{\lambda}{x^2} \big( \Gamma y 
        + 4\mu p z \big),
\end{equation}
where
$x = 2(\Gamma + \mu + \lambda q (1-ap))$,
$y = 4F_\mathrm{new}(1-p) + (4b+a)p -1$,
and $z = b-F_\mathrm{new}(1-a)$.
Using the fact that $b \geq (1-a)/4$, we find that $y\geq0$.
A sufficient condition for the partial derivative to be positive is that $z \geq 0$, i.e., if $b \geq F_\mathrm{new}(1-a)$, then the average consumed fidelity is monotonically increasing in $q$.
Moreover, we can conclude that, if the noise is above certain threshold ($\Gamma > -4\mu p z/y$), the derivative is positive and the pumping is always beneficial, even if it succeeds with a very small probability.

\clearpage 

\section{Bounds for the performance of bilocal Clifford protocols}\label{app.linear_bounds}
In this Appendix, we find bounds to the output fidelity and the probability of success of 2-to-1 purification protocols.
In particular, we show Lemma \ref{lem:bilocal_clifford_bounds}, where upper and lower bounds on the jump function and the success probability of any bilocal Clifford protocol, taking as input a Werner state $\rho_\mathrm{\scriptscriptstyle W}$ and a Bell-diagonal state $\rho_\mathrm{\scriptscriptstyle BD}$. We define the fidelity of a state $\rho$ as $F(\rho,\ket{\phi^+}) = \bra{\phi^+}\rho\ket{\phi^+}$, where $\ket{\phi^+} = (\ket{00}+\ket{11})/\sqrt{2}$ is one of the Bell states. We find the bounds for a system with the following restrictions. 
\begin{itemize}
    \item We consider 2-to-1 purification protocols, i.e., protocols that take two bipartite entangled states as input and output a single bipartite state. This allows us to use these bounds directly for the analysis of the 1G1B system.
    \item We restrict the pumping protocols to bilocal Clifford protocols \cite{Dehaene2003, Jansen2022}, which are a well-known type of purification scheme. We provide more details about this type of protocol in \ref{app.bicliff_subsec}.
    \item We assume that one of the input states is a Werner state (in the 1G1B system, this is the state in the good memory, which suffers from depolarising noise) and the other input state is Bell-diagonal (in the 1G1B system, this is the state generated via heralded entanglement generation and placed in the bad memory).
    Mathematically, the input states can be written, respectively, as
    \begin{equation*}
        \rho_\mathrm{\scriptscriptstyle W} = F \ketbra{\phi^+} + \frac{1-F}{3} \ketbra{\psi^+} + \frac{1-F}{3} \ketbra{\psi^-} + \frac{1-F}{3} \ketbra{\phi^-},
    \end{equation*}
    \begin{equation*}
        \rho_\mathrm{\scriptscriptstyle BD} = F_\mathrm{\scriptscriptstyle BD} \ketbra{\phi^+} + \lambda_1 \ketbra{\psi^+} + \lambda_2 \ketbra{\psi^-} + \lambda_3 \ketbra{\phi^-},
    \end{equation*}
    with $F, F_\mathrm{\scriptscriptstyle BD}, \lambda_1, \lambda_2, \lambda_3 \in [0,1]$ subjected to the normalization constraint $F_\mathrm{\scriptscriptstyle BD} + \lambda_1 + \lambda_2 + \lambda_3 = 1$, and with the Bell states defined as
    \begin{equation*}
        \ket{\phi^+} = \frac{\ket{00}+\ket{11}}{\sqrt{2}}, \;
        \ket{\psi^+} = \frac{\ket{01}+\ket{10}}{\sqrt{2}}, \;
        \ket{\psi^-} = \frac{\ket{01}-\ket{10}}{\sqrt{2}}, \;
        \ket{\phi^-} = \frac{\ket{00}-\ket{11}}{\sqrt{2}}.
    \end{equation*}
    Note that any bipartite state can be brought to Bell-diagonal form while preserving the fidelity by means of twirling (adding extra noise) \cite{Bennett1996a,Horodecki1999}.
    \item We only consider newly generated states with fidelity to some Bell state larger than $1/2$, i.e., we assume $F_\mathrm{\scriptscriptstyle BD} > 1/2$ (note that $F_\mathrm{\scriptscriptstyle BD}>1/2$ is equivalent to $\lambda_i>1/2$ for some $i$, since the states are equivalent upon some Pauli corrections). As shown in \ref{subsec.additional_proofs}, this is a necessary and sufficient condition for the existence of entanglement (otherwise, the state is not useful for purification).
    \item We assume the Werner state has fidelity $F> 1/4$, since the good memory is initially occupied with a state with fidelity larger than $1/2$, and this fidelity can decay at most to $1/4$ due to depolarising noise (see Definition \ref{def:depol_noise}).
\end{itemize}

In \ref{app.bicliff_subsec}, we provide a formal definition of bilocal Clifford protocols.
Then, in \ref{app.linearbounds_subsec}, we prove Lemma \ref{lem:bilocal_clifford_bounds}, where bounds are found for the jump function and success probability of bilocal Clifford protocols in a system with the above restrictions.

\subsection{Bilocal Clifford protocols}\label{app.bicliff_subsec}
Bilocal Clifford protocols \cite{Dehaene2003, Jansen2022} take $n$ bipartite states as input and outputs a single bipartite state.
They consist of the following steps:
\begin{enumerate}
    \item $C^T \otimes C^\dagger$ is applied to the state, where $C$ is some Clifford circuit. A Clifford circuit consists of Hadamard gates, phase gates S, and CNOTs \cite{Gottesman1998a, Gottesman1998}. If the state is held by two separate parties, one of them applies $C^T$ and the other one applies $C^\dagger$.
    \item All of the qubit pairs except one are measured (in a 2-to-1 protocol, one qubit pair is measured and the other one is kept).
    \item Depending on the parity of the measurement outcomes, success or failure is declared. Local unitaries may be performed after a success.
\end{enumerate}

One of the main advantages of bilocal Clifford protocols is that they are relatively simple to execute in practice, since they involve a basic set of gates. Additionally, any stabilizer code $C$ can be mapped to a bilocal Clifford circuit that applies $C^T \otimes C^\dagger$, allowing the analysis of bilocal Clifford circuits from a quantum error-correction perspective \cite{Goodenough2023}.
This type of protocol also includes well-known purification protocols, such as DEJMPS \cite{Deutsch1996}.
% We don't know good non-bilocal Clifford protocols

\subsection{Linear bounds to the performance of bilocal Clifford protocols}\label{app.linearbounds_subsec}
In this Appendix, we prove Lemma \ref{lem:bilocal_clifford_bounds}, where bounds on the jump function and success probability of every bilocal Clifford protocol are found. Consider pumping two states of the form
\begin{align}
        \rho_{\mathrm{\scriptscriptstyle W}} &= F \ket{\Phi^+}\bra{\Phi^+} + \frac{1-F}{3} \left( \ket{\Psi^+}\bra{\Psi^+} +\ket{\Psi^-}\bra{\Psi^-} +\ket{\Phi^-}\bra{\Psi^-}\right) \label{eqn:werner_state}\\ \rho_{\mathrm{\scriptscriptstyle BD}} &= F_{\mathrm{\scriptscriptstyle BD}} \ket{\Phi^+}\bra{\Phi^+} + \lambda_1 \ket{\Psi^+}\bra{\Psi^+} +\lambda_2 \ket{\Psi^-}\bra{\Psi^-} +\lambda_3 \ket{\Phi^-}\bra{\Psi^-}. \label{eqn:bell_diag_state}
\end{align}
Using the methods from \cite{Jansen2022}, we can find the analytical expressions for the output fidelity and success probability for every bilocal Clifford protocol. The restriction to bilocal Clifford protocols and Bell-diagonal states allows us to do this enumeration of analytical functions efficiently \cite{Jansen2022, Goodenough2023}. There are only seven protocols that provide a unique combination of $J$ and $p$, as shown in Table \ref{tab:7_protocols}. We refer to the $i$-th jump function and success probability as $J_i(F,\rho_{\mathrm{\scriptscriptstyle BD}})$ and $p_i(F,\rho_{\mathrm{\scriptscriptstyle BD}})$, for $i=1,...,7$.

\begin{table}[h!]
    \centering
    \caption{The jump function and success probability for all 2-1 bilocal Clifford protocols, with input states given in (\ref{eqn:werner_state}) and (\ref{eqn:bell_diag_state}) .
    }\label{tab:7_protocols}
    \vspace{-2mm} % Adjust the height of the space between caption and tabular
\begin{tabular}{|c|c|c|} 
 \hline
 Protocol & Jump function & Success probability \\ \hline
 1 & $\frac{(4\lambda_1 + 3 \lambda_2 + 3 \lambda_3 - 3)F - \lambda_1}{(4\lambda_2+4\lambda_3-2)F - \lambda_2 - \lambda_3 - 1}$ & $ \frac{2}{3}  ( 1 - 2\lambda_2 -2\lambda_3 )F+ \frac{1}{3} (1+\lambda_2+\lambda_3)$\\ \hline
 2 & $\frac{(3\lambda_1 + 4 \lambda_2 + 3 \lambda_3 -3)F - \lambda_2}{(4\lambda_1+4\lambda_3-2)F - \lambda_1 - \lambda_3 -1}$ & $ \frac{2}{3}  ( 1 - 2\lambda_3 -2\lambda_1 )F+ \frac{1}{3} (1+\lambda_3+\lambda_1)$ \\ \hline
 3 & $\frac{(3\lambda_1 + 3 \lambda_2 + 4 \lambda_3 -3)F - \lambda_3}{(4\lambda_1+4\lambda_2-2)F - \lambda_1 - \lambda_2 -1}$ & $ \frac{2}{3}  ( 1 - 2\lambda_1 -2\lambda_2 )F+ \frac{1}{3} (1+\lambda_1+\lambda_2)$\\ \hline
 4 & $F$ & $F_\mathrm{\scriptscriptstyle BD} + \lambda_1$ \\ \hline
 5 & $F$ & $F_\mathrm{\scriptscriptstyle BD} + \lambda_2$\\ \hline
 6 & $F$ & $F_\mathrm{\scriptscriptstyle BD} + \lambda_3$\\ \hline
 7 & $F_{\text{BD}}$ & $\frac{2}{3} F + \frac{1}{3}$ \\ \hline
\end{tabular}
\end{table}

We see that for these particular input states, $J_4$, $J_5$ and $J_6$ produce no change in the fidelity of $\rho_{\mathrm{\scriptscriptstyle W}}$. They also have a non-unity success probability. It would therefore be advantageous to simply perform no action instead of attempting Protocols $4$-$6$. Similarly, $J_7$ assumes the fidelity of the Bell-diagonal state, which is the same change as performing replacement. Since replacement can be achieved with probability one, it does not make sense to perform Protocol 7. Therefore, the only remaining `non-trivial' protocols are Protocols 1-3. In the following, we therefore find bounds for the jump function of Protocols 1-3. Notice that there is symmetry in the $\lambda_i$: $J_2$ and $p_2$ can be obtained by permuting $(\lambda_1,\lambda_2,\lambda_3)$ in $J_1$ and $p_1$, and similarly for $J_3$ and $p_3$. 

In the following, we show Lemma \ref{lem:bilocal_clifford_bounds}.

\begin{figure}[t]
  \centering
  \includegraphics[width=0.7\linewidth]{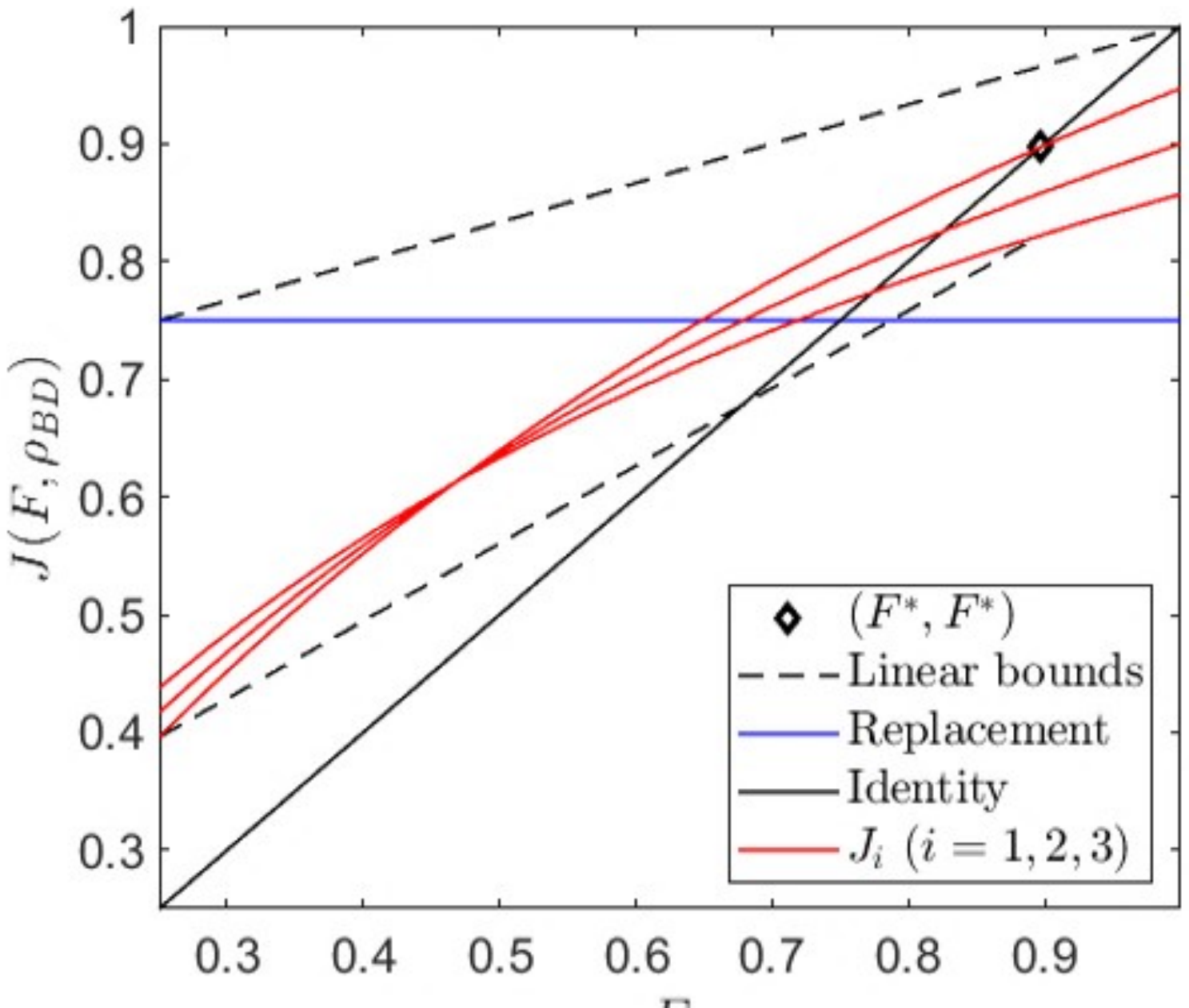}
  \caption{Linear bounds for the jump function of bilocal Clifford protocols (black dashed lines). The jump functions shown are $J_1$-$J_3$ (red lines), $J_4$-$J_6$ (identity operation, black line), and $J_7$ (probabilistic replacement, blue line). $F^*$ is the highest fidelity achievable by pumping a low-fidelity Werner state with the fixed Bell-diagonal state $\rho_{\mathrm{\scriptscriptstyle BD}}$. The lower bound holds in the range $[1/4,F^*]$. The upper bound holds in the range $[1/4,1]$. Here, $F_{\mathrm{\scriptscriptstyle BD}}=0.75$ and $\rho_{\mathrm{\scriptscriptstyle BD}} = (0.75,0.125,0.833,0.0417)$.}
  \label{fig:linear_bounds}
\end{figure}

\begin{proof}[Proof of Lemma \ref{lem:bilocal_clifford_bounds}]
We firstly show the linear lower bound (i.e. the formulae given in (\ref{eq.lower_bounds_ab})). We assume that $\lambda_1 \geq \lambda_2 \geq \lambda_3$. Then, by symmetry in the $\lambda_i$, one may retrieve the bound by setting $ \lambda_{\min} = \lambda_{3}$ and $\lambda_{\max} =\lambda_1$. In order to show this bound, we make use of the following collection of results. It is important to note that when showing all of the following results, $\rho_{\mathrm{\scriptscriptstyle BD}}$ is fixed. 
\begin{enumerate}
 \item \textbf{Proposition \ref{prop:intersection_jump_functions}, Corollory \ref{cor:best_jump_fn}, Proposition \ref{prop.Fstar}} -- the formula for $F^*$ is derived (Equation \ref{eqn:def_F_star_maintext}). This is the maximum achievable fidelity achievable in the 1G1B system, with fixed Bell-diagonal input state $\rho_{\mathrm{\scriptscriptstyle BD}}$. Therefore, at any given time $t$, the fidelity $F(t)$ of the stored link in the 1G1B system (see Definition \ref{def:F(t)}) satisfies $F(t)\leq F^*$.
 \item At $F=F^*$, Protocol 3 provides the best output fidelity,  $$ J_1\left(F^*,\rho_{\mathrm{\scriptscriptstyle BD}}\right) \leq J_2\left(F^*,\rho_{\mathrm{\scriptscriptstyle BD}}\right) \leq J_3\left(F^*,\rho_{\mathrm{\scriptscriptstyle BD}}\right)$$ (\textbf{Proposition \ref{prop:intersection_jump_functions} and Corollary \ref{cor:best_jump_fn}}). 
     \item At $F=1/4$, Protocol 1 provides the best output fidelity, i.e. $$J_3\left(F^*,\rho_{\mathrm{\scriptscriptstyle BD}}\right) \leq J_2\left(F^*,\rho_{\mathrm{\scriptscriptstyle BD}}\right) \leq J_1\left(F^*,\rho_{\mathrm{\scriptscriptstyle BD}}\right), $$ since  $J_i(1/4,\rho_{\mathrm{\scriptscriptstyle BD}}) = (F_{\mathrm{\scriptscriptstyle BD}}+\lambda_i)/2$.
    \item For $i=1,2,3$, $J_i(F,\rho_{\mathrm{\scriptscriptstyle BD}})$ is a concave function of $F$ (\textbf{Proposition \ref{prop:J_i_concave_increasing}}).
\end{enumerate}
In particular, the third result means that any straight line taken between two points on $J_i$ must lie below the curve itself. The linear lower bound is the linear function connecting the points 
\begin{equation}
    \left(F^*,J_{1}(F^*,\rho_{\mathrm{\scriptscriptstyle BD}})\right), \;\;\; \left(\frac{1}{4},J_{3}\left(\frac{1}{4},\rho_{\mathrm{\scriptscriptstyle BD}}\right)\right),
\label{eqn:points_to_interpolate_lower_bound}
\end{equation} which is given by $$J_{\mathrm{LB}}(F,\rho_{\mathrm{\scriptscriptstyle BD}}) = \left(\frac{J_{1}(F^*,\rho_{\mathrm{\scriptscriptstyle BD}})-\frac{F_{\mathrm{\scriptscriptstyle BD}}+\lambda_3}{2}}{F^*-\frac{1}{4}}\right)\left(F-\frac{1}{4}\right)+\frac{F_{\mathrm{\scriptscriptstyle BD}}+\lambda_3}{2},$$ where we have used the fact that $J_i(1/4,\rho_{\mathrm{\scriptscriptstyle BD}}) = (F_{\mathrm{\scriptscriptstyle BD}}+\lambda_i)/2$. Letting $\lambda_{\max}=\lambda_1 $ and $\lambda_{\min}=\lambda_3$, this may be rearranged into the form $$J_{\mathrm{LB}}(F,\rho_{\mathrm{new}}) = a_{\mathrm{l}}F+b_{\mathrm{l}}, $$ with $a_{\mathrm{l}}$ and $b_{\mathrm{l}}$ given in Lemma \ref{lem:bilocal_clifford_bounds} (in (\ref{eq.lower_bounds_ab})). When choosing the points in (\ref{eqn:points_to_interpolate_lower_bound}), we are joining the line corresponding to the lowest of the $J_i$ for both $F = 1/4$ and $F=F^*$.  By the concavity property, this is therefore a lower bound for all of the $J_i$ in the region $[1/4,F^*]$. See Figure \ref{fig:linear_bounds} for an illustration of this lower bound.

We now show the upper bound. We choose this to be the linear function connecting the points $(1/4, F_{\mathrm{\scriptscriptstyle BD}}) $ and $ (1,1)$, which is given by $$ J_{\mathrm{\scriptscriptstyle UB}}(F,\rho_{\mathrm{new}}) = \left( \frac{1-F_{\mathrm{\scriptscriptstyle BD}}}{1-\frac{1}{4}}\right)\left(F-\frac{1}{4}\right)+F_{\mathrm{\scriptscriptstyle BD}},$$
and may be rearranged into the form $$J_{\mathrm{UB}}(F,\rho_{\mathrm{new}}) = a_{\mathrm{u}}F+b_{\mathrm{u}}, $$
with $a_{\mathrm{u}}$ and $b_{\mathrm{u}}$ given in Lemma \ref{lem:bilocal_clifford_bounds} (in (\ref{eq.upper_bounds_ab})). We show that this is an upper bound with the following steps. Again, for ease of notation, we exploit the symmetry in $\lambda_i$ and assume that $\lambda_1 \geq \lambda_2 \geq \lambda_3$.
\begin{enumerate}
\item In the domain $F>0$, the jump functions $J_1$, $J_2$ and $J_3$ intersect at the same point $F_{\mathrm{int}}$. Moreover, for $i=1,2,3$, $J_i(F_{\mathrm{int}},\rho_{\mathrm{\scriptscriptstyle BD}}) = \sqrt{\frac{F_{\mathrm{\scriptscriptstyle BD}}}{2}} <F_{\mathrm{\scriptscriptstyle BD}}$. (\textbf{Proposition \ref{prop:intersection_jump_functions}}). 
\item In the domain $F\in [F_{\mathrm{int}},1]$, the jump function outputting the highest-fidelity outcome out of protocols 1-3 is $J_3$ (\textbf{Corollary \ref{cor:best_jump_fn}}). 
\item For $i=1,2,3$, $J_i$ is an increasing and concave function of $F$ (\textbf{Proposition \ref{prop:J_i_concave_increasing}}).
\item Consider the tangent to $J_3$ at $F=1$. This lies below $J_{\mathrm{UB}}$ in the range $F\in [1/4,1]$ (\textbf{Proposition \ref{prop:tangent_below}}).
\end{enumerate}
By result (3) from the above list (concavity), we see that the tangent to $J_3$ at $F=1$ upper bounds $J_3$ for all $F$. By result (2) from the above list, this also upper bounds $J_1$ and $J_2$ in the range $F\in [F_{\mathrm{int}},1]$. Therefore, by result (4), $J_{\mathrm{\scriptscriptstyle UB}}$ upper bounds $J_1$, $J_2$ and $J_3$ in the range $F\in [F_{\mathrm{int}},1]$. Moreover, for $F<F_{\mathrm{int}}$, by results (1) and (3), $J_i(F,\rho_{\mathrm{\scriptscriptstyle BD}})\leq \sqrt{\frac{F_{\mathrm{\scriptscriptstyle BD}}}{2}}<F_{\mathrm{\scriptscriptstyle BD}} \leq J_{\mathrm{\scriptscriptstyle UB}}(F,\rho_{\mathrm{\scriptscriptstyle BD}})$, by the definition of $J_{\mathrm{\scriptscriptstyle UB}}$ ($J_{\mathrm{\scriptscriptstyle UB}}$ runs through the point $(1/4,F_{\mathrm{\scriptscriptstyle BD}})$ and is increasing). This suffices to show that the upper bound holds.

Finally, we show the bounds for $p_i$. Recalling that $F_{\mathrm{\scriptscriptstyle BD}}+\lambda_1+\lambda_2+\lambda_3 = 1$, we have 
\begin{align*}
    \pdv F p_1(F,\rho_{\mathrm{\scriptscriptstyle BD}}) = \frac{2}{3}(1-\lambda_2-\lambda_3) &= \frac{2}{3}(2F_{\mathrm{\scriptscriptstyle BD}}+2\lambda_1 - 1) \\ &\geq \frac{2}{3}(2F_{\mathrm{\scriptscriptstyle BD}} - 1) >0.
\end{align*}
Therefore, $p_1(F,\rho_{\mathrm{\scriptscriptstyle BD}})$ is an increasing function of $F$. By symmetry, $p_2$ and $p_3$ are also increasing functions of $F$. Since the fidelity $F(t)$ of the 1G1B system always lies in the region $F(t)\in [1/4,F^*]$, it follows that at any point in time, the success probability $p$ of purification may be bounded with
\begin{equation*}
    p_i\left(\frac{1}{4},\rho_{\mathrm{\scriptscriptstyle BD}}\right) \leq p \leq p_i\left(F^*,\rho_{\mathrm{\scriptscriptstyle BD}}\right).
\end{equation*}
\end{proof}

Below are the collection of results that were used to show the bounds on the jump functions.
\begin{proposition}
\label{prop:intersection_jump_functions}
    In the domain $F>0$, jump functions 1-3 intersect exactly once at the same point $F_{\text{int}}$, such that $J_i(F_{\text{int}},\rho_{\mathrm{new}}) = \sqrt{\frac{F_{\mathrm{\scriptscriptstyle BD}}}{2}} < F_{\mathrm{\scriptscriptstyle BD}}$. 
\end{proposition}
\begin{proof}
    We firstly compute the intersection point of jump functions 1 and 2. This occurs at the $F$ value which satisfies
    \begin{align*}
        \frac{(4\lambda_1 + 3 \lambda_2 + 3 \lambda_3 - 3)F - \lambda_1}{(4\lambda_2+4\lambda_3-2)F - \lambda_2 - \lambda_3 - 1} = \frac{(3\lambda_1 + 4 \lambda_2 + 3 \lambda_3 -3)F - \lambda_2}{(4\lambda_1+4\lambda_3-2)F - \lambda_1 - \lambda_3 -1},
    \end{align*}
    or alternatively, recalling that $F_{\mathrm{\scriptscriptstyle BD}}+\lambda_1+\lambda_2+\lambda_3=1$,
    \begin{align*}
        \frac{(\lambda_1-3F_{\mathrm{\scriptscriptstyle BD}})F - \lambda_1}{(2-4F_{\mathrm{\scriptscriptstyle BD}}-4\lambda_1)F - 2 + F_{\mathrm{\scriptscriptstyle BD}}+\lambda_1} = (1\leftrightarrow 2),
    \end{align*}
    where to obtain the RHS we exchange labels 1 and 2 of the LHS. This is equivalent to
    \begin{align*}
        \left((\lambda_1-3F_{\mathrm{\scriptscriptstyle BD}})F - \lambda_1 \right) \left( (2-4F_{\mathrm{\scriptscriptstyle BD}}-4\lambda_2)F - 2 + F_{\mathrm{\scriptscriptstyle BD}}+\lambda_2 \right) - (1\leftrightarrow 2) = 0,
    \end{align*}
    which simplifies to 
    \begin{equation}
        (\lambda_1-\lambda_2)\left((2-16 F_{\mathrm{\scriptscriptstyle BD}})F^2+(8F_{\mathrm{\scriptscriptstyle BD}}-4)F + 2-F_{\mathrm{\scriptscriptstyle BD}} \right) = 0.
    \end{equation}
    Then, if $\lambda_1\neq \lambda_2$, the points of intersection depend only on $F_{\mathrm{\scriptscriptstyle BD}}$ and therefore are symmetric in $\lambda_1$, $\lambda_2$ and $\lambda_3$. The points of intersection are given by 
    \begin{equation}
        F = \frac{4F_{\mathrm{\scriptscriptstyle BD}}-2 \pm 3 \sqrt{2 F_{\mathrm{\scriptscriptstyle BD}}}}{2(8F_{\mathrm{\scriptscriptstyle BD}}-1)}
    \end{equation}
    and recalling that $F_{\mathrm{\scriptscriptstyle BD}}\in (1/2,1]$, the solution lying in the domain of interest ($F>0$) is 
    \begin{equation*}
        F_{\text{int}} = \frac{4F_{\mathrm{\scriptscriptstyle BD}}-2 + 3 \sqrt{2 F_{\mathrm{\scriptscriptstyle BD}}}}{2(8F_{\mathrm{\scriptscriptstyle BD}}-1)}.
    \end{equation*}
    Then, since $F_{\text{int}}$ is symmetric in $\lambda_1$, $\lambda_2$ and $\lambda_3$, all jump functions intersect at the point $F_{\text{int}}$. One may also show that $$J_i(F_{\text{int}},\rho_{\mathrm{new}}) = \sqrt{\frac{F_{\mathrm{\scriptscriptstyle BD}}}{2}}, $$ e.g. using software such as Mathematica. Since $F_{\mathrm{\scriptscriptstyle BD}}<1/2$, we have $$ \sqrt{\frac{1}{2}} < \sqrt{F_{\mathrm{\scriptscriptstyle BD}}} \Leftrightarrow \sqrt{\frac{F_{\mathrm{\scriptscriptstyle BD}}}{2}} < F_{\mathrm{\scriptscriptstyle BD}}. $$
\end{proof}
We now continue with the following corollary.
\begin{corollary}
    Suppose that $\lambda_1\geq \lambda_2\geq\lambda_3$. Then, for $F\geq F_{\text{int}}$,
    \begin{equation}  J_3(F,\rho_{\mathrm{\scriptscriptstyle BD}}) \geq J_2(F,\rho_{\mathrm{\scriptscriptstyle BD}}) \geq J_1(F,\rho_{\mathrm{\scriptscriptstyle BD}}),\label{eqn:jump_inequality}
    \end{equation}
    \label{cor:best_jump_fn}
\end{corollary}
\begin{proof}
    From Proposition \ref{prop:intersection_jump_functions}, $J_1$, $J_2$ and $J_3$ will not intersect again for $F>F_{\text{int}}$. Therefore, their ordering remains the same for all $F>F_{\text{int}}$. The jump function outputting the largest fidelity in this range will therefore also have the largest limit as $F\rightarrow \infty$. We see that 
    \begin{equation*}
    \lim_{F\rightarrow \infty} J_i(F,\rho_{\mathrm{\scriptscriptstyle BD}}) = \frac{ 3F_{\mathrm{\scriptscriptstyle BD}}-\lambda_i}{4F_{\mathrm{\scriptscriptstyle BD}}+\lambda_i-2},
    \end{equation*}
    which is a decreasing function of $\lambda_i$. Therefore, $\lambda_3 = \min \{ \lambda_1, \lambda_2, \lambda_3\}$ gives the largest limit, and $J_1$ satisfies (\ref{eqn:jump_inequality}). 
\end{proof}
From Proposition \ref{prop:intersection_jump_functions} and Corollary \ref{cor:best_jump_fn}, we know which of $J_1$, $J_2$ and $J_3$ provide the best fidelity for $F\in [1/2,1]$. With the following proposition, we see that for some lower fidelities, it is better to replace with the bad link rather than choose to pump.
\begin{proposition}\label{prop.Fstar}
The largest fidelity obtainable by pumping a low-fidelity Werner state with $\rho_{\mathrm{\scriptscriptstyle BD}}$ and bilocal Clifford protocols is
\begin{equation}
    F^* = \frac{2F_{\mathrm{\scriptscriptstyle BD}}-1+\sqrt{(2F_{\mathrm{\scriptscriptstyle BD}}-1)^2 + 2 \lambda_\mathrm{min}(2F_{\mathrm{\scriptscriptstyle BD}} -1 + 2 \lambda_\mathrm{min})}}{2(2F_{\mathrm{\scriptscriptstyle BD}}-1 + 2 \lambda_\mathrm{min})},
    \label{eqn:def_F_star}
\end{equation}
where $\lambda_\mathrm{min} = \min\{\lambda_1,\lambda_2,\lambda_3\}$. %Then, for all $F \in [\frac{1}{4},F^*]$, it holds that
%\begin{equation}
%    J_\mathrm{i}(\rho_\mathrm{good}, \rho_\mathrm{bad}) \leq F^*
%\end{equation}
%for all $i=1,...,7$. Moreover, the inequality is saturated when $F=F^*$ and $\lambda_i = \min \{\lambda_1,\lambda_2,\lambda_3 \}$. 
\end{proposition}
\begin{proof}
    Consider applying pumping protocol $i\in\{1,2,3\}$. This stops improving the Werner state fidelity at the value of $F$ such that 
    \begin{align*}
        & F^* = J_i(F^*,\rho_{\mathrm{\scriptscriptstyle BD}}) \\ &\Leftrightarrow F^* =  \frac{(\lambda_i-3F_{\mathrm{\scriptscriptstyle BD}})F^* - \lambda_i}{(2-4F_{\mathrm{\scriptscriptstyle BD}}-4\lambda_i)F^* - 2 + F_{\mathrm{\scriptscriptstyle BD}}+\lambda_i} \\ &\Leftrightarrow 0 = (2-4F_{\mathrm{\scriptscriptstyle BD}}-4\lambda_i)F^2 +(4F_{\mathrm{\scriptscriptstyle BD}}-2)F+\lambda_i,  \\ 
    \end{align*}
    which has solutions 
    \begin{equation*}
        F = \frac{2F_{\mathrm{\scriptscriptstyle BD}}-1 \pm \sqrt{(2F_{\mathrm{\scriptscriptstyle BD}}-1)^2+2\lambda_i(2F_{\mathrm{\scriptscriptstyle BD}}-1+2\lambda_i)}}{2(2F_{\mathrm{\scriptscriptstyle BD}}-1 +2\lambda_i)},
    \end{equation*}
    one of which is positive and one negative. Recalling that for $F>\frac{1}{2}$, the jump function taking the largest value is $J_i$ with $\lambda_i = \lambda_{\mathrm{min}}$, means that the maximum fidelity achievable is (\ref{eqn:def_F_star}).
\end{proof}
\begin{proposition}
\label{prop:J_i_concave_increasing}
    For any $\rho_{\mathrm{\scriptscriptstyle BD}}$ with $F_{\mathrm{\scriptscriptstyle BD}}>1/2$, , for $i=1,2,3$ $J_i(F,\rho_{\mathrm{\scriptscriptstyle BD}})$ is a strictly concave and increasing function of $F$. 
\end{proposition}
\begin{proof}
    We differentiate $J_i$. Firstly, consider derivatives of functions of the form $$y=\frac{ax+b}{cx+d}.$$ This may be rewritten as $$ y= \frac{a}{c} +\frac{b-\frac{ad}{c}}{cx+d}.$$ Then, 
    \begin{equation}
         \dv{y}{x} = \frac{ad-bc}{(cx+d)^2}, \;\;\;\;\; \dv[2]{y}{x} = -2c\frac{ad-bc}{(cx+d)^3}.
         \label{eqn:deriv_rational_function}
    \end{equation}
    To check the sign of these functions, we must therefore check the sign of $ad-bc$. Recalling that $J_i$ may be rewritten as 
    \begin{equation*}
        J_i(F,\rho_{\mathrm{\scriptscriptstyle BD}}) = \frac{(3F_{\mathrm{\scriptscriptstyle BD}}-\lambda_i)F + \lambda_i}{(4F_{\mathrm{\scriptscriptstyle BD}}+4\lambda_i-2)F +2 - F_{\mathrm{\scriptscriptstyle BD}}-\lambda_i},
    \end{equation*}
  in this case,
  \begin{align*}
      a &=  3F_{\mathrm{\scriptscriptstyle BD}}-\lambda_i >\frac{3}{2} - \frac{1}{2} = 1 \\
      b &= \lambda_i <\frac{1}{2} \\ 
      c &= 4(F_{\mathrm{\scriptscriptstyle BD}}+\lambda_i)-2 \leq 4\cdot 1-2 = 2 \\ 
      d &= 2 - (F_{\mathrm{\scriptscriptstyle BD}}+\lambda_i) \geq 2-1 = 1
  \end{align*}
  and it follows that $ad-bc > 1\cdot 1 - 2\cdot 1/2 = 0$. Then, since 
  \begin{equation*}
      c = 4F_{\mathrm{\scriptscriptstyle BD}}+4\lambda_i-2 >4\cdot \frac{1}{2}+4\lambda_i-2 = 4\lambda_i\geq 0,
  \end{equation*}
  it follows from (\ref{eqn:deriv_rational_function}) that 
  \begin{equation*}
      \pdv F J_i(F,\rho_{\mathrm{\scriptscriptstyle BD}}) >0, \;\;\;\;\; \pdv[2] F J_i(F,\rho_{\mathrm{\scriptscriptstyle BD}}) < 0.
  \end{equation*}
  Therefore, $J_i$ is a strictly concave and increasing function of $F$. 
\end{proof}

\begin{proposition}
     Suppose that $\lambda_1\geq \lambda_2\geq \lambda_3$. Consider the tangent to $J_3(F,\rho_{\mathrm{\scriptscriptstyle BD}})$ at $F=1$. Denote this by $J_{\mathrm{\scriptscriptstyle tan}}(F,\rho_{\mathrm{\scriptscriptstyle BD}})$. Then, this lies below $J_{\mathrm{\scriptscriptstyle UB}}$ for all $F\in [1/4,1]$, i.e. $$ J_{\mathrm{\scriptscriptstyle tan}}(F,\rho_{\mathrm{\scriptscriptstyle BD}}) \leq J_{\mathrm{\scriptscriptstyle UB}}(F,\rho_{\mathrm{\scriptscriptstyle BD}}), $$where  $$J_{\mathrm{\scriptscriptstyle UB}}(F,\rho_{\mathrm{\scriptscriptstyle BD}}) = \frac{4(1-F_{\mathrm{\scriptscriptstyle BD}})}{3} F+ \frac{4F_{\mathrm{\scriptscriptstyle BD}}-1}{3} $$ is the linear upper bound from Lemma \ref{lem:bilocal_clifford_bounds}.
     \label{prop:tangent_below}
\end{proposition}
\begin{proof}
    We firstly compute the formula for the tangent to $J_i$ at $F=1$. Recalling the formula  (\ref{eqn:deriv_rational_function}), this has gradient $$\pdv{F}J_3(F,\rho_{\mathrm{\scriptscriptstyle BD}}) \big|_{F = 1} = \frac{ad-bc}{(c+d)^2} = \frac{6F_{\mathrm{\scriptscriptstyle BD}}-3(F_\mathrm{\scriptscriptstyle BD}+\lambda_3)^2}{\left(3(F_{\mathrm{\scriptscriptstyle BD}}+\lambda_3)\right)^2} = \frac{2 F_{\mathrm{\scriptscriptstyle BD}}}{3(F_{\mathrm{\scriptscriptstyle BD}}+\lambda_3)^2}-\frac{1}{3}. $$
Since the tangent runs through the point $(1,J_3(1,\rho_{\mathrm{\scriptscriptstyle BD}}))$, it has formula 
\begin{equation*}
    J_{\mathrm{tan}}(F,\rho_{\mathrm{\scriptscriptstyle BD}}) = \left(\frac{2 F_{\mathrm{\scriptscriptstyle BD}}}{3(F_{\mathrm{\scriptscriptstyle BD}}+\lambda_3)^2}-\frac{1}{3}\right)(F-1)+\frac{F_{\mathrm{\scriptscriptstyle BD}}}{F_{\mathrm{\scriptscriptstyle BD}}+\lambda_3},
\end{equation*}
where we have used $J_i(1,\rho_{\mathrm{\scriptscriptstyle BD}}) = F_{\mathrm{\scriptscriptstyle BD}}/(F_{\mathrm{\scriptscriptstyle BD}}+\lambda_i)$. We note that at $F=1$, $$J_{\mathrm{\scriptscriptstyle UB}}(1,\rho_{\mathrm{\scriptscriptstyle BD}}) = 1 \geq  \frac{F_{\mathrm{\scriptscriptstyle BD}}}{F_{\mathrm{\scriptscriptstyle BD}}+\lambda_3} = J_{\mathrm{\scriptscriptstyle tan}}(1,\rho_{\mathrm{\scriptscriptstyle BD}}).$$ Therefore, to show the proposition, it suffices to show that 
\begin{equation}
    J_{\mathrm{\scriptscriptstyle UB}}\left(\frac{1}{4},\rho_{\mathrm{\scriptscriptstyle BD}}\right) \geq J_{\mathrm{\scriptscriptstyle tan}}\left(\frac{1}{4},\rho_{\mathrm{\scriptscriptstyle BD}}\right),
    \label{eqn:inequality_tangent_F=1/4}
\end{equation}
since both $J_{\mathrm{\scriptscriptstyle UB}}$ and $J_{\mathrm{\scriptscriptstyle tan}}$ are linear in $F$ and therefore intersect at most once. Now,
\begin{align*}
    J_{\mathrm{\scriptscriptstyle UB}}\left(\frac{1}{4},\rho_{\mathrm{\scriptscriptstyle BD}}\right) -  J_{\mathrm{\scriptscriptstyle tan}}\left(\frac{1}{4},\rho_{\mathrm{\scriptscriptstyle BD}}\right) &= F_{\mathrm{\scriptscriptstyle BD}} - \left(\frac{2 F_{\mathrm{\scriptscriptstyle BD}}}{3(F_{\mathrm{\scriptscriptstyle BD}}+\lambda_3)^2}-\frac{1}{3}\right)\left(-\frac{3}{4}\right)-\frac{F_{\mathrm{\scriptscriptstyle BD}}}{F_{\mathrm{\scriptscriptstyle BD}}+\lambda_3} \\ &= F_{\mathrm{\scriptscriptstyle BD}} - \frac{1}{4} +  \frac{F_{\mathrm{\scriptscriptstyle BD}}}{2(F_{\mathrm{\scriptscriptstyle BD}}+\lambda_3)^2} -\frac{F_{\mathrm{\scriptscriptstyle BD}}}{F_{\mathrm{\scriptscriptstyle BD}}+\lambda_3}.
\end{align*}
Now, let $x\coloneqq F_{\mathrm{\scriptscriptstyle BD}}+\lambda_3$, and $$h(x) \coloneqq F_{\mathrm{\scriptscriptstyle BD}} - \frac{1}{4} +  \frac{F_{\mathrm{\scriptscriptstyle BD}}}{2 x^2} -\frac{F_{\mathrm{\scriptscriptstyle BD}}}{x}.$$ By the assumption that $\lambda_3 = \min \{ \lambda_1,\lambda_2,\lambda_3\}$ and the condition $F_{{\mathrm{\scriptscriptstyle BD}}}+\lambda_1+\lambda_2+\lambda_3=1$, it follows that 
\begin{equation}
\lambda_3 \in \left[0, \frac{1-F_{\mathrm{\scriptscriptstyle BD}}}{3}\right], \;\;\;\; \text{and} \;\;\; x\in \left[F_{\mathrm{\scriptscriptstyle BD}},\frac{1+2F_{\mathrm{\scriptscriptstyle BD}}}{3}\right].
    \label{eqn:range_x}
\end{equation}
To prove the proposition, it therefore suffices to show positivity of $h$ for $x$ in the range (\ref{eqn:range_x}). We start by establishing the monotonicity of $h$: $$\pdv{x}h(x) = -\frac{F_{\mathrm{\scriptscriptstyle BD}}}{ x^3} +\frac{F_{\mathrm{\scriptscriptstyle BD}}}{x^2} = - \frac{F_{\mathrm{\scriptscriptstyle BD}}}{x^3}(1-x) \leq 0, $$ since $x = F_{\mathrm{\scriptscriptstyle BD}}+\lambda_3 \leq 1$. We therefore see that $h$ is decreasing. To show that $h$ is positive in the range (\ref{eqn:range_x}), it therefore suffices to show that $$ h\left( \frac{1+2F_{\mathrm{\scriptscriptstyle BD}}}{3}\right) \geq 0. $$ We have
\begin{align*}
    h\left( \frac{1+2F_{\mathrm{\scriptscriptstyle BD}}}{3}\right) &=  F_{\mathrm{\scriptscriptstyle BD}} - \frac{1}{4} +  \frac{9 F_{\mathrm{\scriptscriptstyle BD}}}{2 (1+2F_{\mathrm{\scriptscriptstyle BD}})^2} -\frac{3 F_{\mathrm{\scriptscriptstyle BD}}}{1+2F_{\mathrm{\scriptscriptstyle BD}}} \\ &= F_{\mathrm{\scriptscriptstyle BD}} - \frac{1}{4} +  6 F_{\mathrm{\scriptscriptstyle BD}} \left(\frac{\frac{3}{4} - \frac{1}{2}(1+2F_{\mathrm{\scriptscriptstyle BD}})}{(1+2F_{\mathrm{\scriptscriptstyle BD}})^2} \right)  \\ &= F_{\mathrm{\scriptscriptstyle BD}} - \frac{1}{4} +  6 F_{\mathrm{\scriptscriptstyle BD}} \frac{\frac{1}{4} - F_{\mathrm{\scriptscriptstyle BD}}}{(1+2F_{\mathrm{\scriptscriptstyle BD}})^2} \\ &= \left(F_{\mathrm{\scriptscriptstyle BD}} - \frac{1}{4}\right)\left(1- \frac{6 F_{\mathrm{\scriptscriptstyle BD}}}{(1+2F_{\mathrm{\scriptscriptstyle BD}})^2} \right).
\end{align*}
Then, since $F_{\mathrm{\scriptscriptstyle BD}}>\frac{1}{2}$, we have 
\begin{align*}
    h\left( \frac{1+2F_{\mathrm{\scriptscriptstyle BD}}}{3}\right) >0 &\Leftrightarrow 1- \frac{6 F_{\mathrm{\scriptscriptstyle BD}}}{(1+2F_{\mathrm{\scriptscriptstyle BD}})^2} >0 \\ &\Leftrightarrow (1+2F_{\mathrm{\scriptscriptstyle BD}})^2 > 6 F_{\mathrm{\scriptscriptstyle BD}} \\ &\Leftrightarrow 4 F_{\mathrm{\scriptscriptstyle BD}}^2 -2 F_{\mathrm{\scriptscriptstyle BD}} +1 >0 \\ &\Leftrightarrow (1-2F_{\mathrm{\scriptscriptstyle BD}})^2 + 2 F_{\mathrm{\scriptscriptstyle BD}} >0,
\end{align*}
which holds. We therefore see that $$ h(x) \geq h\left( \frac{1+2F_{\mathrm{\scriptscriptstyle BD}}}{3}\right)>0 $$ for all $x$ in the range (\ref{eqn:range_x}), and therefore $$J_{\mathrm{\scriptscriptstyle UB}}\left(\frac{1}{4},\rho_{\mathrm{\scriptscriptstyle BD}}\right) -  J_{\mathrm{\scriptscriptstyle tan}}\left(\frac{1}{4},\rho_{\mathrm{\scriptscriptstyle BD}}\right)>0.$$ This suffices to show the proposition.
\end{proof}

\subsection{Additional proofs}\label{subsec.additional_proofs}
\begin{lemma}\label{lemma.bell_diagonal_entangled}
    A Bell-diagonal state 
    \begin{equation*}
        \rho = \lambda_0 \ketbra{\phi^+} + \lambda_1 \ketbra{\psi^+} + \lambda_2 \ketbra{\psi^-} + \lambda_3 \ketbra{\phi^-},
    \end{equation*}
    with $\lambda_0 + \lambda_1 + \lambda_2 + \lambda_3 = 1$,
    is entangled if and only if $\lambda_i>1/2$ for some $i$.
\end{lemma}
\begin{proof}[Proof of Lemma \ref{lemma.bell_diagonal_entangled}]
    We will analyse the entanglement of a Bell-diagonal state using the Peres-Horodecki criterion, which states that a bipartite, \reb{ $2\times 2$ dimensional} quantum state $\rho$ is entangled if and only if the partial transpose of $\rho$ has at least one negative eigenvalue \cite{Peres1996,Horodecki1996}.
    A Bell-diagonal state can be written in the Bell basis as
    \begin{equation*}
        \rho = \lambda_0 \ketbra{\phi^+} + \lambda_1 \ketbra{\psi^+} + \lambda_2 \ketbra{\psi^-} + \lambda_3 \ketbra{\phi^-}.
    \end{equation*}
    In the computational basis, $\{\ket{00}, \ket{01}, \ket{10}, \ket{11}\}$, the Bell-diagonal state can be written as
    \begin{equation*}
        \rho = \begin{pmatrix}
        \lambda_0+\lambda_3 & 0 & 0 & \lambda_0-\lambda_3\\
        0 & \lambda_1+\lambda_2 & \lambda_1-\lambda_2 & 0\\
        0 & \lambda_1-\lambda_2 & \lambda_1+\lambda_2 & 0\\
        \lambda_0-\lambda_3 & 0 & 0 & \lambda_0+\lambda_3
        \end{pmatrix}.
    \end{equation*}
    The partial transpose of this density matrix is given by
    \begin{equation*}
        \rho^\mathrm{PT} = \begin{pmatrix}
        \lambda_0+\lambda_3 & 0 & 0 & \lambda_1-\lambda_2\\
        0 & \lambda_1+\lambda_2 & \lambda_0-\lambda_3 & 0\\
        0 & \lambda_0-\lambda_3 & \lambda_1+\lambda_2 & 0\\
        \lambda_1-\lambda_2 & 0 & 0 & \lambda_0+\lambda_3
        \end{pmatrix}.
    \end{equation*}
    The eigenvalues of the partial transpose are $\xi_i = 1-2\lambda_i$, $i=1,2,3,4$.
    One of the eigenvalues is negative iff $\lambda_i>1/2$ for some $i$.
    Therefore, according to the Peres-Horodecki criterion, the state is entangled iff $\lambda_i>1/2$ for some $i$.
    Since these $\lambda_i$ correspond to the fidelity of $\rho$ to one of the Bell states (e.g., $F(\rho,\ket{\phi^+}) \equiv \bra{\phi^+}\rho\ket{\phi^+} = \lambda_0$), we conclude that the state is entangled iff the fidelity to one of the Bell states is larger than $1/2$.
\end{proof}

\clearpage
\reb{
\section{Numerical simulations}
In our analytical calculations, we assumed a purification protocol with constant success probability (which implies a linear jump function, as shown in Appendix \ref{app:jump_function_general_form}).
This allowed us to derive bounds for the performance of any 1G1B entanglement buffering system that uses bilocal Clifford protocols.
However, the success probability of these purification protocols is in general linear in the fidelity of the buffered state (see Appendix \ref{app:jump_function_general_form}).
In this Appendix, we compare the analytical bounds, which assume a constant success probability, to the actual values obtained via a simulation that considers the true (linear, non-constant) success probability.}

\reb{Our discrete-event simulation keeps track of the buffered link, which decoheres until an event is triggered. These events could correspond to a consumption request (which consumes the buffered memory) or a successful entanglement generation (which is followed by pumping, with probability $q$).
When purification is performed, it succeeds with a probability that depends linearly on the fidelity of the buffered link (see Appendix \ref{app:jump_function_general_form}).}

\reb{To compute the average consumed fidelity and the availability, we run the simulation $N_\mathrm{samples}$ times.
In each realization $i$ of the process, we let the system evolve over $t_\mathrm{sim}$ units of time until convergence to a steady state, and we record the fidelity of the buffered link $F_i(t_\mathrm{sim})$ (if the memory is empty, the fidelity is set to zero, as was specified in Definition~\ref{def:F(t)}).
Then, we estimate the average consumed fidelity as the average fidelity of the buffered link at $t_\mathrm{sim}$ (conditional on the buffered link being present):
\begin{equation}
    \overline F \approx \overline F' \equiv  \frac{\sum_{i=1}^{N_\mathrm{samples}} F_i(t_\mathrm{sim})}{N_\mathrm{samples}'},
\end{equation}
where
\begin{equation}
    N_\mathrm{samples}' = \sum_{i=1}^{N_\mathrm{samples}} \mathds{1}_{F_i(t_\mathrm{sim}) > 0}
\end{equation}
is the number of samples in which $F_i(t_\mathrm{sim}) > 0$ ($\mathds{1}$ is the indicator function). We measure the error in the estimate using the standard error:
\begin{equation}
    \varepsilon_F = \sqrt{\frac{ \sum_{i=1}^{N_\mathrm{samples}'} \left(F_i(t_\mathrm{sim}) - \overline{F}' \right)^2  }{N_\mathrm{samples}' \left(N_\mathrm{samples}'-1\right)} },
\end{equation}
which corresponds to the square root of the unbiased sample variance divided by the number of samples.
The availability is estimated as the proportion of samples in which there is a buffered link at~time~$t_\mathrm{sim}$:
\begin{equation}
    A \approx A' \equiv \frac{1}{N_\mathrm{samples}} \sum_{i=0}^{N_\mathrm{samples}}\mathds{1}_{F_i(t_\mathrm{sim}) > 0},
\end{equation}
Note that $A'$ is the average of a binary random variable. We can therefore model this random variable as Bernoulli-distributed with probability of success $A'$. This yields a variance $A'(1-A')$, which allows us to compute the standard error as
%and therefore the standard error is computed differently \bd{this requires some more explanation (I am not familiar with this result)}:
\begin{equation}
    \varepsilon_A = \sqrt{ \frac{A'(1-A')}{N_\mathrm{samples}} }.
\end{equation}}

\reb{Next, we study again the example from Figure \ref{fig.operating_regimes}, and we compare the bounds discussed in the main text with the results from our simulation.
In Figure \ref{fig.operating_regimes-sim}, we show the same lower and upper bounds (yellow and dark blue lines, respectively) from Figure \ref{fig.operating_regimes}.
We simulated three buffering systems, each of them using the unique bilocal Clifford protocols 1, 2, and 3 from Table \ref{tab:7_protocols} (we neglect protocols 4-7 since they are trivial).
We emphasise that these simulations consider the true probabilities of success (which are linear but non-constant in the fidelity of the buffered link) and the true jump functions (rational in the fidelity of the buffered link) of the purification protocols.
Figure \ref{fig.operating_regimes-sim} shows the availability and average consumed fidelity attained by each of these systems, for different values of $q$.
We first note that protocols 2 (blue circles) and 3 (red crosses) are equivalent. This is due to the symmetry of the newly generated state considered in this example, $\rho_\mathrm{new} = F_\mathrm{new} \ketbra{\phi^+} + (1-F_\mathrm{new}) \left(\ketbra{\psi^+} +  \ketbra{\psi^-}\right)/2$.
More importantly, the performance of the simulated systems lies within the analytical bounds, which were derived assuming a constant probability of success. This serves as empirical evidence that our simplified model is still useful when lifting the assumption about a constant probability of success, and can guide the design of more complex and realistic buffering systems.}

\begin{figure}[h]
  \centering
  \includegraphics[width=\textwidth]{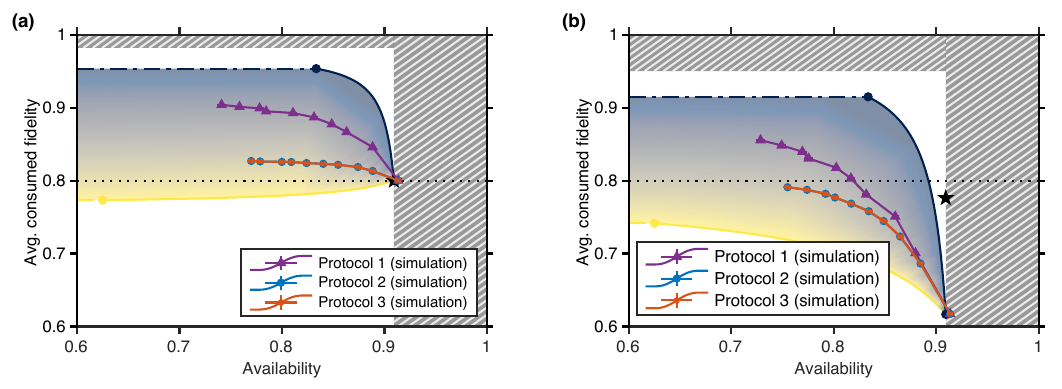}
  \caption{\reb{Bounds derived assuming a constant probability of success still apply when the assumption is lifted.
  \textbf{(a)} Noiseless memories ($\Gamma = 0$) or \textbf{(b)} noisy memories ($\Gamma = 5\cdot10^{-2}$ a.u.).
    For a given target availability, the average consumed fidelity is within the blue/yellow region (see Corollary \ref{corollary.bounds_avgconsfid}).
Availability is maximized for $q=0$ ($q$ is the probability of purification after successful entanglement generation), and it decreases for increasing $q$.
  White regions cannot be achieved by bilocal Clifford protocols.
  Striped regions cannot be achieved by any pumping protocol.
  Black star: performance of the replacement protocol (buffered link is replaced by new links).
  Dotted line: fidelity of newly generated entangled links.
  Solid lines with markers: performance of the 1G1B system obtained via simulation, using the true jump functions and true probabilities of success of purification protocols 1, 2, and 3 from Table \ref{tab:7_protocols}, ($q=0$ for the rightmost data point, decreasing in intervals of 0.111 until reaching $q=1$ in the leftmost data point). The simulation considers a linear probability of success, unlike the analytical calculations, in which this probability is assumed to be constant.
    Parameters used in this example (times and rates in the same arbitrary units): $\lambda = 1$, $\mu = 0.1$, $F_\mathrm{new}=0.8$, $\rho_\mathrm{new} = F_\mathrm{new} \ketbra{\phi^+} + (1-F_\mathrm{new}) \left(\ketbra{\psi^+} +  \ketbra{\psi^-}\right)/2$.
    Numerical parameters used in the simulation: $t_\mathrm{sim}=50$, $N_\mathrm{samples}=10^4$.}
  }
  \label{fig.operating_regimes-sim}
\end{figure}

\end{document}